\documentclass[11pt]{article}
\usepackage[ruled,vlined,linesnumbered,hangingcomment]{algorithm2e}
\usepackage[colorlinks=true,linkcolor=blue,citecolor=red,urlcolor=blue]{hyperref}
\usepackage{amsfonts,amssymb}
\usepackage{comment}
 
\usepackage{titlesec}
\usepackage{tikz}
\usepackage{bbding}
\usepackage{pifont}
\usepackage{titletoc}
\usepackage{hhline}
\usepackage{amsmath}
\usepackage{listings}
\usepackage{verbatim}
\usepackage{makecell}
\usepackage{tabu}
\usepackage{graphicx}
\usepackage{enumitem}
\usepackage{mathrsfs}
\usepackage{setspace}
\usepackage{bbm}
\usepackage{mathtools}
\usepackage{titletoc}
\usepackage{geometry} 
\usepackage{threeparttable}
\usepackage{url} 

\usepackage{listings}
\usepackage{amssymb,amsmath,amsfonts,latexsym}
\usepackage{amsmath,graphicx,bm,xcolor,url}
\usepackage[caption=false]{subfig} 
\usepackage{array}
\usepackage{verbatim}
\usepackage{bm}
\usepackage{verbatim}
\usepackage{textcomp}
\usepackage{mathrsfs}
\usepackage{relsize}
\usepackage{subfig}
 \usepackage{amsthm}

\catcode`~=11 \def\UrlSpecials{\do\~{\kern -.15em\lower .7ex\hbox{~}\kern .04em}} \catcode`~=13 

\allowdisplaybreaks[3]

\newcommand{\nn}{\nonumber}

\newcommand{\calA}{\mathcal{A}}

\newcommand{\calC}{\mathcal{C}}
\newcommand{\calD}{\mathcal{D}}

\newcommand{\calH}{\mathcal{H}}

\newcommand{\calK}{\mathcal{K}}
\newcommand{\calL}{\mathcal{L}}

\newcommand{\calN}{\mathcal{N}}

\newcommand{\calP}{\mathcal{P}}
\newcommand{\calQ}{\mathcal{Q}}

\newcommand{\calU}{\mathcal{U}}
\newcommand{\calV}{\mathcal{V}}
\newcommand{\calW}{\mathcal{W}}
\newcommand{\calX}{\mathcal{X}}

\newcommand{\ba}{\mathbf{a}}
\newcommand{\bA}{\mathbf{A}}

\newcommand{\bc}{\mathbf{c}}

\newcommand{\be}{\mathbf{e}}

\newcommand{\bg}{\mathbf{g}}

\newcommand{\bh}{\mathbf{h}}

\newcommand{\bI}{\mathbf{I}}

\newcommand{\bM}{\mathbf{M}}
\newcommand{\bn}{\mathbf{n}}

\newcommand{\bp}{\mathbf{p}}

\newcommand{\bq}{\mathbf{q}}

\newcommand{\bR}{\mathbf{R}}
\newcommand{\bs}{\mathbf{s}}

\newcommand{\bu}{\mathbf{u}}
\newcommand{\bU}{\mathbf{U}}
\newcommand{\bv}{\mathbf{v}}

\newcommand{\bw}{\mathbf{w}}

\newcommand{\bx}{\mathbf{x}}
\newcommand{\bX}{\mathbf{X}}
\newcommand{\by}{\mathbf{y}}

\newcommand{\bbeta}{\bm{\beta}}

\newcommand{\btau}{\bm{\tau}}

\DeclareMathOperator{\rank}{rank}

\newtheorem{theorem}{Theorem}

\newtheorem{assumption}{Assumption}

\newtheorem{definition}{Definition} 

\newcommand{\qednew}{\nobreak \ifvmode \relax \else
      \ifdim\lastskip<1.5em \hskip-\lastskip
      \hskip1.5em plus0em minus0.5em \fi \nobreak
      \vrule height0.75em width0.5em depth0.25em\fi}

\newcommand{\scrH}{\mathscr{H}}

\newcommand{\scrN}{\mathscr{N}}

\newcommand{\scrU}{\mathscr{U}}

\usepackage[T1]{fontenc}
\usepackage[utf8]{inputenc}
 
\usepackage{amsthm}
\usepackage{xcolor, tikz}
 
\newtheorem{pro}{Proposition}

\newtheorem{coro}{Corollary}
\newtheorem{lem}{Lemma}

\newtheorem{rem}{Remark}

\usepackage{bm}

\DeclareMathOperator{\sign}{sign}

\DeclareMathOperator{\rad}{rad}

\DeclareMathOperator{\sfP}{\mathsf{P}}

\usepackage{authblk}

\newif\ifrevision
\revisionfalse   
\ifrevision
  \newcommand{\rev}[1]{\textcolor{teal!80!black}{#1}}
\else
  \newcommand{\rev}[1]{#1}
\fi

\newif\ifjrnotes
\jrnotestrue   

\ifjrnotes
  \newcommand{\jrnote}[1]{%
    {\color{red}\footnotesize\textsf{[JR: #1]}}%
  }
\else
  \newcommand{\jrnote}[1]{}
\fi

\title{Optimal Quantized Compressed Sensing via Projected Gradient Descent}

\date{(\today)}

\author[$\dag$]{Junren Chen}
\author[$\dag$]{Ming Yuan} 
\affil[$\dag$]{Department of Statistics, Columbia University}
\begin{document}
\maketitle

\long\def\symbolfootnote[#1]#2{\begingroup\def\thefootnote{\fnsymbol{footnote}}\footnote[#1]{#2}\endgroup}

\symbolfootnote[0]{This research was supported in part by NSF Grants DMS-2015285 and DMS-2052955. Emails: \texttt{jc6315@columbia.edu} (JC); \texttt{ming.yuan@columbia.edu} (MY).}

\begin{abstract}
This paper provides a unified treatment to  the recovery of structured signals living in a star-shaped set from general quantized measurements $\calQ(\bA\bx-\btau)$, where $\bA$ is a sensing matrix, $\btau$ is a vector of (possibly random) quantization thresholds, and $\calQ$ denotes an $L$-level quantizer. The ideal estimator of hamming distance minimization (HDM) is optimal  but typically infeasible to compute. Under sub-Gaussian design, we study the projected gradient descent (PGD) algorithm with respect to the one-sided $\ell_1$-loss and identify the conditions under which PGD achieves the same error rate as HDM, up to logarithmic factors. These conditions include estimates of the separation probability, small-ball probability and some moment bounds that are easy to validate.
For \rev{the} multi-bit case,  we also develop a complementary approach based on product embedding to show global convergence. When applied to popular  models such as 1-bit compressed sensing with Gaussian $\bA$ and zero $\btau$ and the dithered 1-bit/multi-bit models with sub-Gaussian $\bA$ and uniform dither $\btau$, our unified treatment yields error rates that improve on or match the sharpest results in all instances. Particularly, PGD achieves the information-theoretic optimal rate $\widetilde{O}(\frac{k}{mL})$ for recovering $k$-sparse signals, and the rate $\widetilde{O}((\frac{k}{mL})^{1/3})$ for effectively sparse signals. For 1-bit compressed sensing of sparse signals, our result recovers the optimality of normalized binary iterative hard thresholding (NBIHT) that was proved very recently.  
\end{abstract}
\bigskip


\vspace{3mm}


\newpage 
 
\setcounter{tocdepth}{2}

{\hypersetup{linkcolor=black}
\small
\setstretch{0.95}
\setlength{\parskip}{0pt}
\thispagestyle{empty}
\tableofcontents
\thispagestyle{empty}
}

\newpage
\section{Introduction}\label{sec:intro}
Consider an $L$-level quantizer $\calQ $  which quantizes $q\in \mathbb{R}$ to
\begin{align}\label{eq:Llevel}
    \calQ(q) = \begin{cases}
        q_1,~~&\text{if }~q< b_1\\
        q_2,~~&\text{if }~b_1\le q< b_2\\
        \cdots\\
        q_{L-1},~~&\text{if }~b_{L-2}\le q< b_{L-1}\\
        q_L,~~&\text{if }~q\ge b_{L-1}
    \end{cases}
\end{align}
for some quantization thresholds $b_1<b_2<\cdots<b_{L-1}$ and certain values $q_1<q_2<\cdots<q_L$. If $L\ge 3$, we define 
\begin{align}
    \Delta: = \min_{j=1,\cdots,L-2}|b_{j+1}-b_j| \label{defreso}
\end{align}
as the resolution of the quantizer $\calQ$. For the 1-bit quantizer that outputs either $1$ or $-1$, we let $\Delta = 2$ for technical convenience.\footnote{For (almost) bounded inputs, the 1-bit quantizer can be well approximated by the (scaled) uniform quantizer $\calQ_\delta$ in (\ref{eq:Qdelta}) with large enough $\delta$; see \cite[Fig. 2]{chen2023parameter} for instance. It is thus reasonable to define the resolution of 1-bit quantizer as {\it some large enough constant}. Here, setting $\Delta=2$ is simply a convenient option for our further development.}
Note that  $(q_i)_{i=1}^L$ are only used to indicate which quantization cell the unquantized input falls in, hence their exact values are not important; 
we can thus assume without loss of generality that\footnote{\rev{In particular, this can be ensured by a relabeling of the output quantized levels, which leaves the information in the recovery problem unchanged.}}  
\begin{align}\label{wolg}
    q_{i+1} = q_i+\Delta,\quad i = 1,2,\cdots,L-1. 
\end{align}
See Figure \ref{fig:multi_one_bit_quantize} for an illustration.
\begin{figure}[ht!]
	\begin{centering}
		 \includegraphics[width=0.8\columnwidth]{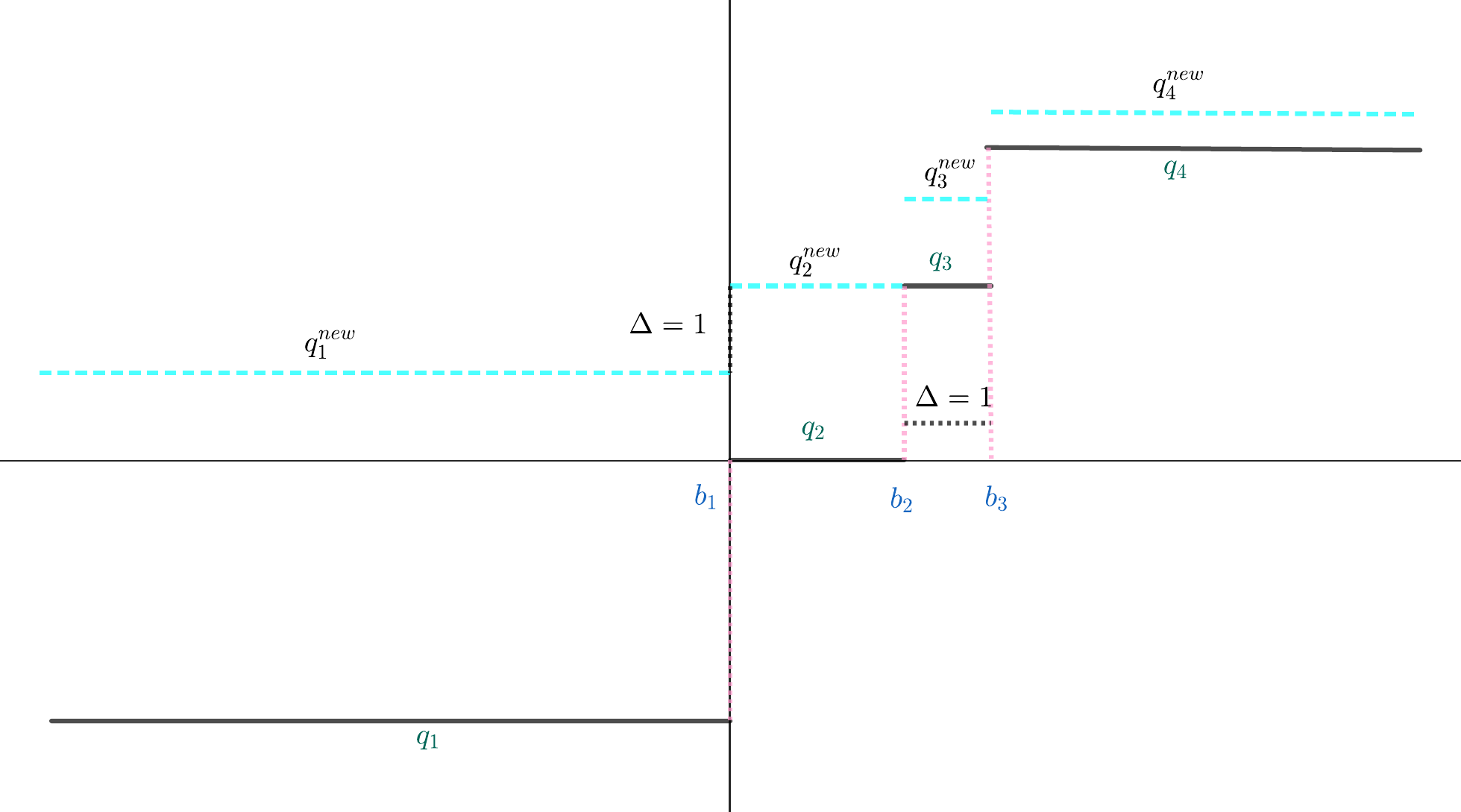}  
		\par
	\end{centering}
	
	\caption{\label{fig:multi_one_bit_quantize}\small An example of $(L=4)$-level quantizer. Note that we can assume (\ref{wolg}) because, for instance in the right figure, we can always modify $\{q_i\}_{i=1}^4$ to $\{q_i^{new}\}_{i=1}^4.$ Similarly, in 1-bit case where $\Delta=2$, we can always assume $(q_1,q_2)=(-1,1)$ to fulfill (\ref{wolg}).} 
\end{figure}

Given $m$ sensing vectors $\{\ba_i\}_{i=1}^m$ from $\mathbb{R}^n$, a set   $\calK $ in $\mathbb{R}^n$ that captures signal structure, the goal of   quantized compressed sensing   is to recover $\bx$ from
\begin{align}\label{eq:qcsmea}
    y_i = \calQ(\ba_i^\top \bx -\tau_i ),\qquad i=1,2,\cdots,m,
\end{align}
where $\tau_i$ denotes either \rev{a fixed quantization threshold or a randomly drawn} {\it dither}. For the latter case, as with \cite{dirksen2021non,jung2021quantized,thrampoulidis2020generalized,xu2020quantized}, we focus on uniform dithers  $\{\tau_i\}_{i=1}^m$ that are i.i.d. uniformly distributed over $[-\Lambda,\Lambda]$ and independent of everything else. We shall refer to $\Lambda$ as the dithering level and simply set $\Lambda=0$ to return setting without dithering.
In the vector form, we can write
\begin{align*}
    \by = \calQ(\bA\bx -\btau),
\end{align*}
where $\bA = [\ba_1,\cdots,\ba_m]^\top\in \mathbb{R}^{m\times n}$, $\btau= [\tau_1,\cdots,\tau_m]^\top$, and the quantizer is applied in an elementwise fashion.

On the signal structure,
we shall assume that $\calK$ is star-shaped set satisfying
\begin{align*}
    t \calK\subset\calK,\quad\forall t\in [0,1]. 
\end{align*} 
This general assumption was adopted in prior works such as \cite{plan2016generalized,plan2017high}, covering many interesting examples, including  sparse vectors in
\begin{align*} 
    \calK = \Sigma^n_k = \{\bu\in \mathbb{R}^n:\bu\text{ is }k\text{-sparse}\},
\end{align*}
 low-rank matrices in
\begin{align*} 
    \calK = M_{\bar{r}}^{n_1,n_2}= \{\bM\in \mathbb{R}^{n_1\times n_2}:\rank(\bM)\le \bar{r}\},
\end{align*}
and effectively sparse signals in
\begin{align*}
    \calK= \sqrt{k}\mathbb{B}_1^n:=\{\bu\in \mathbb{R}^n:\|\bu\|_1\le \sqrt{k}\},
\end{align*}
 just to name a few. We additionally use 
$$\mathbbm{A}^\beta_{\alpha} =\{\bu \in \mathbb{R}^n:\alpha\le\|\bu\|_2\le\beta\}$$
  with $\beta\ge\alpha\ge 0$
 to capture the norm constraint on the signals. \rev{Note that $\mathbbm{A}_0^\beta= \beta\mathbb{B}_2^n$ reduces to the $\ell_2$ ball of radius $\beta$.} Reserving $\calX$ for the signal space, we are interested in the recovery of signals $\bx$ in $$ \calX:=\calK\cap\mathbbm{A}^\beta_{\alpha}$$ from $\by=\calQ(\bA\bx-\btau)$.

Our general framework accommodates some of the most common quantized compressed sensing models as special cases.
\paragraph{1-bit compressed sensing (1bCS).} The most common and well-understood model is the so-called 1-bit compressed sensing (e.g., \cite{jacques2013robust,plan2012robust,plan2013one,plan2014dimension,boufounos20081,matsumoto2024binary,friedlander2021nbiht,awasthi2016learning,plan2016generalized,plan2017high,genzel2023unified}) where we seek to recover $\bx\in\calK$ from 
    \begin{align}\label{eq:1bcs}
    \by = \sign(\bA\bx).
\end{align}
In this case, we have $\Delta = 2$ (by our convention), $\Lambda =0$, and (\ref{wolg}) is satisfied. Note that the observations provide no information on the signal norm $\|\bx\|_2$, thus we set $\alpha=\beta=1$ and hence $\mathbbm{A}^\beta_{\alpha}=\mathbb{S}^{n-1}:=\{\bu\in\mathbb{R}^n:\|\bu\|_2=1\}$ restricts our attention to the signals in $\calX:=\calK\cap \mathbb{S}^{n-1}$.  
It is well known that accurate reconstruction can be achieved by efficient algorithms when $\bA$ is an i.i.d. Gaussian ensemble (e.g., \cite{plan2012robust,plan2016generalized,plan2013one}). These results, however, cannot be extended to sub-gaussian measurements without additional assumptions. See, e.g., \cite{ai2014one,goldstein2019non}.  


\paragraph{Dithered 1-bit compressed sensing (D1bCS).} Suppose that $\btau$ is uniformly distributed over $[-\lambda,\lambda]^m$ for some $\lambda>0$, we   seek to recover $\bx$ from measurements
\begin{align}\label{eq:d1bcs}
	\by = \sign(\bA\bx-\btau).
\end{align}
In this case, we have $\Delta =2$ (by convention), $\Lambda = \lambda$, and (\ref{wolg}) holds.
It was also commonly assumed that the signals have $\ell_2$-norms bounded by $R$ for some $R>0$. See, e.g., \cite{thrampoulidis2020generalized,dirksen2021non,jung2021quantized,dirksen2020one,dirksen2023robust}. Without loss of generality, we shall set $(\alpha,\beta)=(0,1)$ so that $\mathbbm{A}_{\alpha}^\beta = \mathbb{B}_2^n:=\{\bu\in\mathbb{R}^n:\|\bu\|_2\le 1\}$ and consider signals in $\calX:=\calK\cap \mathbb{B}_2^n$. Accurate reconstruction with norm can be achieved under sub-Gaussian $\bA$ (as precised by Assumption \ref{assum1}). See, e.g., \cite{thrampoulidis2020generalized,dirksen2021non,jung2021quantized}. 

\paragraph{Dithered multi-bit compressed sensing (DMbCS).} More generally, we can consider a resolution-$\delta$ uniform  quantizer  
    \begin{align}\label{eq:Qdelta}
        \calQ_\delta(a) = \delta \Big(\Big\lfloor \frac{a}{\delta}\Big\rfloor + \frac{1}{2}\Big).
    \end{align}
Under the uniform dither $\btau\sim\scrU([-\frac{\delta}{2},\frac{\delta}{2}]^m)$ and sub-Gaussian $\bA$, we can accurately recover structured signals from  
 \begin{align}\label{eq:dqcs}
 	\by = \calQ_\delta(\bA\bx -\btau).
 \end{align}
  See, e.g., \cite{thrampoulidis2020generalized,genzel2023unified,xu2020quantized}. While $\calQ_\delta(a)$ does not immediately sample {\it finite bits} from $a$ (i.e.,  the range of $\calQ_{\delta}$ is the infinite set $\{(l-\frac{1}{2})\delta:l\in\mathbb{Z}\}$), it is often more practical to use a  uniform quantizer with saturation in applications.
Specifically, for some even integer $L\ge 4$, we  consider the less informative   version of  $\calQ_\delta$ defined as  
   \begin{align}\label{eq:QdeltaL}
   	\calQ_{\delta,L}(a) = \calQ_\delta(a)\cdot \mathbbm{1}\Big(|a|< \frac{L\delta}{2}\Big) + \frac{(L-1)\delta}{2}  \mathbbm{1}\Big(a\ge\frac{L\delta}{2}\Big)+ \frac{(1-L)\delta}{2}  \mathbbm{1}\Big(a\le\frac{-L\delta}{2}\Big).
   \end{align}
   See Figure \ref{fig:saturation} for an intuitive illustration of $\calQ_\delta$ and $\calQ_{\delta,4}$. 
    Again we set $\mathbbm{A}_{\alpha}^\beta=  \mathbb{B}_2^n$. In this case, we have $\Delta=\delta$, $\Lambda=\frac{\delta}{2}$, and (\ref{wolg}) is satisfied.
 Our goal is to reconstruct $\bx\in \calX:=\calK\cap\mathbb{B}_2^n$ from 
    \begin{align}\label{eq:mbsatu}
	        \by = \calQ_{\delta,L}(\bA\bx -\btau).
	    \end{align}

\begin{figure}[ht!]
	\begin{centering}
		 \includegraphics[width=0.5\columnwidth]{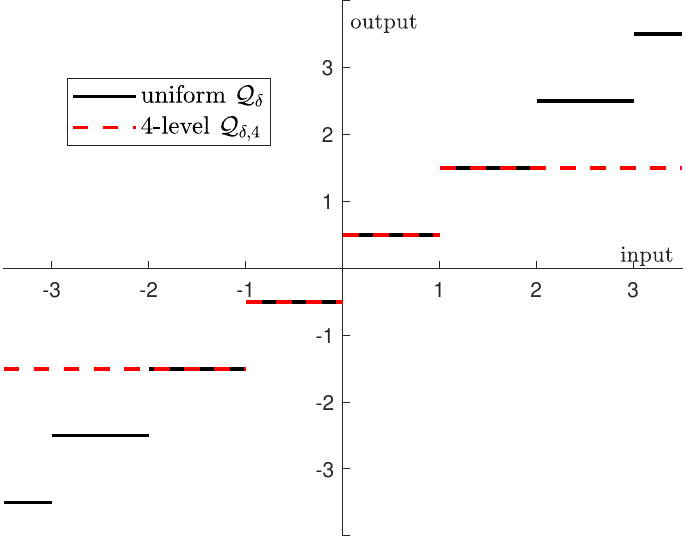}  
		\par
	\end{centering}
	
	\caption{\label{fig:saturation}\small The uniform quantizer $\calQ_\delta$ and its saturated version $\calQ_{\delta,4}$  under the resolution $\delta=1$.}
\end{figure}

\rev{We shall focus on the noiseless case of these quantized compressed sensing problems, with a brief discussion on the noisy setting provided in Section \ref{sec:discuss}.}
For recovering $k$-sparse signals from $\calQ(\bA\bx-\btau)$ with any fixed $(\bA,\btau)\in \mathbb{R}^{m\times n}\times \mathbb{R}^{m\times 1}$, the information-theoretic lower bound is known to be $\Omega(\frac{k}{mL})$.\footnote{\rev{Here and hereafter, we use the standard complexity notation that will be formally introduced in Section \ref{sec:notation}.}}  On the other hand, an ideal decoding strategy is to return an estimate having   quantized measurements coincident with the observations, which can be formulated as constrained hamming distance minimization (HDM).    This ideal decoder attains the optimal rate $\widetilde{O}(\frac{k}{m})$  for   1bCS of $k$-sparse signals with Gaussian $\bA$ and D1bCS with sub-Gaussian matrix. See  \cite{jacques2013robust,dirksen2021non,oymak2015near}. But the exact sparsity rarely occurs and may not be a realistic signal structure for practitioners. Hence,  effectively sparse signals residing in a scaled $\ell_1$-ball are   of particular interest. It was also known \cite{dirksen2021non,oymak2015near} that HDM achieves uniform error rate  $\widetilde{O}((\frac{k}{m})^{1/3})$ for recovering signals in $\sqrt{k}\mathbb{B}_1^n$ in 1bCS and D1bCS. As a side contribution, we will unify these information-theoretic bounds on the converse or attainability in Section \ref{sec:itbound}. Particularly, we reproduce the lower bound $\Omega(\frac{k}{mL})$ in Theorem \ref{thm:lower}, and then    we show in Theorem \ref{thm:itup} that HDM achieves the error rates stated above for any quantized compressed sensing systems with sub-Gaussian sensing vectors and the following two essential components:
 \begin{itemize}
     \item (Small-ball probability)  an upper bound on the   probability of the unquantized measurement (after dithering) being close to some  $b_j$;

     \item (Separation probability) a lower bound on the probability of two signals having different quantized measurement.\footnote{See Appendix \ref{appendix_sepa} for the reason why this is referred to as separation probability.}  
 \end{itemize}

The main goal of this work is to devise efficient algorithms  that can achieve the optimal reconstruction error rate. Let us first review some of the computationally efficient and theoretically guaranteed algorithms in this area.

For noiseless 1bCS, the linear programming approach that achieves $\widetilde{O}((\frac{k}{m})^{1/5})$ appeared as the first such an algorithm \cite{plan2013one}, then the convex relaxation \cite{plan2012robust} was  shown to achieve $\widetilde{O}((\frac{k}{m})^{1/4})$ as well as the robustness to a set of noise patterns. The recent work \cite{chinot2022adaboost} studied a more general binary classification context and showed that Adaboost achieves an error rate $\widetilde{O}((\frac{k}{m})^{1/3})$.
While these algorithms are capable of recovering effectively sparse signals, when restricted to $k$-sparse signals, generalized Lasso and projected-back projection (PBP) \rev{achieve a faster rate of $\widetilde{O}((\frac{k}{m})^{1/2})$} \cite{plan2017high,plan2016generalized}.

For D1bCS and DMbCS that are compatible with {\it sub-Gaussian} sensing matrix, generalized Lasso achieves the non-uniform error rate of $\widetilde{O}((\frac{\Delta k}{m})^{1/2})$ for a fixed signal in $\Sigma^n_k$ \cite{thrampoulidis2020generalized}. Note that  $\Delta=\delta$ in DMbCS. Also note that some DMbCS results are obtained under the uniform quantizer $\calQ_\delta$ (instead of $\calQ_{\delta,L}$), and $\delta$ in these results could be roughly  thought of as $1/L$---the reason is that $\calQ_{\delta,L}$ can cover inputs bounded by $\frac{L\delta}{2}$ without saturation, and $L\delta=O(\sqrt{\log m})$ is typically sufficient under sub-Gaussian matrix. Uniformly for all effectively sparse signals living in an $\ell_2$-ball, minimizing the one-sided $\ell_1$-loss over the signal space achieves $\widetilde{O}((\frac{\Delta k}{m})^{1/3})$ \cite{jung2021quantized}. This is sharper than $\widetilde{O}((\frac{k}{m})^{1/4})$ achieved by an earlier convex relaxation approach  \cite{dirksen2021non} for D1bCS, which is based on a loss function   being  ignorant of the dithers $(\tau_i)_{i=1}^m$. For DMbCS,   PBP surprisingly provides accurate estimate under RIP matrices \cite{xu2020quantized}, but note that its sharpest error rate $\widetilde{O}((1+\delta)(\frac{k}{m})^{1/2})$ for $k$-sparse signals exhibits a worse leading factor than generalized Lasso when $\delta\ll 1$. While the results for generalized Lasso in \cite{thrampoulidis2020generalized} are nonuniform, the recent work   \cite{genzel2023unified} established the uniform error rate $\widetilde{O}(\sqrt{\Delta}(\frac{k}{m})^{1/4})$ over all effectively sparse signals for 1bCS, D1bCS and the uniform quantized model (\ref{eq:dqcs}).

In summary, all the results reviewed above are no faster than $\widetilde{O}(\Delta(\frac{k}{m})^{1/2})$ in recovering $k$-sparse signals, thereby essentially inferior to the optimal rate $O(\frac{k}{mL})$; and are no faster than $\widetilde{O}((\frac{\Delta k}{m})^{1/3})$ in recovering effectively sparse signals, for which the optimal rate, to our best knowledge, remains unclear (see Remark \ref{lower_eff} for further discussion).

Following some prior attempts \cite{liu2019one,friedlander2021nbiht}, normalized binary iterative hard thresholding (NBIHT) was recently shown by \cite{matsumoto2024binary} to be optimal in recovering all $\bx\in \Sigma^n_k\cap \mathbb{S}^{n-1}$ from the 1bCS measurements $\by=\sign(\bA\bx)$. That is, NBIHT achieves a uniform error rate of  $\widetilde{O}(\frac{k}{m})$ in $\ell_2$-norm.  In our prior work  \cite{chen2024one}, similar   ideas were used to establish an efficient and optimal algorithm for one-bit phase retrieval.  Nonetheless, similar optimal results are still lacking for other models such as D1bCS and DMbCS, as well as other signal structures like low-rank matrices and effectively sparse signals. Our results fill in this void. 

\begin{table}[ht!]
\caption{A Partial Review of Efficient Algorithms for Recovering (Effectively) Sparse Signals} 
\centering
\renewcommand\arraystretch{1.4}
\begin{threeparttable}\label{table1}
\begin{tabular}{cccc}
    \hline\hline
    1bCS algorithm~~~~ & ~~~~ error rate ~ ~~~ &~~~~ signal space ~ ~~~& ~~~~uniformity \\
    \hline
 Linear Program \cite{plan2013one} & $\widetilde{O}((\frac{k}{m})^{1/5})$  & $\sqrt{k}\mathbb{B}_1^n\cap\mathbb{S}^{n-1}$ &  \ding{51}\\ 
     Convex Relaxation \cite{plan2012robust}& $\widetilde{O}((\frac{k}{m})^{1/4})$ & $\sqrt{k}\mathbb{B}_1^n\cap\mathbb{S}^{n-1}$ & \ding{55} \\
     Generalized Lasso \cite{plan2016generalized} & $\widetilde{O}((\frac{k}{m})^{1/2})$ & $\Sigma^n_k\cap \mathbb{S}^{n-1}$ & \ding{55}
    \\ 
    PBP \cite{plan2017high} & $\widetilde{O}((\frac{k}{m})^{1/2})$ & $\Sigma^n_k\cap \mathbb{S}^{n-1}$ & \ding{55}
    \\
    Adaboost \cite{chinot2022adaboost} & $\widetilde{O}((\frac{k}{m})^{1/3})$ & $\sqrt{k}\mathbb{B}_1^n \cap \mathbb{S}^{n-1}$& \rev{\ding{55}}
    \\
    \textbf{PGD} (a.k.a. NBIHT) \cite{matsumoto2024binary} &  $\widetilde{O}(\frac{k}{m})$ &  $\Sigma^n_k\cap \mathbb{S}^{n-1}$  & \ding{51}\\
    \textbf{PGD}   & $\widetilde{O}((\frac{k}{m})^{1/3})$ & $\sqrt{k}\mathbb{B}_1^n\cap \mathbb{S}^{n-1}$ &  \ding{51}
    \\
    \hline
    \hline 
    D1bCS algorithm~~~ &~~~~   error rate~~~~ & ~~~~signal space   ~~~~&~~~~ uniformity \\
    \hline 
    Convex Relaxation \cite{dirksen2021non} & $\widetilde{O}((\frac{k}{m})^{1/4})$  &  $\sqrt{k}\mathbb{B}_1^n\cap \mathbb{B}_2^n$& \ding{51} \\ 
   Con. Rel. w.r.t. (\ref{eq:1sided}) \cite{jung2021quantized} & $\widetilde{O}((\frac{k}{m})^{1/3})$ &   $\sqrt{k}\mathbb{B}_1^n\cap \mathbb{B}_2^n$& \ding{51} \\
   Generalized Lasso \cite{thrampoulidis2020generalized} & $\widetilde{O}((\frac{k}{m})^{1/2})$ &  $\Sigma^{n}_k\cap \mathbb{B}_2^n$ & \ding{55} \\
   \textbf{PGD}   &  $\widetilde{O}(\frac{k}{m})$ & $\Sigma^{n}_k\cap \mathbb{B}_2^n$ & \ding{51} \\
      \textbf{PGD}    &  $\widetilde{O}((\frac{k}{m})^{1/3})$ & $\sqrt{k}\mathbb{B}_1^n\cap\mathbb{B}_2^n$ & \ding{51}
    \\
    \hline
    \hline 
    DMbCS algorithm~~~ &~~~~   error rate ~~~~& ~~~~signal space   ~~~~&~~~~ uniformity \\\hline  
      Con. Rel. w.r.t. (\ref{eq:1sided}) \cite{jung2021quantized} & $\widetilde{O}((\frac{k}{mL})^{1/3})$ &$\sqrt{k}\mathbb{B}_1^n\cap \mathbb{B}_2^n$  & \ding{51},$\calQ_{\delta,L}$\\ 
      Generalized Lasso \cite{chen2024uniform,thrampoulidis2020generalized} & $\widetilde{O}(\delta(\frac{k}{m})^{1/2})$& $\Sigma^n_k\cap \mathbb{B}_2^n$& \makecell[c]{\cite[\ding{55},$\calQ_\delta$]{thrampoulidis2020generalized}\\\cite[\ding{51},$\calQ_\delta$]{chen2024uniform}}\\ PBP \cite{xu2020quantized} & 
      $\widetilde{O}((1+\delta)(\frac{k}{m})^{1/2})$ & $\Sigma^n_k\cap \mathbb{B}_2^n$ &  \ding{51},$\calQ_\delta$
      \\
      \textbf{PGD}   & $\widetilde{O}(\frac{k}{mL})$  & $\Sigma^n_k\cap \mathbb{B}_2^n$ &  \ding{51},$\calQ_{\delta,L}$
      \\
 \textbf{PGD} & $\widetilde{O}((\frac{k}{mL})^{1/3})$  & $\sqrt{k}\mathbb{B}_1^n\cap \mathbb{B}_2^n$ &  \ding{51},$\calQ_{\delta,L}$
      \\
       \hline
\end{tabular} 
\begin{tablenotes}
\footnotesize
\item[$\ast$] We only review results for 1bCS under standard Gaussian $\bA$, D1bCS/DMbCS under sub-Gaussian $\bA$; we review the {\it sharpest} result for a work containing multiple guarantees of different flavors. \rev{All these results are valid under the sample complexity of $m=\widetilde{\Omega}(k)$, ignoring polynomial dependence on $\lambda$ in D1bCS and DMbCS which typically scales as $\widetilde{O}(1)$.}
\item[$\dag$] The result is said to be uniform (non-uniform, resp.) if it holds for all signals (a fixed signal, resp.) in $\calX$  with a single draw of the sensing ensemble.
\item[$\ddag$] For DMbCS, we enclose under ``uniformity'' whether the result is for $\calQ_\delta$ (\ref{eq:dqcs}) or the less informative $\calQ_{\delta,L}$ (\ref{eq:mbsatu}). 
\item[$\mathsection$] \rev{For D1bCS, the rate $\widetilde{O}((k/m)^{1/2})$ can be achieved uniformly over $\Sigma^n_k\cap \mathbb{B}_2^n$ by some efficient procedures even under partial circulant matrices that are harder to work with than sub-Gaussian matrices \cite{dirksen2023robust}.} 
\end{tablenotes}
\end{threeparttable}
\end{table} 

\subsection{Main contributions}

\rev{Having introduced the general quantized model, its main special cases, and the closest algorithmic guarantees in the literature, we now summarize the main contributions of the paper.}

We propose a {\it projected gradient descent} (PGD) algorithm for the general quantized compressed sensing problem. The algorithm performs (sub)gradient descent with respect to the one-sided $\ell_1$-loss, followed by the sequential projections onto the star-shaped set $\calK$ and the norm constraint set $\mathbbm{A}_\alpha^\beta$. Note that the computational efficiency of PGD  only hinges on the projection onto $\calK$, which can be computed efficiently for (effectively) sparse vectors and low-rank matrices. Our major theoretical finding is as follows: under suitable conditions, including those for analyzing HDM and some {\it additional moment bounds} which convey a type of independence between the marginals of the sensing vector and the separation event (cf. Assumptions \ref{assum_sg1}--\ref{assum_sg3} and relevant interpretations),  
\begin{align*}
    \textit{PGD achieves essentially the same error rate as HDM, up to logarithmic factors.}
\end{align*}

Our proof strategies are inspired by arguments in \cite{matsumoto2024binary,chen2024one}, while as will be clear, we provide significant generalizations (in terms of signal structure   captured by a star-shaped set, the model, the algorithm per se) and new technical contributions (especially in the multi-bit case).

The convergence of PGD follows from a type of {\it restricted approximate invertibility condition} (RAIC) with respect to the interplay among several parties: the quantized sampler  $(\calQ,\bA,\btau)$, the star-shaped set $\calK$, and the step size $\eta$. With true signal $\bx$, loss function $\calL(\bu)$ and step size $\eta$, RAIC provides a uniform bound on the dual norm of $\bu-\bx-\eta \partial \calL(\bu)$ with respect to certain low-complexity set. Thus, RAIC essentially conveys a message that the actual subgradient descent step $\eta \calL(\bu)$ approximates the ``ideal'' step $\bu-\bx$, up to some error terms. \rev{We note that this intuition of RAIC departs from that of the  {\it regularity condition} commonly used in statistical nonconvex optimization (e.g., \cite{candes2015phase}),} and we expect that RAIC is interesting in its own right and will prove useful in many other nonconvex optimization problems. \rev{For instance, the subsequent works \cite{chen2025unified,abdalla2026robust} analyzed general nonlinear estimation problems via RAIC, with examples including phase retrieval, single index models, and so on.}

To establish RAIC under the proposed conditions, we invoke a covering argument similar to \cite{matsumoto2024binary,chen2024one} with notable generalizations in that our analysis is built upon generic assumptions with a set of refinements such as we no longer need to separately treat the large-distance regime and small-distance regime. New and fundamental challenges arise in the multi-bit setting where the more intricate gradient adds great challenge to the above-described analysis and forbids a unified treatment to the general context of quantized compressed sensing. To overcome these difficulties, we propose to    
 {\it clip} the gradient to that in the 1-bit case and then control the deviation caused by gradient clipping. This idea proves effective if we restrict our attention to two points having distance much smaller than the quantization resolution. The main  intuition is that  {\it the multi-bit quantizer behaves similarly to the 1-bit quantizer in terms of its ability to distinguish such two close enough points}. Nonetheless, this approach only leads to local convergence for the multi-bit case. We then introduce a complementary approach, which is based on a type of {\it product embedding} (or {\it projection distortion}) property, to establish global convergence.

 To apply our unified framework to a specific model, it then suffices to validate the conditions, which can   be accomplished by elementary calculations and estimations. We  will particularly specialize our theory to the popular models 1bCS, D1bCS and DMbCS, showing that PGD achieves error rates which {\it improve on or match the sharpest existing results in all contexts}. In particular, we summarize in Table \ref{table1} both existing results and the ones we establish for recovering (effectively) sparse signals.

 While prior works \cite{matsumoto2024binary,friedlander2021nbiht} are only for sparse recovery, our works also straightforwardly yield low-rank recovery guarantees: for all three models (1bCS, D1bCS and DMbCS), under appropriate Frobenius norm constraint, our PGD algorithm can recover exactly low-rank matrices in $M^{n_1,n_2}_{\bar{r}}$ with an   error rate of $\widetilde{O}(\frac{\bar{r}(n_1+n_2)}{mL})$ which is optimal in light of Theorem \ref{thm:lower}, and the effectively low-rank matrices with nuclear norm bounded by $\sqrt{\bar{r}}$  to an error rate of $\widetilde{O}((\frac{\bar{r}(n_1+n_2)}{mL})^{1/3})$. 

\subsection{Notation}\label{sec:notation}

We introduce some generic notation used throughout the work. Regular letters denote scalars, and boldface letters denote vectors and matrices.  We denote by $\mathbb{B}_2^n(\bu;r)$ the $\ell_2$-ball with radius $r$ and center $\bu$, shortened to $\mathbb{B}_2^n(r)$ if $\bu=0$. For $a,b\in \mathbb{R}$ we write $a\vee b = \max\{a,b\}$ and $a\wedge b = \min\{a,b\}$. For   $\calU,\calV\subset\mathbb{R}^n$, $\calP_{\calU}(\cdot)$ denotes the projection onto $\calU$ under $\ell_2$-norm, and we let $t\calU=\{t\bu:\bu\in\calU\}$ and  $\calU-\calV=\{\bu-\bv:\bu\in\calU,\bv\in\calV\}$. The localization of $\calU$ to the radius of $r$ is defined as $\calU_{(r)}:=(\calU-\calU)\cap \mathbb{B}_2^n(r)$.
  We work with the dual norm $\|\bu\|_{\calU^\circ}=\sup_{\bw\in \calU}\bw^\top\bu$.

  Two natural quantities are used to characterize the complexity of a set $\calU$: the Gaussian width $\omega(\calU)=\mathbbm{E}\sup_{\bu\in\calU}\langle\bg,\bu\rangle$ where $\bg\sim\calN(0,\bI_n)$ and $\langle \bg,\bu\rangle = \bg^\top\bu$, and the metric entropy $\scrH(\calU,r) = \log\scrN(\calU,r)$ where $\scrN(\calU,r)$ is the minimal number of radius-$r$ $\ell_2$-ball needed to cover $\calU$, that is, the covering number. We also denote by $|\calU|$ the cardinality of a finite set $\calU$. The sub-Gaussian norm of a random variable $X$ is defined as $\|X\|_{\psi_2}=\inf\{t>0: \mathbbm{E}\exp(\frac{X^2}{t^2})\le 2\}$, while $\|X\|_{L^p} = (\mathbbm{E}|X|^p)^{1/p}$ denotes the $L_p$-norm. The sub-Gaussian norm of a random vector $\bX$ is defined as $\|\bX\|_{\psi_2}=\sup_{\bv\in \mathbb{S}^{n-1}}\|\bv^\top\bX\|_{\psi_2}$.

  We use $C_0,C_1,C_2,\cdots,c_0,c_1,c_2,\cdots$ to denote constants that may be universal or depend on other quantities that will be specified. Their values may vary from line to line. We write $T_1\lesssim T_2$ or $T_1=O(T_2)$ if $T_1\le C_1T_2$ for some $C_1$, and write $T_1\gtrsim T_2$ or $T_1=\Omega(T_2)$ if $T_1\ge c_1T_1$ for some $c_1$. Note that $\widetilde{O}(\cdot)$ and $\widetilde{\Omega}(\cdot)$ are the less precise versions of $O(\cdot)$ and $\Omega(\cdot)$ which ignore logarithmic factors. The symbol ``$\sim$'' will be used to  connect identical distributions (that may be represented by certain random variables).

The rest of the paper is organized as follows. We first present the information-theoretic bounds in Section \ref{sec:itbound}.
Section \ref{sec:unified} formally introduces PGD and its unified analysis. In Section \ref{sec:instance} we specialize the theory to 1bCS, D1bCS and DMbCS. Section \ref{sec:discuss} discusses the extension of our results to noisy settings. We further complement our theoretical results via a number of numerical simulations in Section \ref{sec:experiment}. We conclude the paper in Section \ref{sec:conclusion}. Most of the proofs are relegated to the appendix. 

\section{Information-Theoretic Bounds}\label{sec:itbound}
We focus on non-adaptive memoryless quantization where both $\bA$ and $\btau$ are drawn before observing any measurement. 
Note that the information-theoretic lower bounds are well established in the literature (see, e.g., \cite{boufounos2015quantization,jacques2013robust,chen2024one}) and follow from a simple counting argument.

\begin{theorem}[Information-theoretic lower bound]\label{thm:lower}
Suppose there is a $K$-dimensional linear subspace $V_K$ contained in $\calK$. 
    Given $\bA\in \mathbb{R}^{m\times n}$, $\btau\in \mathbb{R}^m$ and $\beta>\alpha\ge 0$, our goal is to recover $\bx\in \calW:=\calK\cap\mathbb{A}_{\alpha,\beta}$ from $\by=\calQ(\bA\bx-\btau)$ where $\calQ$ is the $L$-level quantizer as in (\ref{eq:Llevel}),
    or recover $\bx\in\calW:=\calK\cap \mathbb{S}^{n-1}$ from $\by = \sign(\bA\bx)$ as in 1bCS (\ref{eq:1bcs}). If $(L-1)m \ge K$, then for any decoder that takes $\by$ as input and returns $\hat{\bx}$ as an estimate of $\bx$, there exists some constant $c>0$ that may depend on $(\alpha,\beta)$ such that  
    \begin{align*}
        \sup_{\bx\in\calW}~\|\hat{\bx}-\bx\|_2\ge\frac{cK}{mL}.
    \end{align*}
\end{theorem}
\begin{proof}
    We include a proof in Appendix \ref{lower_proof} for the sake of completeness. 
\end{proof}
\begin{rem}\label{rem:optimal}
Theorem \ref{thm:lower} suggests that for recovering $k$-sparse signals from $\calQ(\bA\bx-\btau)$,  no algorithm     
  can achieve an error rate faster than $O(\frac{k}{mL})$. Similarly, because there is a $\bar{r}(n_1\vee n_2)$-dimensional  subspace contained in $M^{n_1,n_2}_{\bar{r}}$, no algorithm 
  achieves error rate faster than $O(\frac{\bar{r}(n_1+n_2)}{mL})$
  for recovering rank-$\bar{r}$ matrices in $\mathbb{R}^{n_1\times n_2}$. 
\end{rem}
\begin{rem}
    \label{lower_eff}
Relatively little is known about the information-theoretic lower bound for recovering effectively sparse signals in $\calK=\sqrt{k}\mathbb{B}_1^n$. The closest development is probably \cite{dirksen2022sharp} who showed for D1bCS that $\Omega(\delta^{-3})$ Gaussian measurements are needed to achieve the uniform random hyperplane tessellation   over $\sqrt{k}\mathbb{B}_1^n\cap \mathbb{B}_2^n$ with distortion $\delta$, which is nonetheless not necessary for signal reconstruction to $\delta$ $\ell_2$-error. In fact, the optimality of $\widetilde{O}((\frac{k}{m})^{1/3})$ remains an open question (see, e.g., \cite{chinot2022adaboost}).   
\end{rem}
\begin{rem}
    \rev{We emphasize that this work is restricted to memoryless quantization. In fact, if adaptive quantization is considered, then an exponentially decaying error rate is achievable \cite{baraniuk2017exponential}.}
\end{rem}

  For 1bCS with standard Gaussian $\bA$ and D1bCS with sub-Gaussian $\bA$ and uniform $\btau$ \rev{(where $L=2$)}, it was known that the \rev{error rates in Remark \ref{rem:optimal}, namely $O(\frac{k}{m})$ for recovering $k$-sparse vectors and $O(\frac{\bar{r}(n_1+n_2)}{m})$ for recovering matrices in $M^{n_1,n_2}_{\bar{r}}$,}  are attainable up to a logarithmic factor by (constrained) hamming distance minimization (HDM): 
\begin{align}\label{eq:hdm}
	 \hat{\bx}_{\rm hdm}=\arg\min_{\bu\in\calX}~d_H\big(\calQ(\bA\bu-\btau),\by\big), 
\end{align}
where
\begin{align*}
	d_H(\bu,\bv) = \sum_{i=1}^m \mathbbm{1}(u_i\ne v_i).
\end{align*}
See \cite[Theorem 2]{jacques2013robust}, \cite[Theorem 2.5]{oymak2015near} and  \cite[Theorem 1.9]{dirksen2021non}.  To unify these existing performance bounds, we shall develop a useful notion referred to as  {\it quantized embedding property}  to relate the Euclidean distance of two signals and the hamming distance between their quantized measurements. To establish the quantized embedding property, we only need to make three generic assumptions.

As mentioned, we work with sub-Gaussian design throughout this work. 

\begin{assumption}[Sensing vectors]\label{assum1}
   The sensing vectors $\ba_1,\ba_2,\cdots,\ba_m$ are i.i.d. isotropic sub-Gaussian, i.e., they satisfy $\mathbbm{E}(\ba_i\ba_i^\top)=\bI_n$ and $\|\ba_i\|_{\psi_2}\le C_0$ for some absolute constant $C_0$.
\end{assumption}

\begin{rem}
    We pause to provide some implications of Assumption \ref{assum1}. First, it leads to the tail bounds and the moment bounds (e.g., \cite{vershynin2018high})
\begin{subequations}
    \label{eq:d1bcs_sg_abs}
    \begin{gather}\label{eq:d1bcs_sgtail}
    \mathbbm{P}(|\ba_i^\top\bu|\ge t)\le 2\exp(-2\bar{c}_0t^2),\quad \forall t\ge 0,~\forall\bu\in \mathbb{S}^{n-1},\\\label{eq:d1bcs_moment}
    (\mathbbm{E}|\ba_i^\top \bu|^p)^{1/p} \le \bar{c}_1\sqrt{p},\quad\forall p\in \mathbb{Z}_+,~\forall \bu\in\mathbb{S}^{n-1}
\end{gather}
\end{subequations}
which hold for some absolute constants $\bar{c}_0$ and $\bar{c}_1$. Moreover,  by Lemma  \ref{lem:L1L2} in Appendix \ref{app:lemma}, there exists some absolute constant $L_0$ such that   
\begin{subequations}
    \begin{gather}
        \label{eq:L1L2equ}
    \mathbbm{E}|\ba_i^\top\bu| \ge L_0,\quad\forall \bu\in \mathbb{S}^{n-1},\\
    \label{eq:prob_ge12}\mathbbm{P}\Big(|\ba_i^\top\bu|\ge\frac{1}{2}\Big)\ge 2L_0,\quad \forall \bu\in \mathbb{S}^{n-1}.
    \end{gather}
\end{subequations}  
\end{rem}

Next, we characterize the small-ball probability of the unquantized measurement $\ba_i^\top \bu-\tau_i$ being close to some quantization thresholds $(b_j)_{j=1}^{L-1}$. In all our assumptions, we use $c^{(i)}$'s to denote constants that are typically universal when specialized to a specific model. We will use $\varepsilon^{(i)}$'s instead if these constants are extremely small or $0$. \rev{Recall that $\Delta$ is defined in (\ref{defreso}) and we assume that (\ref{wolg}) holds with no loss of generality. Also, $\Lambda$ denotes the dithering level such that $\btau\sim {\rm Unif}[-\lambda,\lambda]^m$.}

\begin{assumption}[Small-ball probability]\label{assum2}
    For some $c^{(1)}>0$ and for any $\bu\in \mathbbm{A}_\alpha^\beta$, it holds  that 
    \begin{align*}
        \mathbbm{P}\left(\min_{j\in[L-1]}|\ba_i^\top\bu-\tau_i-b_j|\le t\right)\le\frac{c^{(1)}t}{\Delta\vee\Lambda},\quad\forall t>0. 
    \end{align*}
\end{assumption}

The final assumption for quantized embedding property is on the probability of two points being distinguished by the quantized sampler. This quantity, referred to as separation probability, is defined as 
\begin{align}\label{eq:Puv}
    \sfP_{\bu,\bv}:= \mathbbm{P}\Big(\calQ(\ba_i^\top\bu-\tau_i)\ne \calQ(\ba_i^\top\bv-\tau_i)\Big)
\end{align}
and will be   recurring in subsequent   developments. 
\begin{assumption}[Lower bound on $\sfP_{\bu,\bv}$]\label{assum3}
    For any $\bu,\bv\in\mathbbm{A}_\alpha^\beta$,
    \begin{align*} \sfP_{\bu,\bv}\ge  c^{(2)}\min\left\{\frac{\|\bu-\bv\|_2}{\Delta\vee\Lambda},1\right\} 
    \end{align*}
    holds for some  $c^{(2)}>0$. 
\end{assumption}
 
 \rev{We establish two high-probability events, referred to as quantized embedding properties, under these assumptions. The first event $E_{\rm small}$ states that the measurement vectors of two points $\bu,\bv$ of small Euclidean distance are also close in hamming distance. In contrast, the event $E_{\rm large}$ guarantees large hamming distance in measurement vectors for $\bu,\bv$ having large Euclidean distance.} 
 
\begin{theorem}[Quantized embedding property] 
    \label{thm:local_embed} Let $\calW$ be a  set contained in $\mathbbm{A}_\alpha^\beta$, for \rev{$\epsilon\in(0,c_0(\Delta\vee\Lambda))$}, we let $\epsilon' = \frac{c_1 \epsilon}{ \log^{1/2}(\frac{\Delta\vee\Lambda}{\epsilon})}$ and suppose
    \begin{align}
        \label{eq:local_scale}
        m\ge C_2 (\Delta\vee\Lambda)\left(\frac{\omega^2(\calW_{(3\epsilon'/2)})}{\epsilon^3}+\frac{\scrH(\calW,\epsilon')}{\epsilon}\right),
    \end{align}
    where $c_0$ and $c_1$ are small enough, $C_2$ are large enough.   We have the following two statements: 
    \begin{itemize}
    [leftmargin=5ex,topsep=0.25ex]
		\setlength\itemsep{-0.1em}
        \item   Under Assumptions \ref{assum1}--\ref{assum2}, for some constants $c_0,c_1,C_2,c_3,C_4$ depending on $c^{(1)}$, with probability at least $1-2\exp(-\frac{c_3m\epsilon}{\Delta\vee\Lambda})$ the event \rev{$E_{\rm small}$} holds: for any $\bu,\bv\in \calW$ obeying $\|\bu-\bv\|_2\le \frac{\epsilon'}{2}$, we have 
        \begin{align*}
            d_H\big(\calQ(\bA\bu-\btau),\calQ(\bA\bv-\btau)\big)\le \frac{C_4m\epsilon}{\Delta\vee\Lambda}. 
        \end{align*}
        \item Under Assumptions \ref{assum1}--\ref{assum3},  for some constants $c_0,c_1,C_2,c_5,c_6$ depending on $(c^{(1)},c^{(2)})$, with probability at least $1-\exp(-\frac{c_5m\epsilon}{\Delta\vee\Lambda})$ the event \rev{$E_{\rm large}$} holds: for any $\bu,\bv\in\calW$ obeying $\|\bu-\bv\|_2\ge 2\epsilon$, we have 
        \begin{align}\label{eq:El}
            d_H \big(\calQ(\bA\bu-\btau),\calQ(\bA\bv-\btau)\big) \ge c_6m \min\left\{\frac{\|\bu-\bv\|_2}{\Delta\vee\Lambda},1\right\} .
        \end{align}
    \end{itemize}
\end{theorem}
\begin{proof}
The proof of Theorem \ref{thm:local_embed} is similar to that from \cite[Theorem 2.1]{chen2024one} (which is inspired by pioneering works \cite{plan2014dimension,oymak2015near,dirksen2021non}) and is included in Appendix \ref{app:tessella} for completeness. 
\end{proof}

Theorem \ref{thm:local_embed} immediately implies the following performance bound for HDM; this is an information-theoretic upper bound due to the computational intractability of HDM. 

\begin{theorem}
    [Information-theoretic upper bound]\label{thm:itup} Under Assumptions \ref{assum1}--\ref{assum3}, for some star-shaped set $\calK$, we recover $\bx\in \calX= \calK\cap \mathbbm{A}_\alpha^\beta$ from $\by=\calQ(\bA\bx-\btau)$ by solving HDM \rev{as in} (\ref{eq:hdm}). There exist constants \rev{$c_0,c_1,C_2,c_3$} depending on $(c^{(1)},c^{(2)})$, for any $r\in (0,c_0(\Delta\vee\Lambda))$ and the corresponding $r'=\frac{c_1r}{\log^{1/2}(\frac{\Delta\vee\Lambda}{r})}$, if 
    \begin{align*}
        m \ge C_2(\Delta\vee\Lambda)\left(\frac{\omega^2(\calK_{(r')})}{r^3}+\frac{\scrH(\calX,r')}{r}\right),
    \end{align*}
    then with probability at least $1-\exp(-c_3\scrH(\calX,r'))$ we have 
    $$\|\hat{\bx}_{\rm hdm}-\bx\|_2 \le 2r,\quad\forall \bx\in\calX.$$
\end{theorem}
\begin{proof}
\rev{This is a direct implication of $E_{\rm large}$ in Theorem \ref{thm:local_embed}. See the detailed proof in Appendix \ref{app:itupper}.}
\end{proof}

    \rev{We mention that Theorems \ref{thm:local_embed}--\ref{thm:itup} for 1bCS and D1bCS are known in the literature \cite{oymak2015near,dirksen2021non}, and our main work here is to identify the essential components (i.e., Assumptions \ref{assum1}, \ref{assum2}, \ref{assum3}) for these results to hold true and incorporate the multi-bit model as well.}
    
We shall discuss some explicit upper bounds implied by Theorem \ref{thm:itup}. In the following we treat $\alpha\le\beta$ as given constants and assume $\beta\le1$. 

\begin{rem}[Exactly sparse vectors or low-rank matrices]\label{rem:exactsparserate} To consider the recovery of exactly $k$-sparse vectors, we specialize Theorem  \ref{thm:itup} to $\calK=\Sigma^n_k$. Note that $\omega^2((\Sigma^n_k)_{(r')})\le (r')^2\omega^2(\Sigma^n_{2k}\cap \mathbb{B}_2^{2n})$. \rev{By (\ref{sparsegw})--(\ref{sparseentropy}) in Lemma \ref{lem:gwcoveres},} we establish that HDM achieves error rate $\widetilde{O}(\frac{(\Delta\vee\Lambda)k}{m})$ \rev{under $m=\widetilde{\Omega}(k)$}. Similarly, \rev{we set 
 $\calK=M^{n_1,n_2}_{\bar{r}}$ and use (\ref{lrgw})--(\ref{lrcover}) in Lemma \ref{lem:gwcoveres},} yielding 
 that HDM achieves Frobenius norm error rate $\widetilde{O}(\frac{(\Delta\vee\Lambda)\bar{r}(n_1+n_2)}{m})$ in the recovery of rank-$\bar{r}$ matrices in $\mathbb{R}^{n_1\times n_2}$, \rev{so long as $m=\widetilde{\Omega}(\bar{r}(n_1+n_2))$}.  These match the optimal rate implied by Theorem \ref{rem:optimal}. (In the sequel, we will validate Assumptions \ref{assum2}--\ref{assum3} under $L(\Delta\vee\Lambda)=O(1)$ for 1bCS, D1bCS and DMbCS, thus we can replace $\Delta\vee\Lambda$ in the above error rates with $L^{-1}$.)
\end{rem}

 For $q\in(0,1]$, the (quasi) $\ell_q$ norm of $\bu=(u_i)_{i=1}^n$ is defined as $\|\bu\|_q= (\sum_{i=1}^n|u_i|^q)^{1/q}$. Also denote the unit $\ell_q$ ball in $\mathbb{R}^n$ by $\mathbb{B}_q^n = \{\bu\in \mathbb{R}^n:\|\bu\|_q\le 1\}$. In what follows, we consider the recovery of effectively sparse vectors. One can obtain the counterparts for effectively low-rank matrices by replacing the vector $\ell_q$ norm by the Schatten-$q$ (quasi) norm (i.e., the $\ell_q$ quasi-norm of the singular values).

\begin{rem}[Effectively sparse vectors] \label{rem:effsparse}For $k\le n$,
it is standard to set $\calK=\sqrt{k}\mathbb{B}_1^n$ to capture the effectively $k$-sparse vectors in $\mathbb{B}_2^n$ (e.g., \cite{plan2012robust,plan2013one,chinot2022adaboost,jung2021quantized}). \rev{By $\omega^2(\calK_{(r')})\le \omega^2(2\sqrt{k}\mathbb{B}_1^n \cap \mathbb{B}_2^n)$, along with (\ref{sparsegw}) and (\ref{l1sparseentropy}) in Lemma \ref{lem:gwcoveres}, Theorem \ref{thm:itup} implies the error rate $\widetilde{O}((\frac{(\Delta\vee\Lambda)k}{m})^{1/3})$ of HDM under $m=\widetilde{\Omega}((\Delta\vee\Lambda)^{-2}k)$. One can indeed establish a fine-grained upper bound by $\omega^2(\calK_{(r')})\le \omega^2(r'\mathbb{B}_2^n)\lesssim (r')^2n$ and $\scrH(\calX,r')\le \scrH(\mathbb{B}_2^n,r')\le n\log(\frac{3}{r'})$. This yields an upper bound  $\widetilde{O}(\frac{(\Delta\vee\Lambda)n}{m})$ under $m=\widetilde{\Omega}(n)$. Under $m=\widetilde{\Omega}((\Delta\vee\Lambda)^{-2}k)$, since $\frac{(\Delta \vee \Lambda)n}{m}\lesssim (\frac{(\Delta \vee \Lambda)k}{m})^{1/3}$ holds only when $m=\widetilde{\Omega}(n)$, we conclude that HDM achieves the error rate 
\[\min\bigg\{\Big(\frac{(\Delta\vee \Lambda)k}{m}\Big)^{1/3},\frac{(\Delta\vee \Lambda)n}{m}\bigg\},\quad\textrm{up to logarithmic factors.}\]
} 
Notably, this exhibits a transition from $(k/m)^{1/3}$ to $n/m$ when $k$ increases to $n$.  
\end{rem}

\begin{rem}[$\ell_q$ sparse vectors]\label{rem:lqsparse}  \rev{For $k$ much smaller than $n$, there is an essential gap between the $\frac{(\Delta\vee\Lambda)k}{m}$ rate for exactly $k$-sparse signals and $(\frac{(\Delta\vee\Lambda)k}{m})^{1/3}$ for effectively $k$-sparse signals. Intriguingly, we can bridge the two cases smoothly by considering the $\ell_q$ sparse signals living in $\calK_q:=k^{\frac{1}{q}-\frac{1}{2}}\mathbb{B}_q^n$ for $q\in(0,1]$. We note that $\calK_{q}\cap \mathbb{B}_2^n\subset \calK_{q'}\cap \mathbb{B}_2^n$ if $0<q\le q'\le 1$, hence larger $q$ corresponds to a weaker sparsity. By using Lemma \ref{lem:lq-width-entropy}, setting $\calK=\calK_q$ in Theorem \ref{thm:itup} establishes the HDM error rate
\[\min\bigg\{\Big(\frac{(\Delta\vee\Lambda)k}{m}\Big)^{\frac{2-q}{2+q}},\frac{(\Delta\vee\Lambda)n}{m}\bigg\},\quad\textrm{up to logarithmic factors}\]
under $m=\widetilde{\Omega}((\Delta\vee\Lambda)^{-\frac{2q}{2-q}}k)$.} This recovers the error rate in Remark \ref{rem:effsparse} if $q=1$, and the $\frac{(\Delta\vee\Lambda)k}{m}$ rate in Remark \ref{rem:exactsparserate} if $q\to 0$. 
\end{rem}

\section{Projected Gradient Descent}\label{sec:unified}
The main achievement of this paper is to develop a computationally efficient procedure to attain the error rates of HDM. We propose to do so via projected gradient descent (PGD). We shall see that this yields efficient algorithms that achieve the error rates in Remarks \ref{rem:exactsparserate}--\ref{rem:effsparse}, up to logarithmic factors. \rev{For the $\ell_q,~q\in(0,1)$ sparse vectors discussed in Remark  \ref{rem:lqsparse}, our PGD is not efficient due to the computational infeasibility of projection onto the nonconvex set $\calK_q$, yet we believe the development here is of interest, e.g., in light of some existing heuristic for approximating such projection (e.g., \cite{yang2022towards}). Indeed, some works (e.g., \cite{bahmani2013unifying}) analyzed  the performance of $\ell_q$ PGD for compressed sensing.}

We first derive a slightly different hamming distance loss. Under the $L$-level quantizer, the measurement $y_i=\calQ(\ba_i^\top\bx-\tau_i)$ can   be equivalently viewed as $L-1$ binary measurements:
\begin{align*}
    y_{ij} = \sign(\ba_i^\top\bx-\tau_i-b_j),\quad j = 1,\cdots,L-1.
\end{align*}
In total, we have  $(L-1)m$ binary measurements $\{\sign(\ba_i^\top\bx-\tau_i-b_j:j\in[L-1],~i\in[m])\}$, and it is natural to consider   the (scaled) hamming distance loss
    \begin{align*}
  \calL(\bu):=\frac{\Delta}{m}\sum_{i=1}^m \sum_{j=1}^{L-1}\mathbbm{1}\big(\sign(\ba_i^\top\bu-\tau_i-b_j)\ne y_{i,j}\big)= \frac{\Delta}{m}\sum_{i=1}^m\sum_{j=1}^{L-1}\mathbbm{1}\big(-y_{ij}(\ba_i^\top\bu-\tau_i-b_j)\ge 0\big).
\end{align*}
While this loss function is discrete and nonconvex, we shall relax $\calL(\bu)$ to the one-sided $\ell_1$-loss 
\begin{align}\nn 
    &\calL_1(\bu) = \frac{\Delta}{m}\sum_{i=1}^m\sum_{j=1}^{L-1} \max\big\{-y_{ij}(\ba_i^\top\bu-\tau_i-b_j),0\big\}
    \\&= \frac{\Delta}{2m}\sum_{i=1}^m \sum_{j=1}^{L-1}\Big[|\ba_i^\top\bu-\tau_i-b_j|-y_{ij}(\ba_i^\top\bu-\tau_i-b_j)\Big],\label{eq:1sided}
\end{align}
in light of the idea of replacing $\mathbbm{1}(t>0)$ by $\max\{t,0\}$. The (sub-)gradient of $\calL_1(\bu)$ at $\bx^{(t-1)}$ takes the form 
\begin{align*}
    \partial \calL_1(\bx^{(t-1)}) &= \frac{\Delta}{2m}\sum_{i=1}^m \sum_{j=1}^{L-1}\big(\sign(\ba_i^\top\bu-\tau_i-b_j)-y_{ij}\big)\ba_i\\
    &=\frac{1}{m}\sum_{i=1}^m \big(\calQ(\ba_i^\top \bx^{(t-1)}-\tau_i)-\calQ(\ba_i^\top \bx-\tau_i)\big)\ba_i,
\end{align*}
where the second equality holds because   we have assumed $q_{i+1}=q_i+\Delta$ for $i\in[L-1]$; see (\ref{wolg}).

As we are interested in {\it uniform reconstruction}, it is beneficial to work with a notation for the gradient \rev{that exhibits the dependence on the underlying signal}. Let us define
\begin{align}\label{eq:huvdef1}
    \bh(\bu,\bv) := \frac{1}{m}\sum_{i=1}^m\big(\calQ(\ba_i^\top\bu-\tau_i)-\calQ(\ba_i^\top\bv-\tau_i)\big)\ba_i= \frac{1}{m}\bA^\top\big(\calQ(\bA\bu-\btau)-\calQ(\bA\bv-\btau)\big), 
\end{align}
\rev{which denotes the gradient at $\bu$ when $\bv$ is the underlying signal.} 
 After performing gradient descent to get $\bx^{(t-1)}-\eta\cdot\bh(\bx^{(t-1)},\bx)$, the most natural choice is to project onto $\calX= \calK\cap \mathbbm{A}_{\alpha}^\beta$, but this is not immediately efficient even for convex $\calK$ (e.g., $\calX=\sqrt{k}\mathbb{B}_1^n\cap \mathbb{S}^{n-1}$). Therefore, we propose to project $\bx^{(t-1)}-\eta\cdot\bh(\bx^{(t-1)},\bx)$ onto $\calK$ and $\mathbbm{A}_{\alpha}^\beta$ sequentially. It should be noted that for cone $\calK$, the two ways are equivalent in light of $\calP_{\calX}=\calP_{\mathbbm{A}_{\alpha}^\beta}\circ\calP_{\calK}$.

We formally outline the proposed recovery procedure in Algorithm \ref{alg:pgd}.

 \begin{algorithm}
	\label{alg:pgd}
	\caption{Projected Gradient Descent for Quantized Compressed Sensing}
	\textbf{Input}: Quantized observations $\by=\calQ(\bA\bx-\btau)$, sensing ensemble $(\bA,\btau)$, a set $\calK$ capturing signal structure, a set $\mathbbm{A}_\alpha^\beta$ constraining signal norm, initialization $\bx^{(0)}$, step size $\eta$

	\textbf{For}
	$t = 1, 2, 3,\cdots$  \textbf{do}
	\begin{gather}\label{eq:pgdrule}
		\tilde{\bx}^{(t)} =  \calP_{\calK}\big(\bx^{(t-1)} - \eta \cdot \bh(\bx^{(t-1)},\bx)
		\big), \\ \label{normproj}
  \bx^{(t)} = \calP_{\mathbbm{A}_\alpha^\beta}\big(\tilde{\bx}^{(t)}\big).
	\end{gather}
\end{algorithm}
\begin{rem}
    For 1bCS of $k$-sparse signals, we shall set $\calK=\Sigma^n_k$ and $\mathbbm{A}_\alpha^\beta=\mathbb{S}^{n-1}$, then Algorithm \ref{alg:pgd} becomes NBIHT  \cite{jacques2013robust,matsumoto2024binary,friedlander2021nbiht}.  
\end{rem}

\subsection{Convergence via  RAIC}
 Our analysis of Algorithm \ref{alg:pgd} is based on the following deterministic property called {\it restricted approximate invertibility condition} (RAIC). It generalizes similar notions in \cite{matsumoto2024binary,chen2024one,friedlander2021nbiht} from $k$-sparse  signals ($\calK=\Sigma^n_k$) to general signal structures captured by star-shaped set.
\begin{definition}
    [RAIC] \label{def:raic}
    Under some quantizer $\calQ$, for some given $\calD\subset \mathbb{R}^n$ and $\bm{\mu}=(\mu_1,\mu_2,\mu_3,\mu_4)$ with non-negative scalars $(\mu_i)_{i=1}^4$, we say  $(\calQ,\bA,\btau,\calK,\eta)$ respects $(\calD,\bm{\mu})$-RAIC at scale $\phi>0$ if  
    \begin{align}\label{raicformu}
       \frac{1}{\phi}\| \bu-\bv-\eta\cdot\bh(\bu,\bv)\|_{\calK_{(\phi)}^\circ}\le  \mu_1\|\bu-\bv\|_2+\sqrt{\mu_2\cdot\|\bu-\bv\|_2}+\mu_3 
    \end{align}
    holds for any $\bu,\bv\in \calD$ obeying $\|\bu-\bv\|_2\le \mu_4$, where $\bh(\bu,\bv)$ is defined as in (\ref{eq:huvdef1}).
\end{definition}
In light of 
$$\frac{1}{\phi}\|\bu\|_{\calK^\circ_{(\phi)}} = \frac{1}{\phi}\sup_{\bw\in \calK_{(\phi)}}\langle\bw,\bu\rangle =\sup_{\bw\in (\frac{\calK}{\phi})_{(1)}}\langle \bw,\bu\rangle = \|\bu\|_{(\frac{\calK}{\phi})_{(1)}^\circ},$$
 RAIC delivers a uniform bound on the dual norm of $\bu-\bv- \eta \cdot\bh(\bu,\bv)$ with regard to $(\frac{\calK}{\phi})_{(1)}$; thus, on the current iterate $\bu=\bx^{(t)}$ and underlying signal $\bv=\bx$, RAIC gives the message that the ``ideal descent step'' $\bx^{(t)}-\bx$ (in the sense of $\bx^{(t)}-(\bx^{(t)}-\bx)= \bx$) is well approximated by the ``actual subgradient step'' $\eta \cdot \bh(\bx^{(t)},\bx)$, up to a few error terms. \rev{See Figure \ref{fig:raic} for the intuition of RAIC,  which departs from that of {\it regularity condition} widely adopted in the statistical nonconvex optimization literature (see, e.g., \cite{candes2015phase,tu2016low}).} 
\begin{figure}[ht!]
	\begin{centering}
		 \includegraphics[width=0.4\columnwidth]{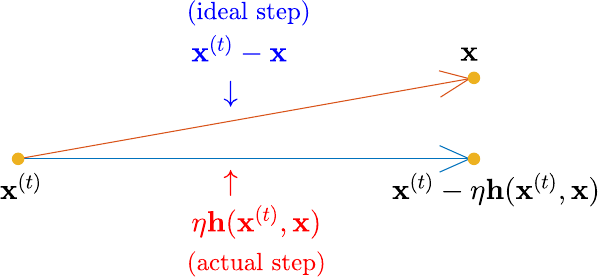}  
		\par
	\end{centering}
	
	\caption{\label{fig:raic}\small  RAIC indicates closeness of the ideal step and the actual subgradient step.} 
\end{figure}

\begin{rem}[Selection of $\phi$]\label{rem:suffraic}
Because $\calK$ is a star-shaped set, we have $\frac{\calK}{\phi_1} \subset\frac{\calK}{\phi_2}$ if $\phi_1>\phi_2$, and hence $\frac{1}{\phi}\|\bu\|_{\calK^\circ_{(\phi)}} = \|\bu\|_{(\frac{\calK}{\phi})^\circ_{(1)}}$ decreases with $\phi$. Therefore, RAIC with smaller $\phi$ is stronger and accordingly harder to establish. However, in controlling the error after projection, we will rely on Lemma \ref{planlem} which puts $\phi$ and $\frac{1}{\phi}\|\bu-\bv-\eta\cdot\bh(\bu,\bv)\|_{\calK_{(\phi)}^\circ}$ on an equal footing, making RAIC with overly large $\phi$ useless. We note that a reasonable tradeoff is to select $\phi$ at an order comparable to $\frac{1}{\phi}\|\bu-\bv-\eta\cdot\bh(\bu,\bv)\|_{\calK_{(\phi)}^\circ}$. In the sequel, we will choose $\phi=r$ if RAIC with $\mu_3 =\widetilde{O}(r)$ is sought after.  
\end{rem}
\begin{rem}[Conic RAIC]
    If $\calK$ is a cone, then the choice of $\phi>0$ becomes inessential because (\ref{raicformu}) is always equivalent to 
    \begin{align*}
        \|\bu-\bv-\eta\cdot\bh(\bu,\bv)\|_{\calK^\circ_{(1)}}\le \mu_1\|\bu-\bv\|_2+\sqrt{\mu_2\cdot\|\bu-\bv\|_2}+ \mu_3,
    \end{align*} 
    which is   similar to the definitions of RAIC in \cite{friedlander2021nbiht,matsumoto2024binary}. 
\end{rem}
 
      We shall refer to it as local RAIC if the localized parameter $\mu_4$ is essentially smaller than the diameter of $\calD$. As  we shall see in Theorem \ref{thm:convergence}, under small $\mu_1$, RAIC implies that PGD linearly converges to estimates with error rate $O(\mu_2+\mu_3)$. Under extremely small $\mu_1$, RAIC  even implies a faster convergence to the same error.

      To keep our theorem statement concise, we pause to elaborate on some technical assumptions used in Theorem \ref{thm:convergence}.

\paragraph{Auxiliary assumptions for Theorem \ref{thm:convergence}:} First, we need to define  
$$\kappa_{\alpha}: = \begin{cases}
    1~,\qquad\text{if }\alpha=0\\
    2~,\qquad \text{if }\alpha >0
\end{cases}$$
to distinguish the nature of a projection onto $\mathbbm{A}_\alpha^\beta$, which is convex when $\alpha=0$ but nonconvex when $\alpha>0$; similar idea appeared in previous analysis of projected gradient descent (e.g., \cite{soltanolkotabi2019structured,oymak2017fast}). Recall that we utilize sequential projections as in (\ref{eq:pgdrule})--(\ref{normproj}) \rev{to resolve a potential computational issue,} while this leads to the new issue that $\bx^{(t)}$ may not live in $\calK$. Our remedy is to define a slightly larger signal set $\overline{\calX}:=(2\calK)\cap \mathbbm{A}_\alpha^\beta$ and show $\bx^{(t)}\in 2\calK$. Taken collectively, we suppose that $(\calQ,\bA,\btau,\calK,\eta)$ respects $(\overline{\calX},\bm{\mu})$-RAIC for some $\bm{\mu}=(\mu_1,\mu_2,\mu_3,\mu_4)$ at some scale $\phi$ obeying $\phi\le \mu_3$, where the non-negative $(\mu_i)_{i=1}^4$ satisfy $\mu_1<\frac{1}{2\kappa_\alpha}$ and 
\begin{align}
\label{barkappa}\underbrace{\left(\frac{\kappa_\alpha\sqrt{\mu_2}+\sqrt{\kappa_\alpha^2\mu_2+2\kappa_\alpha(1-2\mu_1\kappa_\alpha)\mu_3}}{1-2\mu_1\kappa_\alpha}\right)^2}_{:=\bar{\kappa}}<\mu_4.
\end{align}
 If $\calK$ is not a cone and $\alpha>0$,
we suppose
        \begin{align}
            \mu_1\mu_4+\sqrt{\mu_2\mu_4}+\mu_3\le \frac{\alpha}{4}.\label{eq:addit}
        \end{align}
        to ensure $\bx^{(t)}\in 2\calK$. 

\begin{theorem}\label{thm:convergence} 
  Under the above assumptions, for any $\bx\in\calX=\calK\cap \mathbbm{A}_\alpha^\beta$, we run  Algorithm \ref{alg:pgd} with step size $\eta$ and some $\bx^{(0)}\in \mathbb{B}_2^n(\bx;\mu_4)\cap \calX$ to obtain $\{\bx^{(t)}\}_{t=0}^\infty$. 
    Then we have the following statements: 
    \begin{itemize}
        \item  For all $\bx\in\calX$ and any $t\ge 0$, we have 
        \begin{align}\label{raic2linear}
            \|\bx^{(t)}-\bx\|_2 \le \Big(\kappa_\alpha\mu_1+\frac{1}{2}\Big)^t \mu_4 + 4\kappa_\alpha \left[\frac{\kappa_\alpha\mu_2}{(1-2\kappa_\alpha\mu_1)^2}+\frac{\mu_3}{1-2\kappa_\alpha\mu_1}\right]. 
        \end{align}

        \item If additionally $\mu_1 \le\sqrt{\frac{\mu_2}{\mu_4}}+\frac{\mu_3}{\mu_4}$, we define $\hat{\kappa}:= 4(\kappa_\alpha\sqrt{\mu_2}+\sqrt{\kappa_\alpha^2\mu_2+\kappa_\alpha\mu_3})^2$ and assume $\mu_4\ge \hat{\kappa}$. Then for all $\bx\in\calX$ and any $t\ge 0$, we have 
        \begin{align}\label{raic2quadratic}
            \|\bx^{(t)}-\bx\|_2 \le \mu_4^{2^{-t}}\hat{\kappa}^{1-2^{-t}}.  
        \end{align}
    \end{itemize}
\end{theorem}
\begin{proof}
     The proof can be found in Appendix \ref{app:convergence}.
\end{proof}
\rev{The main assumption of Theorem \ref{thm:convergence} is the RAIC. As we illustrated, RAIC characterizes how close the grdient is to the ideal descent step. Hence, it may not be surprising that RAIC implies the convergence of PGD. The main idea of the proof is to show that, under RAIC, the errors of the iterates satisfy certain recurrence inequality, which in general leads to linear convergence as in (\ref{raic2linear}) but may also yield a faster convergence as in (\ref{raic2quadratic}) under sufficiently small $\mu_1.$} 

\begin{rem}
    \rev{Although $\bx^{(0)}\in \mathbb{B}_2^n(\bx;\mu_4)$ is required, Theorem \ref{thm:convergence} yields global convergence results for some models where we can choose $\mu_4$ larger than the diameter of $\calX$. More specifically, if $\calX\subset \mathbb{A}_\alpha^\beta$ and $\mu_4=2\beta$, then PGD with any $\bx_0\in \calX$ enjoys the above convergence guarantee.}
\end{rem}

\subsection{A Sharp Approach to RAIC}\label{sec:sharp_ana}
In light of Theorem \ref{thm:convergence}, it suffices to focus on controlling $\phi^{-1}\| \bu-\bv-\eta\cdot\bh(\bu,\bv)\|_{\calK^\circ_{(\phi)}}$  for all $\bu,\bv\in 2\calK\cap \mathbbm{A}_{\alpha}^\beta$ satisfying $\|\bu-\bv\|_2\le\mu_4$.  We will use the shorthand $\overline{\calX}:=2\calK\cap \mathbbm{A}_\alpha^\beta$. In this subsection, we present a sharp approach to RAIC through generalizing techniques in \cite{matsumoto2024binary,chen2024one} with some new developments, \rev{including a clipped gradient in (\ref{eq:hathpq}) for incorporating the multi-bit case, the elimination of the small-distance regime, and the more amenable and interpretable moment bounds in Assumptions \ref{assum_sg1}--\ref{assum_sg3}. We shall elucidate on these in the sequel.}

The overall framework is a covering argument. For some $$r\in \Big(0,\frac{\Delta\vee\Lambda}{4}\wedge \frac{\mu_4}{2}\Big),$$ we  let $\calN_r$ be a minimal $r$-net of $\overline{\calX}$ such that $\log|\calN_r|=\scrH(\overline{\calX},r)$. Then, for any $\bu,\bv\in \overline{\calX}$ obeying $\|\bu-\bv\|_2\le \mu_4$, we find their closest points in $\calN_r$:
 \begin{align}\label{eq:findu1v1}
    \bu_1:=\text{arg}\min_{\bw\in\calN_r}\|\bw-\bu\|_2~\text{~and~}~\bv_1:=\text{arg}\min_{\bw\in\calN_r}\|\bw-\bv\|_2.
\end{align}
It follows that $\|\bu_1-\bv_1\|_2\le\|\bu_1-\bu\|_2+\|\bu-\bv\|_2+\|\bv-\bv_1\|_2\le\mu_4+2r<2\mu_4.$ We will ultimately achieve RAIC  with $\mu_3=\widetilde{O}(r)$, so we simply choose $\phi=r$ and seek to control $\frac{1}{r}\|\bu-\bv-\eta\cdot \bh(\bu,\bv)\|_{\calK_{(r)}^\circ}$ (see Remark \ref{rem:suffraic}).

We define the constraint sets 
    \begin{gather*}
        \calD^{(2)}_r :=\{(\bp,\bq)\in\overline{\calX}:\|\bp-\bq\|_2\le r\},\\
        \calN^{(2)}_{r,\mu_4}:=\{(\bp,\bq)\in\calN_r\times\calN_r: 0<\|\bp-\bq\|_2\le 2\mu_4\}. 
    \end{gather*}
    Observing  from (\ref{eq:huvdef1}) the following
    \begin{align*}
        &\bh(\bu_1,\bv_1)-\bh(\bu,\bv)\\
        &= \frac{1}{m}\sum_{i=1}^m \big[\calQ(\ba_i^\top\bu-\tau_i)-\calQ(\ba_i^\top\bv-\tau_i)-\calQ(\ba_i^\top\bu_1-\tau_i)+\calQ(\ba_i^\top\bv_1-\tau_i)\big]\ba_i  
        \\
        &=\bh(\bu_1,\bu)+\bh(\bv,\bv_1)
    \end{align*}
    we can proceed as 
         \begin{align}
        &r^{-1}\big\|\bu-\bv-\eta\cdot\bh(\bu,\bv)\big\|_{\calK^\circ_{(r)}} \nn\\&\le  r^{-1}\big\|\bu_1-\bv_1-\eta\cdot\bh(\bu_1,\bv_1)\big\|_{\calK^\circ_{(r)}} + r^{-1}\|\bu-\bu_1\|_{\calK_{(r)}^\circ}+r^{-1}\|\bv-\bv_1\|_{\calK_{(r)}^\circ}\nn \\& \qquad+\frac{\eta}{r} \Big(\|\bh(\bu_1,\bu)\|_{\calK_{(r)}^\circ}+ \|\bh(\bv,\bv_1)\|_{\calK_{(r)}^\circ}\Big)
       \nn
        \\
        &\le r^{-1}\big\|\bu_1-\bv_1-\eta\cdot\bh(\bu_1,\bv_1)\big\|_{\calK_{(r)}^\circ} +2 r+\frac{2\eta}{r}\cdot \sup_{(\bp,\bq)\in\calD^{(2)}_r}\big\|\bh(\bp,\bq)\big\|_{\calK_{(r)}^\circ},\label{eq:large_decompose}
    \end{align}
    where in the last line we observe $(\bu_1,\bu),(\bv,\bv_1)\in\calD_{r}^{(2)}$ and take the supremum over them. 
Note that $\|\bu_1-\bv_1\|_2\le 2\mu_4$, hence either $\bu_1=\bv_1$ or $(\bu_1,\bv_1)\in \calN^{(2)}_{r,\mu_4}$ holds. Therefore, bounding $r^{-1}\|\bu_1-\bv_1-\eta\cdot\bh(\bu_1,\bv_1)\|_{\calK_{(r)}^\circ}$ amounts to uniformly controlling $r^{-1}\|\bp-\bq-\eta\cdot\bh(\bp,\bq)\|_{\calK^\circ_{(r)}}$ over the discrete constraint set $\calN^{(2)}_{r,\mu_4}$. 
In view of the third term in (\ref{eq:large_decompose}), another task is to bound $\|\bh(\bp,\bq)\|_{\calK_{(r)}^\circ}$ over $(\bp,\bq)\in \calD_r^{(2)}$.

\subsubsection{Gradient clipping}
 Since $\calQ$ is non-decreasing, we can write
    \begin{align}\nn
    \bh(\bp,\bq)& = \frac{1}{m}\sum_{i=1}^m \sign(\ba_i^\top(\bp-\bq))|\calQ(\ba_i^\top\bp-\tau_i)-\calQ(\ba_i^\top\bq-\tau_i)|\cdot\ba_i\\
    & = \frac{1}{m}\sum_{i\in\bR_{\bp,\bq}} \sign(\ba_i^\top(\bp-\bq))|\calQ(\ba_i^\top\bp-\tau_i)-\calQ(\ba_i^\top\bq-\tau_i)|\cdot\ba_i,\label{eq:hpq}
\end{align} 
where in the second equality we define the index set  
\begin{align*}
    \bR_{\bp,\bq} := \big\{i\in[m]:\calQ(\ba_i^\top\bp-\tau_i)\ne\calQ(\ba_i^\top\bq-\tau_i)\big\}.
\end{align*}  
For convenience, we further introduce $E^{(i)}_{\bp,\bq}$ to denote the event that $\bp$ and $\bq$ have different $i$-th measurements
$$E^{(i)}_{\bp,\bq}:= \{i\in\bR_{\bp,\bq}\}.$$

In the 1-bit case, we have $|\calQ(\ba_i^\top\bp-\tau_i)-\calQ(\ba_i^\top\bq-\tau_i)|=\Delta$ for $i\in\bR_{\bp,\bq}$, hence $\bh(\bp,\bq)$ reduces to 
\begin{align*}
    \bh(\bp,\bq) = \frac{\Delta}{m}\sum_{i\in\bR_{\bp,\bq}}\sign(\ba_i^\top(\bp-\bq))\ba_i.
\end{align*}
This is not true for the multi-bit setting where $|\calQ(\ba_i^\top\bp-\tau_i)-\calQ(\ba_i^\top\bq-\tau_i)|$, if being non-zero, can take values in $\{\ell\Delta:\ell=1,2,\cdots,L-1\}$. The multiple values   greatly  add to the difficulty  in deriving   sharp concentration bounds for   $r^{-1}\| \bp-\bq-\eta\cdot\bh(\bp,\bq)\|_{\calK_{(r)}^\circ}$ because the moments of the marginals of $ |\calQ(\ba_i^\top\bp-\tau_i)-\calQ(\ba_i^\top\bq-\tau_i)|\ba_i$ could have distinct behaviors in different regimes.\footnote{In particular, $|\calQ(\ba_i^\top\bp-\tau_i)-\calQ(\ba_i^\top\bq-\tau_i)|$ behaves similarly to $\Delta\cdot\mathbbm{1}(i\in\bR_{\bp,\bq})$ if $\|\bp-\bq\|_2\ll \Delta$, while it is more comparable to $|\ba_i^\top(\bp-\bq)|$ if $\|\bp-\bq\|_2\gg \Delta$.} This issue also precludes a unified treatment to the multi-bit case and 1-bit case and increases the difficulty in formulating the forthcoming Assumptions \ref{assum_sg1}--\ref{assum_sg3}.

Our remedy is to {\it clip}  the gradient $\bh(\bp,\bq)$ to
\begin{align}\label{eq:hathpq}
    \hat{\bh}(\bp,\bq):= \frac{\Delta}{m}\sum_{i\in\bR_{\bp,\bq}}\sign(\ba_i^\top(\bp-\bq))\ba_i 
\end{align}
and then work with $\hat{\bh}(\bp,\bq)$.
Nonetheless, this induces some deviation that we need to separately treat; in particular, triangle inequality gives 
\begin{subequations}
    \begin{gather}\label{eq:clipping1}
    r^{-1}\| \bp-\bq-\eta \cdot\bh(\bp,\bq)\|_{\calK_{(r)}^\circ}\le r^{-1}\|\bp-\bq-\eta\cdot \hat{\bh}(\bp,\bq)\|_{\calK_{(r)}^\circ}+\frac{\eta}{r}\|\bh(\bp,\bq)-\hat{\bh}(\bp,\bq)\|_{\calK_{(r)}^\circ},\\
    \frac{\eta}{r} \cdot  \|\bh(\bp,\bq)\|_{\calK_{(r)}^\circ} \le \frac{\eta}{r}\cdot  \|\hat{\bh}(\bp,\bq)\|_{\calK^\circ_{(r)}} + \frac{\eta}{r}\|\bh(\bp,\bq)-\hat{\bh}(\bp,\bq)\|_{\calK^\circ_{(r)}},\label{eq:clipping2}
\end{gather}
\end{subequations}
where $\frac{\eta}{r}\| \bh(\bp,\bq)-\hat{\bh}(\bp,\bq)\|_{\calK_{(r)}^\circ}$ captures the deviation.

Our hope is that this deviation term is uniformly  small over $\calN_{r,\mu_4}^{(2)}$ and $\calD_r^{(2)}$, which ensures that the gradient clipping has minimal impact on our goals of bounding $r^{-1}\|\bp-\bq-\eta\cdot \bh(\bp,\bq)\|_{\calK_{(r)}^\circ}$ over $\calN_{r,\mu_4}^{(2)}$ and $\frac{\eta}{r}\|\bh(\bp,\bq)\|_{\calK^\circ_{(r)}}$ over $\calD_{(r)}^{(2)}$. Intuitively, this is plausible if $\mu_4$ for RAIC is much smaller than  $\Delta$. For instance, in DMbCS (\ref{eq:mbsatu}), $|\calQ_{\delta,L}(\ba_i^\top\bp-\tau_i)-\calQ_{\delta,L}(\ba_i^\top\bq-\tau_i)|\ge 2\delta$ implies $|\ba_i^\top(\bp-\bq)|\ge\delta$, which is evidently     a rare event for $\|\bp-\bq\|_2\le \mu_4\ll \delta$ under sub-Gaussian $\ba_i$; in other words, if $i\in\bR_{\bp,\bq}$ and $\|\bp-\bq\|_2\ll \delta$, then $|\calQ(\ba_i^\top\bp-\tau_i)-\calQ(\ba_i^\top\bq-\tau_i)|=\delta$ holds with very high probability.

More precisely, we need the deviation term to be well controlled as in the following assumption. 

\begin{assumption}[Deviation of gradient clipping]
    \label{assump4}
    For some $c^{(3)},c^{(4)}$ and some  small enough $\varepsilon^{(1)}$, 
    \begin{subequations}
        \begin{gather}\label{eq:boundclip1}
        \frac{\eta}{r}\|\bh(\bp,\bq)-\hat{\bh}(\bp,\bq)\|_{\calK_{(r)}^\circ}\le  \frac{\eta\Delta[\varepsilon^{(1)}\|\bp-\bq\|_2 + c^{(3)}r]}{\Delta\vee\Lambda},~~\forall(\bp,\bq)\in \calN_{r,\mu_4}^{(2)},\\\label{eq:boundclip2}
      \frac{\eta}{r}\| \bh(\bp,\bq)-\hat{\bh}(\bp,\bq)\|_{\calK_{(r)}^\circ}\le \frac{\eta\Delta c^{(4)}r}{\Delta\vee\Lambda},~~\forall (\bp,\bq)\in \calD_r^{(2)}.
    \end{gather}\label{eq:deviterms}
    \end{subequations}
\end{assumption}

\paragraph{Bound the first term on the right-hand side of (\ref{eq:clipping2}):}
The dual norm of the clipped gradient $r^{-1}\|\hat{\bh}(\bp,\bq)\|_{\calK_{(r)}^\circ}$  can be uniformly controlled over $(\bp,\bq)\in\calD_r^{(2)}$ by generic arguments: given that $\|\bp-\bq\|_2\le r$, we first invoke  $E_{\rm small}$ in Theorem \ref{thm:local_embed} to uniformly control the number of contributors and then apply Lemma \ref{lem:max_ell_sum}. This type of arguments will be invoked multiple times in subsequent developments. 

\begin{pro}[Bounding $r^{-1}\|\hat{\bh}(\bp,\bq)\|_{\calK_{(r)}^\circ}$]
\label{pro:boundsmall} 
    Under Assumptions \ref{assum1}--\ref{assum2},  for some constants $c_0,C_1,c_2,C_3$ depending on $c^{(1)}$, if   $r\le c_0 (\Delta\vee\Lambda) $ with sufficiently small $c_0$,  and it holds that 
    \begin{align}
        \label{eq:sample_pro1}
        m \ge C_1(\Delta\vee\Lambda)\left(\frac{\omega^2(\calK_{(r)})}{r^3\log^{3/2}(\frac{\Delta\vee\Lambda}{r})}+\frac{\scrH(\overline{\calX},2r)}{r\log^{1/2}(\frac{\Delta\vee\Lambda}{r})}\right), 
    \end{align}
    then with probability at least $1-\exp(-\frac{c_2mr}{\Delta\vee\Lambda})$ we have 
    \begin{align*}
        \frac{1}{r}\| \hat{\bh}(\bp,\bq)\|_{\calK^\circ_{(r)}}\le \frac{C_3\Delta r}{\Delta\vee\Lambda} \log^{1/2}\Big(\frac{\Delta\vee\Lambda}{r}\Big),~~\forall(\bp,\bq)\in\calD_r^{(2)}.
    \end{align*}
\end{pro}
\begin{proof}
  See the complete proof in Appendix \ref{app:prove_pro1}. 
\end{proof}
\subsubsection{Orthogonal decomposition} \label{orthosubsec}
All that remains is to control $r^{-1}\|\bp-\bq-\eta\cdot \hat{\bh}(\bp,\bq)\|_{\calK_{(r)}^\circ}$ over $(\bp,\bq)\in \calN_{r,\mu_4}^{(2)}$. Note that $\calN_{r,\mu_4}^{(2)}$ is a finite set, we can accomplish the goal by first bounding it for a fixed pair $(\bp,\bq)$ and then taking a union bound over $\calN_{r,\mu_4}^{(2)}$.

For   fixed $(\bp,\bq)\in\overline{\calX}\times \overline{\calX}$ obeying $0<\|\bp-\bq\|_2\le 2\mu_4$, we utilize Lemma \ref{lem:parameterization} to establish the following useful parameterization: there exists some orthonormal $(\bbeta_1:=\frac{\bp-\bq}{\|\bp-\bq\|_2},\bbeta_2)$ such that\footnote{Note that $\bbeta_1,\bbeta_2$ here and some subsequent notation depend on $(\bp,\bq)$, while we drop such dependence   for succinctness.} 
\begin{align}\label{eq:parapq}
    \bp = u_1\bbeta_1+u_2\bbeta_2~\text{ and }~\bq = v_1\bbeta_1+u_2\bbeta_2
\end{align}
 for some coordinates $(u_1,u_2,v_1)$ obeying $u_1>v_1$ and $u_2\ge 0$. 
 Now we have the following orthogonal decomposition of the clipped gradient
\begin{align}\label{eq:ortho_decomp}
     \hat{\bh}(\bp,\bq)= \langle\hat{\bh}(\bp,\bq),\bbeta_1\rangle\bbeta_1+\langle \hat{\bh}(\bp,\bq),\bbeta_2\rangle\bbeta_2 + [\underbrace{\hat{\bh}(\bp,\bq)-\langle\hat{\bh}(\bp,\bq),\bbeta_1\rangle\bbeta_1-\langle \hat{\bh},(\bp,\bq),\bbeta_2\rangle\bbeta_2}_{:=\hat{\bh}^\bot(\bp,\bq)}]
\end{align}
which allows us to decompose the dual norm as 
    \begin{align}
    &\frac{1}{r}\big\|\bp-\bq-\eta\cdot\hat{\bh}(\bp,\bq)\big\|_{\calK_{(r)}^\circ} \nn\\
    &= \frac{1}{r}\sup_{\bw\in \calK_{(r)}} \bw^\top\left[(\bp-\bq-\eta\cdot\langle\hat{\bh}(\bp,\bq),\bbeta_1\rangle\bbeta_1)-(\eta\cdot \langle\hat{\bh}(\bp,\bq),\bbeta_2\rangle\bbeta_2)-(\eta\cdot \hat{\bh}^\bot(\bp,\bq))\right]\nn\\
    &\le  \big|\|\bp-\bq\|_2- \eta\cdot \langle\hat{\bh}(\bp,\bq),\bbeta_1\rangle\big|  + \big|\eta \cdot \langle\hat{\bh}(\bp,\bq),\bbeta_2\rangle\big| + r^{-1}\cdot\sup_{\bw\in\calK_{(r)}}\eta \cdot\bw^\top \hat{\bh}^\bot(\bp,\bq)\nn\\
    :&= \big|\|\bp-\bq\|_2- T_1^{\bp,\bq}\big|+  |T_2^{\bp,\bq}|+ r^{-1} \cdot T_3^{\bp,\bq}.\label{eq:large_further_decom}
\end{align}  
 We seek to control the three terms in (\ref{eq:large_further_decom}).

\subsubsection{Sharp concentration bounds}
   Substituting $\hat{\bh}(\bp,\bq)=\frac{\Delta}{m}\sum_{i\in\bR_{\bp,\bq}}\sign(\ba_i^\top\bbeta_1)\ba_i$, we shall readily find
 \begin{gather*}
  T_1^{\bp,\bq}:=\eta \cdot \langle \hat{\bh}(\bp,\bq),\bbeta_1\rangle =  \frac{\eta\Delta}{m}\sum_{i=1}^m|\ba_i^\top\bbeta_1|\mathbbm{1}(E^{(i)}_{\bp,\bq}),\\
    T_2^{\bp,\bq}:=\eta\cdot\langle\hat{\bh}(\bp,\bq),\bbeta_2\rangle  =  \frac{\eta\Delta}{m}\sum_{i=1}^m\sign(\ba_i^\top\bbeta_1)(\ba_i^\top\bbeta_2) \mathbbm{1}(E^{(i)}_{\bp,\bq}),\\
   T_3^{\bp,\bq}:=\sup_{\bw\in\calK_{(r)}}\eta\cdot\bw^\top\hat{\bh}^\bot(\bp,\bq)= \sup_{\bw\in\calK_{(r)}}\frac{\eta\Delta}{m}\sum_{i\in\bR_{\bp,\bq}}\sign(\ba_i^\top\bbeta_1)\big[\ba_i-(\ba_i^\top\bbeta_1)\bbeta_1-(\ba_i^\top\bbeta_2)\bbeta_2\big]^\top\bw. \nn 
\end{gather*} 
Recall our convention $E^{(i)}_{\bp,\bq}=\{i\in\bR_{\bp,\bq}\}$.

We seek  sharp concentration bounds for $(T_i^{\bp,\bq})_{i=1}^3$. The core ideas are to proceed with arguments aware of the distance between $\bp$ and $\bq$. For instance, if $\bp$ and $\bq$ are close enough, the number of effective contributors to $T_i^{\bp,\bq}$ (i.e.,  $|\bR_{\bp,\bq}|$) will be much fewer than $m$, thus one can hope for tighter concentration bounds. To that end, our particular strategy for bounding $T_3^{\bp,\bq}$ is to  first establish the concentration bound by conditioning on $|\bR_{\bp,\bq}|$, and then further get rid of the conditioning by analyzing $|\bR_{\bp,\bq}|$, a binomial variable, via Chernoff bound. To establish the conditional concentration bound, we shall show the  sub-Gaussianity of the conditional distribution $[\ba_i-(\ba_i^\top\bbeta_1)\bbeta_1-(\ba_i^\top\bbeta_2)\bbeta_2]^\top\bw|\{i\in\bR_{\bp,\bq}\}$ and then use Lemma \ref{lem:tala} to obtain a bound relevant to the Gaussian width of $\calK_{(r)}$.  Similar idea was used in \cite{chen2024one,matsumoto2024robust}. On the other hand, for the simpler terms $T_1^{\bp,\bq}$ and $T_2^{\bp,\bq}$, we find that bounding the moments as in (\ref{eq:sg1_sg}) and (\ref{eq:sg2_sg}) and then using Bernstein's inequality (see Lemma \ref{lem:chernoff}) is sufficient.

  We make the following assumptions on these matters. 
  

\begin{assumption}\label{assum_sg1}
    For any $\bp,\bq\in \overline{\calX}$ obeying $0<\|\bp-\bq\|_2\le2\mu_4$, we suppose that
    \begin{subequations}
        \begin{gather}\label{eq:sg1_sg}
        \mathbbm{E}\Big(|\ba_i^\top\bbeta_1|^p\mathbbm{1}(E^{(i)}_{\bp,\bq})\Big)\le c^{(5)}\mathsf{P}_{\bp,\bq}\cdot \frac{p!}{2},\quad\forall p\ge 2\\
       \left|\mathbbm{E}\big(|\ba_i^\top\bbeta_1|\mathbbm{1}(E^{(i)}_{\bp,\bq})\big) - \frac{c^{(6)}\|\bp-\bq\|_2}{\Delta\vee\Lambda}\right|\le \frac{\varepsilon^{(2)}\|\bp-\bq\|_2}{\Delta\vee\Lambda},  \label{eq:sg1_exp}
    \end{gather}\label{eq:sg1_todo}
    \end{subequations}
    hold for some $c^{(5)},c^{(6)}>0$, $\varepsilon^{(2)}\ge 0$. 
\end{assumption}
\begin{assumption}\label{assum_sg2}
    For any $\bp,\bq\in \overline{\calX}$ obeying $0<\|\bp-\bq\|_2\le 2\mu_4$, we suppose that
    \begin{subequations}
        \begin{gather}\label{eq:sg2_sg}
        \mathbbm{E}\Big(|\ba_i^\top\bbeta_2|^p\mathbbm{1}(E^{(i)}_{\bp,\bq})\Big)\le  c^{(7)}\mathsf{P}_{\bp,\bq}\cdot \frac{p!}{2},\quad\forall p\ge 2\\\label{eq:sg2_exp}
         \Big|\mathbbm{E}\Big(\sign(\ba_i^\top\bbeta_1)\ba_i^\top\bbeta_2\mathbbm{1}(E^{(i)}_{\bp,\bq})\Big)\Big|\le \frac{\varepsilon^{(3)}\|\bp-\bq\|_2}{\Delta\vee\Lambda},
    \end{gather}
    \end{subequations}
    hold for some $c^{(7)}>0$, $\varepsilon^{(3)}\ge 0$.
\end{assumption}
For any $\bw\in \mathbb{R}^{n}$, we define $$J_{i,\bw}^{\bp,\bq}:=\sign(\ba_i^\top\bbeta_1)[\ba_i-(\ba_i^\top\bbeta_1)\bbeta_1-(\ba_i^\top\bbeta_2)\bbeta_2]^\top\bw.$$
\begin{assumption}\label{assum_sg3}
For any $\bp,\bq\in \overline{\calX}$ obeying $0<\|\bp-\bq\|_2\le 2\mu_4$, we suppose that 
   \begin{subequations}
        \begin{gather}\label{eq:sg3_sg}
         \mathbbm{E}\Big(|J_{i,\bw}^{\bp,\bq}|^p\mathbbm{1}(E^{(i)}_{\bp,\bq})\Big)\le \sfP_{\bp,\bq} (c^{(8)}\sqrt{p})^p,\quad\forall p\ge 2,~~\forall\bw\in \mathbb{S}^{n-1}\\\label{eq:sg3_exp}
       \Big|\mathbbm{E}\Big(J_{i,\bw}^{\bp,\bq}\mathbbm{1}(E^{(i)}_{\bp,\bq})\Big)\Big|\le \frac{\varepsilon^{(4)}\|\bp-\bq\|_2}{\Delta\vee\Lambda},~~ \forall\bw\in\mathbb{S}^{n-1}
    \end{gather}
   \end{subequations}
    hold for some $c^{(8)}>0$, $\varepsilon^{(4)}\ge 0$. 
\end{assumption}

As mentioned, $\varepsilon^{(i)}$'s denote sufficiently small constants, thus (\ref{eq:sg1_exp}), (\ref{eq:sg2_exp}) and (\ref{eq:sg3_exp}) precisely characterize the mean of $T_i^{\bp,\bq}$; for instance, (\ref{eq:sg1_exp}) essentially characterize 
$$\mathbbm{E}[T_1^{\bp,\bq}]=\eta\Delta \mathbbm{E}\big[|\ba_i^\top\bbeta_1|\mathbbm{1}(E^{(i)}_{\bp,\bq})\big]\approx \frac{c^{(6)}\eta\Delta\|\bp-\bq\|_2}{\Delta\vee\Lambda}.$$ These are relevant in the selection of the step size $\eta$   to ensure contraction. While we can validate (\ref{eq:sg2_exp}) and (\ref{eq:sg3_exp}) for 1bCS, D1bCS and DMbCS,   meaning that $T_2^{\bp,\bq}$ and $T_3^{\bp,\bq}$ nearly concentrate about $0$, in some more intricate problem this may not be true and we will need a more general condition as in (\ref{eq:sg1_exp}). See one-bit phase retrieval \cite{chen2024one} for instance.


The moment bounds (\ref{eq:sg1_sg}), (\ref{eq:sg2_sg}) and (\ref{eq:sg3_sg}) are more interesting and, in some sense, convey a type of {\it independence} between the marginals of $\ba_i$ and the separation event $E^{(i)}_{\bp,\bq}$. We shall use (\ref{eq:sg1_sg}) to elaborate on this. In fact, the most natural way to bound $ \mathbbm{E}(|\ba_i^\top\bbeta_i|^p\mathbbm{1}(E^{(i)}_{\bp,\bq}))$ is via Cauchy–Schwarz inequality: 
  \begin{align}\label{csachieve}
      \mathbbm{E}\Big(|\ba_i^\top\bbeta_i|^p\mathbbm{1}(E^{(i)}_{\bp,\bq})\Big)\le \Big(\mathbbm{E}\big[|\ba_i^\top\bbeta_1|^{2p}\big]\mathbbm{P}(E^{(i)}_{\bp,\bq})\Big)^{1/2}\le \sqrt{\sfP_{\bp,\bq}}(c_1\sqrt{p})^p\le c_2\sqrt{\sfP_{\bp,\bq}}\cdot\frac{p!}{2}
  \end{align}
  for some absolute constants $c_1$ and $c_2$.
   However, this is weaker than what we assumed in (\ref{eq:sg1_sg}) if $\sfP_{\bp,\bq}=o(1)$. On the other hand,  if $\ba_i^\top\bbeta_1$ and $E^{(i)}_{\bp,\bq}$ are independent, then we can easily reach (\ref{eq:sg1_sg}) by  
    \begin{align}
      \mathbbm{E}\Big(|\ba_i^\top\bbeta_i|^p\mathbbm{1}(E^{(i)}_{\bp,\bq})\Big)= \sfP_{\bp,\bq}\mathbbm{E}\big[|\ba_i^\top\bbeta_i|^p\big] \le \sfP_{\bp,\bq}(c_1\sqrt{p})^p\le c_2\sfP_{\bp,\bq}\cdot\frac{p!}{2}.
  \end{align}
  Therefore, roughly put, (\ref{eq:sg1_sg}) indicates a type of independence between $\ba_i^\top\bbeta_i$ and $E^{(i)}_{\bp,\bq}$. As we shall see, such   ``independence''  comes from the rotational invariance of Gaussian sensing vector in 1bCS, while is due to the additional randomness of $\btau$ in D1bCS and DMbCS.

  A more precise description of such independence is that ``{\it the event $E_{\bp,\bq}^{(i)}$ does not significantly affect the sub-Gaussianity of $\ba_i^\top\bbeta_1,~\ba_i^\top\bbeta_2$ and $J_{i,\bw}^{\bp,\bq}$}.'' It is easy to note that (\ref{eq:sg1_sg}), (\ref{eq:sg2_sg}) and (\ref{eq:sg3_sg}) are equivalent to 
\begin{subequations}
     \begin{gather}
    \mathbbm{E}\Big(|\ba_i^\top\bbeta_1|^p\big|E^{(i)}_{\bp,\bq}\Big) \le c^{(5)}\cdot\frac{p!}{2},\\
\mathbbm{E}\Big(|\ba_i^\top\bbeta_2|^p\big|E^{(i)}_{\bp,\bq}\Big) \le c^{(7)}\cdot\frac{p!}{2},\\      \mathbbm{E}\Big(|J^{\bp,\bq}_{i,\bw}|^p\big|E^{(i)}_{\bp,\bq}\Big) \le (c^{(8)}\sqrt{p})^p,\quad \forall \bw\in \mathbb{S}^{n-1} 
  \end{gather}
\end{subequations}
  for any $p\ge 2$. As $\|X\|_{\psi_2}\asymp \sup_{p\ge 1}(\mathbbm{E}|X|^p)^{1/p}/\sqrt{p}$ and $\|X\|_{\psi_1}\asymp \sup_{p\ge 1}(\mathbbm{E}|X|^p)^{1/p}/p$ \cite{vershynin2018high}, these conditions essentially characterize the sub-exponential tails of $\ba_i^\top\bbeta_1|E^{(i)}_{\bp,\bq}$ and $\ba_i^\top\bbeta_2|E^{(i)}_{\bp,\bq}$
  \begin{align}
      \big\|\ba_i^\top\bbeta_1|E^{(i)}_{\bp,\bq}\big\|_{\psi_1}=O(c^{(5)})\quad\text{and}\quad \big\|\ba_i^\top\bbeta_2|E^{(i)}_{\bp,\bq}\big\|_{\psi_1}=O(c^{(7)}), \label{conse}
  \end{align}
  and the sub-Gaussian tail of $J_{i,\bw}^{\bp,\bq}|E^{(i)}_{\bp,\bq}$
  \begin{align}
\big\|J_{i,\bw}^{\bp,\bq}|E^{(i)}_{\bp,\bq}\big\|_{\psi_2} = O(c^{(8)}),\quad \forall \bw\in \mathbb{S}^{n-1}.  
      \label{consg}
  \end{align}
  This perspective was 
    employed in \cite{matsumoto2024binary,chen2024one} and will be useful in controlling $T_3^{\bp,\bq}$ in Proposition \ref{pro4}.

With these assumptions in place, we  bound $|\|\bp-\bq\|_2-T_1^{\bp,\bq}|$, $|T_2^{\bp,\bq}|$ and $T_3^{\bp,\bq}$ in the following propositions. 

\begin{pro}[Bounding $|\|\bp-\bq\|_2-T_1^{\bp,\bq}|$]
\label{pro2}
    Under Assumption \ref{assum_sg1},   with probability at least $1-2\exp(-2\scrH(\overline{\calX},r))$, for all $(\bp,\bq)\in\calN_{r,\mu_4}^{(2)}$ we have 
\begin{align*}
    \big|\|\bp-\bq\|_2-T_1^{\bp,\bq}\big| \le 4\eta\Delta\left[ \sqrt{\frac{c^{(5)}\scrH(\overline{\calX},r)\sfP_{\bp,\bq}}{m}}+\frac{\scrH(\overline{\calX},r)}{m}\right]+\left[\Big|1-\frac{\eta c^{(6)} \Delta }{\Delta\vee\Lambda}\Big|+\frac{\eta\Delta \varepsilon^{(2)}}{\Delta\vee\Lambda}\right]\|\bp-\bq\|_2. 
\end{align*} 
\end{pro}
\begin{proof}
    The proof can be found in Appendix \ref{app:pro2}.
\end{proof}
\begin{pro}[Bounding $|T_2^{\bp,\bq}|$]
    \label{pro3}
    Under Assumption  \ref{assum_sg2}, with probability at least $1-2\exp(-2\scrH(\overline{\calX},r))$, for all $(\bp,\bq)\in\calN_{r,\mu_4}^{(2)}$ we have 
    \begin{align*}
        |T_2^{\bp,\bq}|\le 4\eta\Delta\left[ \sqrt{\frac{c^{(7)}\scrH(\overline{\calX},r)\sfP_{\bp,\bq}}{m }}+\frac{\scrH(\overline{\calX},r)}{m}\right]+\frac{\eta\Delta \varepsilon^{(3)}\|\bp-\bq\|_2}{\Delta\vee\Lambda}.
    \end{align*}
\end{pro}
\begin{proof}
    The proof can be found in Appendix \ref{app:prove_pro2}.
\end{proof}
\begin{pro}[Bounding $T_3^{\bp,\bq}$]\label{pro4}
    Under Assumption   \ref{assum_sg3}, 
    with probability at least $1-4\exp(-2\scrH(\overline{\calX},r))$, there exists some absolute constant $C$ such that 
    \begin{align*}
          T_3^{\bp,\bq}& \le Cc^{(8)}r\eta\Delta \sqrt{\Big[\sfP_{\bp,\bq}+\frac{\scrH(\overline{\calX},r)}{m}\Big]\Big[\frac{r^{-2}\omega^2(\calK_{(r)})+\scrH(\overline{\calX},r)}{m}\Big]} + \frac{5r\eta\Delta\varepsilon^{(4)}\|\bp-\bq\|_2}{\Delta\vee\Lambda}\left(1+\frac{\scrH(\overline{\calX},r)}{m\sfP_{\bp,\bq}}\right)
    \end{align*}
    holds for all $(\bp,\bq)\in \calN^{(2)}_{r,\mu_4}$.
\end{pro}
\begin{proof}
     The proof can be found in Appendix \ref{app:pro4}. 
\end{proof}
\begin{rem}[Simplifications of \cite{matsumoto2024binary,chen2024one}] \rev{Our work simplifies   the analysis in \cite{matsumoto2024binary,chen2024one} from some aspects. The main simplification is to provide a unified treatment to all the pairs $(\bp,\bq)\in\calN_{r,\mu_4}^{(2)}$, unlike \cite{matsumoto2024binary,chen2024one} that separately analyzed $(\bp,\bq)\in \calN_r\times \calN_r$ in  the large-distance regime (i.e., $\|\bp-\bq\|_2\ge r$) and the small-distance regime (i.e., $\|\bp-\bq\|_2\le r$). Indeed, the only essential use of $\|\bp-\bq\|_2\ge r$ in their large-distance regime is in instances such as \cite[Equation (92)]{matsumoto2024binary}, where $\|\bp-\bq\|_2$ appears in the probability term and hence it should be suitably large to tolerate a union bound over $\calN_{r,\mu_4}^{(2)}$. To avoid $\|\bp-\bq\|_2\ge r$, we adjust the upper bounds in this part of analysis to ensure a sufficient probability term (see Propositions \ref{pro2}--\ref{pro4}), and it turns out that these slightly different upper bounds, such as $O(\sqrt{\frac{\scrH(\overline{\calX},r)\sfP_{\bp,\bq}}{m}}+\frac{\scrH(\overline{\calX},r)}{m})$ in Proposition \ref{pro2}, are   tight enough for our purpose. Another simplification is that, unlike \cite{matsumoto2024binary,chen2024one} that bound $T_1^{\bp,\bq},T_2^{\bp,\bq}$ by conditioning on $|E^{(i)}_{\bp,\bq}|$ and establishing the sub-Gaussianity of the conditional distribution, we observe that the moment-based Bernstein inequality (Lemma \ref{lem:chernoff}) suffices and therefore impose the moment bounds in (\ref{eq:sg1_sg}) and (\ref{eq:sg2_sg}). These conditions are more convenient to validate and are slightly weaker than conditional sub-Gaussianity; see   (\ref{conse}).} 
\end{rem}
\subsubsection{RAIC and Main Theorem} 
Additionally, we assume an upper bound on $\sfP_{\bu,\bv}$ that complements Assumption \ref{assum3}. 
\begin{assumption}[Upper bound on $\sfP_{\bu,\bv}$]\label{assum_Puvup}
    For some $c^{(9)}>0$, 
    \begin{align*}
        \sfP_{\bu,\bv}\le \frac{ c^{(9)}\|\bu-\bv\|_2}{\Delta\vee\Lambda},\quad \forall\bu,\bv\in \mathbbm{A}_\alpha^\beta .
    \end{align*}
\end{assumption}
Now we can combine all the pieces to obtain our main theorem. We shall pause to present some auxiliary conditions to keep our theorem statement concise.

\paragraph{Auxiliary Assumptions for Theorem \ref{thm:main}:} We will work with Assumptions \ref{assum1}--\ref{assum_Puvup} under    $r>0$ which is a given small enough   positive constant obeying $r \lesssim  \frac{\mu_4}{\log^{1/2}((\Delta\vee\Lambda)/\mu_4)}$.\footnote{This is just a technical assumption to  ensure $\bar{\kappa}\vee\hat{\kappa}<\mu_4$ in Theorem \ref{thm:convergence}.} Recall that $\kappa_{\alpha}=1$ if $\alpha=0$ and $\kappa_\alpha=2$ if $\alpha>0$. We choose the step size $\eta=\frac{\Delta\vee\Lambda}{c^{(6)} \Delta }$ inspired by the bound in Proposition \ref{pro2}. Corresponding to $\mu_1 < \frac{1}{2\kappa_\alpha}$, we assume 
    \begin{align}\label{eq:eta_condition}
\frac{ \sum_{j=1}^3\varepsilon^{(j)}+5\varepsilon^{(4)} }{c^{(6)}}\le \frac{0.49}{\kappa_\alpha}.
    \end{align} 
   If $\calK$ is not a cone and $\alpha>0$, we also assume 
       \begin{align}
\frac{ \sum_{j=1}^3\varepsilon^{(j)}+5\varepsilon^{(4)} }{c^{(6)}}  \le \frac{0.24\alpha}{\mu_4} 
        \label{mu1mu2mu3}
    \end{align}
    which stems from (\ref{eq:addit}) in Theorem \ref{thm:convergence} and ensures $\bx^{(t)}\in2\calK$. Furthermore, we assume the mild condition $\mu_4\le C_0(\Delta\vee\Lambda)$ for some absolute constant $C_0$.  

\begin{theorem}
    [Main Theorem] \label{thm:main}  Under the above assumptions, we suppose that Assumptions \ref{assum1}--\ref{assum_Puvup} hold and $\eta = \frac{\Delta\vee\Lambda}{c^{(6)} \Delta }$. 
    There exist some constants $C_1,c_2,C_3,C_4,C_5,C_6$ depending on  $(c^{(i)})_{i=1}^9$, if 
    \begin{align}\label{eq:sample_main}
        m\ge C_1(\Delta\vee\Lambda)\left(\frac{\omega^2(\calK_{(r)})}{r^3}+\frac{\scrH(\calX,r/2)}{r}\right),
    \end{align}
    then with probability at least $1-\exp(-c_2\scrH(\calX,r))$, we have that $(\calQ,\bA,\btau,\calK,\eta=\frac{\Delta\vee\Lambda}{c^{(6)} \Delta })$ respects $(\overline{\calX},\bm{\mu})$-RAIC at the scale of $r$ with 
        \begin{align}\label{eq:bmmu}
            \bm{\mu} = \left( \frac{ (\sum_{j=1}^3\varepsilon^{(j)}+5\varepsilon^{(4)})}{c^{(6)}},C_3r,C_4r \log^{1/2}\Big(\frac{\Delta\vee\Lambda}{r}\Big),\mu_4\right).
        \end{align} This RAIC implies the following for any $\bx\in\calX$ and the corresponding sequence $\{\bx^{(t)}\}$ produced by running Algorithm \ref{alg:pgd} with initialization $\bx^{(0)}\in \mathbb{B}_2^n(\bx;\mu_4)\cap \calX$ and step size $\eta=\frac{\Delta\vee\Lambda}{c^{(6)} \Delta }$: 
    \begin{itemize}
    [leftmargin=5ex,topsep=0.25ex]
		\setlength\itemsep{-0.1em}
        \item (Linear Convergence)  It holds that, for any integer $t\ge 0$,\footnote{We comment that $0.99$ is in general improvable (e.g., it can often be a much smaller constant in specific models). This won't be pursued in our paper.}
        \begin{align}\label{eq:rate_main}
            \|\bx^{(t)}-\bx\|_2\le 0.99^t\mu_4 + C_5 r\log^{1/2}\Big(\frac{\Delta\vee\Lambda}{r}\Big). 
        \end{align}
        \item \rev{(Faster Convergence)} If  additionally $\mu_1\le \sqrt{\frac{C_3r}{\mu_4}}$ holds, then for any $t\ge 0$,  
        \begin{align}
             \|\bx^{(t)}-\bx\|_2\le \mu_4^{2^{-t}}\left(C_6r\log^{1/2}\Big(\frac{\Delta\vee\Lambda}{r}\Big)\right)^{1-2^{-t}}. 
        \end{align}
    \end{itemize}
\end{theorem}

\begin{proof} 
In this proof, we combine the developments of this section (i.e., Section \ref{sec:sharp_ana}) to establish the RAIC. Then the two convergence guarantees follow directly from  Theorem \ref{thm:convergence}. See the detailed proof in Appendix \ref{app:provemain}. 
\end{proof}
Our main theorem indicates the following takeaway message of our paper.
\begin{rem}
Comparing Theorem \ref{thm:main} and Theorem \ref{thm:itup}, we find that under the further Assumptions \ref{assump4}--\ref{assum_Puvup}, the proposed PGD achieves the same error rate as HDM, up to logarithmic factor.  While HDM is in general computationally intractable, PGD is efficient so long as the projection onto $\calK$ can be executed efficiently, covering some canonical examples such as   $\calK=\Sigma^n_k,~M^{n_1,n_2}_{\bar{r}}$, $\sqrt{k}\mathbb{B}_1^n$ and so on. 
\end{rem}


\subsection{A Complementary Approach}\label{sec:raic_qpe}
In the multi-bit case, the techniques developed in the last subsection only establish RAIC with $\mu_4\ll\Delta$; in fact, if $\|\bu-\bv\|_2\ge\Delta$, then     $\hat{\bh}(\bu,\bv)$ deviates too much from the actual gradient $\bh(\bu,\bv)$. Hence, in general, Assumption \ref{assump4} does not hold.

In this subsection, we develop a complementary approach for the situation where the sharp approach only yields local convergence.   Here, similarly, we introduce $r_1$ as the accuracy we want to achieve;\footnote{This is similar to the role of $r$ for the sharp approach in the last subsection.} hence, as in Remark \ref{rem:suffraic}, we will let $\phi=r_1$ and seek a different way to uniformly control 
\begin{align}
    \frac{1}{r_1}\|\bu-\bv-\eta\cdot\bh(\bu,\bv)\|_{\calK_{(r_1)}^\circ},\quad \bu,\bv\in \overline{\calX}. 
\end{align}

Let us start with a useful perspective for nonlinear observation  (e.g., \cite{plan2016generalized,genzel2016high,plan2017high})---to view it as a properly rescaled linear model under a near-centered non-standard noise. In particular, for some properly chosen $\zeta>0$, we decompose the quantized measurement into a re-scaled linear observation $\zeta\cdot \ba_i^\top\bu$ and an irregular noise  $\xi_i(\bu)$ that hinges on $(\ba_i,\tau_i,\bu)$: 
\begin{align}\label{nonlinearview}
    \calQ(\ba_i^\top\bu -\tau_i) = \zeta \cdot \ba_i^\top\bu + (\underbrace{\calQ(\ba_i^\top\bu-\tau_i)-\zeta \cdot\ba_i^\top\bu}_{:=\xi_i(\bu)}).
\end{align}

Utilizing the decomposition, for any $\bu,\bv\in\overline{\calX}$ we have 
    \begin{align}
    &\frac{1}{r_1}\|\bu-\bv-\eta\cdot\bh(\bu,\bv)\|_{\calK^\circ_{(r_1)}}\nn
   \\&\nn =\frac{1}{r_1}\sup_{\bw\in\calK_{(r_1)}}\Big\langle\bw,\bu-\bv- \frac{\eta}{m}\sum_{i=1}^m\big[\calQ(\ba_i^\top\bu-\tau_i)-\calQ(\ba_i^\top\bv-\tau_i)\big]\ba_i\Big\rangle\\ \nn
    &= \frac{1}{r_1}\sup_{\bw\in \calK_{(r_1)}}\left[\bw^\top\Big(\bI_n - \frac{\eta\zeta}{m}\sum_{i=1}^m\ba_i\ba_i^\top\Big)(\bu-\bv)-\frac{\eta}{m}\sum_{i=1}^m\xi_i(\bu)\ba_i^\top\bw + \frac{\eta}{m}\sum_{i=1}^m\xi_i(\bv)\ba_i^\top\bw\right]\\
    &\le  
\frac{1}{r_1}\sup_{\bw\in\calK_{(r_1)}}\left|\bw^\top \Big(\bI_n-\frac{\eta\zeta}{m}\sum_{i=1}^m \ba_i\ba_i^\top\Big)(\bu-\bv)\right|+2\underbrace{\sup_{\bu\in\overline{\calX}}\sup_{\bw\in\calK_{(r_1)}}\frac{\eta}{mr_1}\sum_{i=1}^m\xi_i(\bu)\ba_i^\top\bw}_{:=\hat{T}}, \label{eq:qpe2raic}
\end{align}

The first term can be controlled uniformly for all $\bu,\bv\in\overline{\calX}$ by the concentration inequality of product process \cite{mendelson2016upper}; see Lemma \ref{lem:product_process}. 

\begin{pro}\label{pro5}
Under Assumption \ref{assum1}, if $m\ge\frac{\omega^2(\calK_{(r_1)})}{r_1^2}$, then there exist some absolute constants $C_1,c_2$, for any $\bu,\bv\in \overline{\calX}$ we have
\begin{align}\label{boundfir}\nn
  &\frac{1}{r_1}\sup_{\bw\in\calK_{(r_1)}}\left|\bw^\top \Big(\bI_n-\frac{\eta\zeta}{m}\sum_{i=1}^m \ba_i\ba_i^\top\Big)(\bu-\bv)\right| \\&\le \left(\frac{C_1|\eta\zeta|\cdot\omega(\calK_{(r_1)})}{r_1\sqrt{m}}+|\eta\zeta-1|\right)\cdot\|\bu-\bv\|_2+ \frac{C_1|\eta\zeta|\omega(\calK_{(r_1)})}{\sqrt{m}}
\end{align}
with probability at least $1-\exp(-\frac{c_2\omega^2(\calK_{(r_1)})}{r_1^2})$.
\end{pro}
\begin{proof}
    See its proof in  Appendix \ref{app:pro5}. 
\end{proof}

In contrast, it takes more technicalities to bound $\hat{T}$ which accounts for the quantization. Indeed, in the literature, the property of $\hat{T}$ being well controlled has proven a major ingredient in the analysis for many other estimators, and is referred to as (global) quantized product embedding (see \cite{chen2024uniform} for $\calQ_\delta$), limited projection distortion (see \cite{xu2020quantized} for $\calQ_\delta$ or even general non-linearity) or sign product embedding (see \cite{plan2012robust} for 1bCS).

\paragraph{Auxiliary Assumptions for Theorem \ref{thm:conver_qpe}:}  As before, we state some secondary assumptions here to ensure cleaner theorem statement.  With {\it arbitrary initialization}, we let  $\{\bx^{(t)}\}$ be the sequence produced by running Algorithm \ref{alg:pgd} with initialization $\bx^{(0)}\in \calX$ and step size $\eta$. In a specific nonlinear model, $\zeta$ in (\ref{nonlinearview}) is often known, and {\it we simply choose $\eta=\zeta^{-1}$ so that $\eta\zeta =1$ and the bound (\ref{boundfir}) is ``minimized''}. If $\calK$ is not a cone and $\alpha>0$, we  additionally  assume $3r_1\le 0.24\alpha$ to ensure (\ref{eq:addit}). Under these assumptions, we have the following guarantee once $\hat{T}$ defined in (\ref{eq:qpe2raic}) is bounded by $r_1$.

\begin{theorem}[Complementary approach] 
\label{thm:conver_qpe}  Under the above setting and Assumption \ref{assum1}, we assume for some small enough given constant $r_1$ and 
$\hat{T}$ defined in (\ref{eq:qpe2raic}) is bounded as $$\hat{T}\le  r_1.$$    
For some   constants $C_1,c_2,c_3,C_4$ that may depend on $(\alpha,\beta)$, if $$m\ge \frac{C_1\omega^2(\calK_{(r_1)})}{r_1^2}$$ and $  r_1\le c_2\beta$, then with probability at least $1-2\exp(-c_3\omega^2(\calK))$, $(\calQ,\bA,\btau,\calK,\eta)$ respects $(\overline{\calX},\bm{\mu})$-RAIC with $\bm{\mu}=(r_1,0,2r_1,2\beta)$, and 
for any $\bx\in\calX$ and any $t\ge 0$ we have    
  \begin{align*}
      \|\bx^{(t)}-\bx\|_2 \le  \left(\frac{C_4\omega(\calK_{(r_1)})}{r_1\sqrt{m}}\right)^t\beta  +24r_1. 
  \end{align*}
\end{theorem}
\begin{proof}
See the detailed proof in Appendix \ref{app:provecomplement}. 
\end{proof}

\section{Instantiation for Specific Models}\label{sec:instance}
In this section,  we validate   Assumptions \ref{assum1}--\ref{assum_Puvup} to establish the optimality of PGD for 1bCS, D1bCS and DMbCS. 
For a specific model, the validation of Assumptions \ref{assum_sg1}--\ref{assum_sg3} consists of elementary probability and moment estimations. It will often be convenient to work with the orthonormal $(\bbeta_1,\bbeta_2)$ such that $\bp=u_1\bbeta_1+u_2\bbeta_2$ and $\bq=v_1\bbeta_1+u_2\bbeta_2$ for some $u_1>v_1$ and $u_2\ge 0$; also note that $\|\bp-\bq\|_2 = u_1-v_1$; see Section \ref{orthosubsec}. This enables the expression 
\begin{align}\label{Eib1b2}
    E^{(i)}_{\bp,\bq}=\{i\in\bR_{\bp,\bq}\} = \big\{\calQ(u_1\ba_i^\top\bbeta_1+u_2\ba_i^\top\bbeta_2 -\tau_i)\ne \calQ(v_1\ba_i^\top\bbeta_1+u_2\ba_i^\top\bbeta_2 -\tau_i)\big\}.
\end{align}

Observe that we can simply focus on sufficiently small $\sfP_{\bp,\bq}$ when we seek to validate (\ref{eq:sg1_sg}), (\ref{eq:sg2_sg}) and (\ref{eq:sg3_sg}). This is because the cases when $\sfP_{\bp,\bq}\ge c$ (where $c$ is any given small absolute constant) can be addressed by Cauchy–Schwarz inequality; see (\ref{csachieve}) for instance. Therefore, upon validating Assumption \ref{assum3} for some absolute constant $c^{(2)}$, we can restrict our attention to $\|\bp-\bq\|_2 \le c(\Delta\vee\Lambda)$ where $c$ is any given small absolute constant.

\rev{This section is organized by model. In the main text we provide the statements needed to invoke the general theory, while we defer lengthier technical derivations to Appendix~\ref{app:verify}.}

 \subsection{1-Bit Compressed Sensing (1bCS)}
 We consider the recovery of $\bx\in \calX:=\calK\cap \mathbb{S}^{n-1}$ from $\by=\sign(\bA\bx)$  where $\bA$ has i.i.d. $\calN(0,1)$ entries. In this case, we have
 $$\Delta = 2,~~\Lambda= 0,~~ \alpha=1,~~\beta = 1,~~b_1=0.$$
 We seek a global RAIC and set $\mu_4=2$; this brings no restriction to $\bu,\bv\in\mathbb{S}^{n-1}$. We proceed to establish Assumptions \ref{assum1}--\ref{assum_Puvup} for 1bCS.

 \begin{lem}[Assumptions \ref{assum1}--\ref{assump4}, \ref{assum_Puvup} for 1bCS] \label{lem:1bcssmallassum}\rev{In the model of 1bCS, Assumption \ref{assum1} holds, Assumption \ref{assum2} holds with $c^{(1)}=\sqrt{8/\pi}$, Assumption \ref{assum3} holds with $c^{(2)}=\frac{2}{\pi}$, Assumption \ref{assump4} holds  with $\varepsilon^{(1)}=c^{(3)}=c^{(4)}=0$, Assumption \ref{assum_Puvup} holds with $c^{(9)}=1$. }
 \end{lem}
 \begin{proof}
     The proofs consist of some elementary probability estimations. See Appendix \ref{app:provelem1}. 
 \end{proof}

 


\begin{lem}[Assumptions \ref{assum_sg1}--\ref{assum_sg3} for 1bCS]
    For $\ba_i\sim\calN(0,\bI_n)$, Assumption \ref{assum_sg1} holds for some absolute constant $c^{(5)}$ and
    $c^{(6)}=\sqrt{2/\pi}$ and $\varepsilon^{(2)}=0$; Assumption \ref{assum_sg2} holds for some absolute constant $c^{(7)}$ and $\varepsilon^{(3)} =0$;  Assumption \ref{assum_sg3} holds for some absolute constant $c^{(8)}$ and $ \varepsilon^{(4)}=0$. \label{fac:1bcs_sg2}
\end{lem}
\begin{proof}
  In the parameterization of $(\bp,\bq)$ via $(\bbeta_1,\bbeta_2)$, we additionally have $v_1=-u_1$  due to $\bp,\bq\in\mathbb{S}^{n-1}$ (see Lemma \ref{lem:parameterization}). As discussed, we can focus on sufficiently small $\|\bp-\bq\|_2=u_1-v_1=2u_1$ in the validation of (\ref{eq:sg1_sg}), (\ref{eq:sg2_sg}) and (\ref{eq:sg3_sg}). 
    The proof consists of some elementary integral estimates and calculations. It
    can be found in Appendix \ref{app:provefact2}. 
\end{proof}
Now we invoke Theorem \ref{thm:main} and come to the following result. 
\begin{theorem}[1bCS via PGD]\label{thm:1bcs}
   We recover $\bx\in\calX=\calK\cap \mathbb{S}^{n-1}$ for some star-shaped set $\calK$ from $\by = \sign(\bA\bx)$ with standard Gaussian $\bA$ by running Algorithm \ref{alg:pgd} with any $\bx^{(0)}\in\calX$ and $\eta=\sqrt{\frac{\pi}{2}}$, which produces a sequence $\{\bx^{(t)}\}_{t=0}^\infty$.  There exist some absolute constants $c_0,C_1,c_2,C_3,C_4$, for any $r\in (0,c_0)$, if 
    \begin{align*}
        m\ge C_1 \left(\frac{\omega^2(\calK_{(r)})}{r^3}+\frac{\scrH(\calX,r/2)}{r}\right),
    \end{align*}
then  with probability at least $1-\exp(-c_2\scrH(\calX,r/2))$, for any $\bx\in \calX$ we have  
 $$ \|\bx^{(t)}-\bx\|_2 \le C_3 r\sqrt{\log(r^{-1})},\quad\forall t\ge C_4\log(\log(r^{-1})).$$  
\end{theorem}
\begin{proof}
   We set $\mu_4=2$ and consider sufficiently small $r$.   As in Lemmas \ref{lem:1bcssmallassum}--\ref{fac:1bcs_sg2}, we have $\eta = \sqrt{\frac{\pi}{2}}= \frac{\Delta\vee\Lambda}{c^{(6)} \Delta }$ and $\varepsilon^{(i)}=0$ ($i=1,2,3,4$), which   justify all the conditions in Theorem \ref{thm:main}. Therefore, we have $$\|\bx^{(t)}-\bx\|_2 \le 2^{2^{-t}}(C_1r\sqrt{\log(r^{-1})})^{1-2^{-t}},\qquad\forall t\ge 0.$$   This implies the desired claim. 
\end{proof}
\begin{rem}\label{1bcsconcrete}
    Specializing Theorem \ref{thm:1bcs} to $\calK=\Sigma^n_k$ recovers the optimality of NBIHT proved in the recent work \cite{matsumoto2024binary}. Moreover, letting $\calK = M^{n_1,n_2}_{\bar{r}}$ leads to the first efficient and optimal algorithm for 1bCS of low-rank matrices. In addition, letting $\calK = \sqrt{k}\mathbb{B}_1^n$ shows PGD achieves $\widetilde{O}((\frac{k}{m})^{1/3})$ in 1bCS of effectively sparse signals, which  matches with Adaboost \cite{chinot2022adaboost}  (the fastest algorithm in existing works). 
\end{rem}

\subsection{Dithered 1-Bit Compressed Sensing (D1bCS)}
We consider the recovery of $\bx\in \calX:= \calK\cap \mathbb{B}_2^n$ from $\by=\sign(\bA\bx-\btau)$ under sub-Gaussian $\bA$ with rows $\{\ba_i\}_{i=1}^m$ satisfying Assumption \ref{assum1} and uniform dither $\btau\sim \scrU([-\lambda,\lambda]^m)$ with some   $\lambda\ge 2$; here, $\bA$ and $\btau$ are independent. In this case,  
we have 
$$\Delta=2,\quad \Lambda = \lambda,\quad \alpha=0,\quad\beta=1,\quad b_1=0.$$
We seek to establish a global RAIC with $\mu_4=2$. \rev{It remains to validate Assumptions \ref{assum2}--\ref{assum_Puvup}, which is accomplished in the following two lemmas.}   

\begin{lem}[Assumptions \ref{assum2}--\ref{assump4}, \ref{assum_Puvup} for D1bCS] \label{lem:d1bcsearly} \rev{In D1bCS, Assumption \ref{assum2} holds with $c^{(1)}=1$, Assumption \ref{assum3} holds with absolute constant $c^{(2)}$, Assumption \ref{assump4} holds  with $\varepsilon^{(1)}=c^{(3)}=c^{(4)}=0$, Assumption \ref{assum_Puvup} holds with $c^{(9)}=\frac{1}{2}$.}
\end{lem}
\begin{proof}
    \rev{The proof consists of some elementary integral estimations. Compared to 1bCS, $\ba_i$'s here no longer possess rotational invariance. Instead, the trick is to utilize the randomness of $\tau_i$ first to transform $\sfP_{\bu,\bv}$ to a technically more amenable form  (see, e.g., (\ref{taulower})). The complete proof can be found in Appendix \ref{app:proveearlyd1bcs}.}
\end{proof}



\begin{lem}[Assumptions \ref{assum_sg1}--\ref{assum_sg3} for D1bCS]  
\label{fact4}  
In our D1bCS setting, suppose that $\ba_i$ satisfies Assumption \ref{assum1} which leads to (\ref{eq:d1bcs_sg_abs}) with some absolute constants $\bar{c}_0,\bar{c}_1$. Then we have that Assumption \ref{assum_sg1} holds for some absolute constant $c^{(5)}$,
    $c^{(6)}=\frac{1}{2}$ and $\varepsilon^{(2)}=8(\bar{c}_1)^2\exp(-\bar{c}_0\lambda^2);$ Assumption \ref{assum_sg2} holds for some absolute constant $c^{(7)}$ and  $\varepsilon^{(3)}=8(\bar{c}_1)^2\exp(-\bar{c}_0\lambda^2)$; Assumption \ref{assum_sg3} holds for some absolute constant $c^{(8)}$ and  $\varepsilon^{(4)}=8(\bar{c}_1)^2\exp(-\bar{c}_0\lambda^2)$.
\end{lem}
\begin{proof}
The main idea, similarly to Lemma \ref{lem:d1bcsearly}, is to take expectation of the uniform dither first. This transforms the expectations into much more amenable forms. 
The proof   can be found in Appendix \ref{app:provefact3}. 
\end{proof}
 
With all assumptions being validated,   Theorem \ref{thm:main} yields the following statement. 
\begin{theorem}[D1bCS via PGD]\label{thm:d1bcsthm}
   Consider the recovery of $\bx\in \calX:=\calK\cap \mathbb{B}_2^n$ for some star-shaped set $\calK$ from   $\by = \sign(\bA\bx-\btau)$ where the $m$ rows of $\bA$ satisfy Assumption \ref{assum1} and  $\btau\sim\scrU([-\lambda,\lambda]^m)$. To this end, we obtain   $\{\bx^{(t)}\}_{t=0}^\infty$ by running Algorithm \ref{alg:pgd} with $\bx^{(0)}\in \calX$ and $\eta=\lambda$.     
    There exist some absolute constants $c_1,C_2,C_3,c_4,C_5,C_6$, given any $r\in(0,c_1)$, if $\lambda\ge C_2$ and 
    \begin{align*}
        m\ge C_3\lambda\left(\frac{\omega^2(\calK_{(r)})}{r^3}+\frac{\scrH(\calX,r/2)}{r}\right),
    \end{align*}
    then with probability at least $1-\exp(-c_4\scrH(\calX,r))$,  for any $\bx\in\calX$ we have 
        $$\|\bx^{(t)}-\bx\|_2\le C_5 r\log^{1/2}\Big(\frac{\lambda}{r}\Big),\quad \forall\,t\ge C_6\log (r^{-1}).$$ 
\end{theorem}
\begin{proof}
Under $\eta=\lambda=\frac{\Delta\vee\lambda}{c^{(6)} \Delta }$ and $\sum_{j=1}^3\varepsilon^{(j)}+5\varepsilon^{(4)}\le 56(\bar{c}_1)^2\exp(-\bar{c}_0\lambda^2)$, we use $\lambda \ge C_2$ with sufficiently large $C_2$  to ensure (\ref{eq:eta_condition}). We set $\mu_4=2$ and then   Theorem \ref{thm:main} implies $$\|\bx^{(t)}-\bx\|_2\le 2 (0.99^t) + C_3r\log^{1/2}\Big(\frac{\lambda}{r}\Big),\qquad \forall t\ge 0$$ which immediately leads to the desired claim.  
\end{proof}
\begin{rem}
    Similarly to Remark \ref{1bcsconcrete}, Theorem \ref{thm:d1bcsthm} implies that PGD is the first efficient and optimal algorithm for D1bCS of sparse vectors or low-rank matrices. Also, PGD achieves $\widetilde{O}((\frac{k}{m})^{1/3})$ for D1bCS of $\bx\in\sqrt{k}\mathbb{B}_1^n\cap \mathbb{B}_2^n$, matching with the fastest known efficient algorithm (i.e., the convex 
 program in \cite{jung2021quantized}). 
\end{rem} 

\subsection{Dithered Multi-Bit Compressed Sensing (DMbCS)}
We proceed to the more intricate multi-bit model DMbCS that concerns the recovery of $\bx\in \calX:=\calK\cap \mathbb{B}_2^n$ from $\by=\calQ_{\delta,L}(\bA\bx-\btau)$, where $\bA$ has rows $(\ba_i)_{i=1}^m$ satisfying Assumption \ref{assum1},  $\btau\sim\scrU([-\frac{\delta}{2},\frac{\delta}{2}]^m)$ is independent of $\bA$. In this setting, we have
$$\Delta = \delta,~~\Lambda=\frac{\delta}{2},~~\alpha=0,~~\beta = 1$$
 and the quantization thresholds $\{b_j:j\in[L-1]\}=\{j\delta :j=1-\frac{L}{2},\cdots,-1,0,1,\cdots,\frac{L}{2}-1\}$.

\subsubsection{The Sharp Local Approach}
We choose $\mu_4\le c_1\delta$ for sufficiently small $c_1$ and aim to show PGD converges to an error rate much sharper than $\mu_4$. We thus also suppose $r\lesssim \delta$. We validate Assumptions \ref{assum2}--\ref{assum_Puvup} in the next three lemmas.

\begin{lem}[Assumptions \ref{assum2}, \ref{assum3}, \ref{assum_Puvup} for DMbCS]
    In our DMbCS setting, if $L\delta\ge C_0$ for some absolute constant $C_0$, then Assumption \ref{assum2} holds with $c^{(1)}=2$, Assumptions \ref{assum3} and \ref{assum_Puvup} hold with some absolute constant $c^{(2)}>0$ and $c^{(9)} = 1$.  \label{fact7}
\end{lem}
\begin{proof}
    Compared to D1bCS, this appears more intricate because of the multiple quantization thresholds, while the core idea remains at using the randomness of $\tau_i\sim \scrU([-\frac{\delta}{2},\frac{\delta}{2}])$. The   proof can be found in Appendix \ref{app:provefact4}. 
\end{proof}

An additional intricacy for the multi-bit model is that the validation of Assumption \ref{assump4} requires substantial technical works. We have the following statement. 
\begin{lem}[Assumption \ref{assump4} for DMbCS]
    \label{fact8}
  In our DMbCS setting, there exist some absolute constants $c_2,C_3,c_4,C_5$ such that the following holds: given any sufficiently small $c_1>0$, if
    $$\mu_4\le c_1\delta,~~m\ge C_3\exp\Big(\frac{c_2}{c_1^2}\Big)\scrH(\calX,r/2)~\text{ and }~m\ge \ \frac{\delta \omega^2(\calK_{(r)})}{r^3},$$
 then with probability at least $1-2\exp(-c_4 \scrH(\calX,r))$, Assumption \ref{assump4} holds with $$ \varepsilon^{(1)}=C_5\exp\Big(-\frac{c_2}{2c_1^2}\Big)$$ and some absolute constants $c^{(3)}$ and $c^{(4)}$.  
\end{lem}
\begin{proof}
   See the complete proof in Appendix \ref{app:provefact5}. 
\end{proof}


\begin{lem}[Assumptions \ref{assum_sg1}--\ref{assum_sg3} for DMbCS]
    \label{fact9}
     In our DMbCS setting, assume that $\mu_4\le c_1\delta$ for some small enough $c_1$, and that $\ba_i$ satisfies (\ref{eq:d1bcs_sg_abs}) for some absolute constants $\bar{c}_0,\bar{c}_1$. We let 
     \begin{align*}
         \varepsilon' := 8 (\bar{c}_1)^2\Big[\exp\Big(-\frac{\bar{c}_0}{4c_1^2}\Big)+\exp\Big(-\frac{\bar{c}_0L^2\delta^2}{16}\Big)\Big],
     \end{align*}
     then we have that Assumption \ref{assum_sg1} holds for some absolute constant $c^{(5)}$, $c^{(6)}=1$ and $\varepsilon^{(2)}=\varepsilon'$; Assumption \ref{assum_sg2} holds for some absolute constant $c^{(7)}$, $\varepsilon^{(3)}=\varepsilon'$; Assumption \ref{assum_sg3} holds for some absolute constant $c^{(8)}$, $\varepsilon^{(4)}=\varepsilon'$. 
\end{lem}
\begin{proof}
    \rev{By using the randomness of $\tau_i$ first, we can essentially treat $\mathbbm{E}_{\tau_i}(\mathbbm{1}(E^{(i)}_{\bp,\bq}))$ as $\frac{|\ba_i^\top(\bp-\bq)|}{\delta}=\frac{\|\bp-\bq\|_2|\ba_i^\top\bbeta_1|}{\delta}$ for $(\bp,\bq)$ obeying $\|\bp-\bq\|_2\le 2\mu_4 \ll \delta$ and  $L\delta$ being large enough.} 
    The proof can be found in Appendix \ref{app:provefact6}.
\end{proof}
With Assumptions \ref{assum1}--\ref{assum_Puvup} being verified, we come to the following statement by invoking Theorem \ref{thm:main}.
\begin{coro}[DMbCS via PGD: Phase II]\label{cor:2ndstagemb}
In our DMbCS setting, we obtain $\{\bx^{(t)}\}_{t=0}^\infty$ by running Algorithm \ref{alg:pgd} with $\bx^{(0)}\in \calX\cap \mathbb{B}_2^n(\bx;c_0\delta)$ and $\eta =1$, where $c_0$ is sufficiently small. There exist some absolute constants $c_1,C_2,C_3,c_4,C_5$,    if $r\in (0,c_1\delta)$  for small enough $c_1$, $L\delta\ge C_2$ for large enough $C_2$, and 
    \begin{align}\label{sssaa}
        m\ge C_3\delta \left(\frac{\omega^2(\calK_{(r)})}{r^3}+\frac{\scrH(\calX,r/2)}{r}\right),
    \end{align}
    then with probability at least $1-\exp(-c_4\scrH(\calX,r))$, for any $\bx\in \calX$ and any $t\ge 0$ it holds that 
    \begin{align*}
        \|\bx^{(t)}-\bx\|_2 \le 0.99^t (c_0\delta) + C_5 r\log^{1/2}\Big(\frac{\delta}{r}\Big). 
    \end{align*}
\end{coro}
\begin{proof}
  We set $\mu_4=c_0\delta$ with small enough $c_0$ and   only need to check the conditions in Theorem \ref{thm:main}. 
    For $r\in (0,c_1\delta)$ with small enough $c_1$,  the sample complexity in Lemma \ref{fact8} (i.e., $m\ge \frac{\delta \omega^2(\calK_{(r)}) }{r^3}$)  can be ensured by (\ref{sssaa}) in our statement.  
    By Lemmas \ref{fact8}--\ref{fact9} and $\eta =1$, we have   
    $$\frac{\eta\Delta (\sum_{j=1}^4\varepsilon^{(j)})}{\Delta\vee\Lambda}\le C_5\exp\Big(-\frac{c_3}{2c_0^2}\Big)+24 (\bar{c}_1)^2\left[\exp\Big(-\frac{\bar{c}_0}{4c_0^2}\Big)+\exp\Big(-\frac{\bar{c}_0L^2\delta^2}{16}\Big)\right]$$
    for some absolute constants $c_3$ and $C_5$. Therefore, with small enough $c_0$ and large enough $L\delta$, we can ensure (\ref{eq:eta_condition}). Then we invoke Theorem \ref{thm:main} to establish the result. 
\end{proof}

\subsubsection{The Complementary Global Approach}
The incomplete part of Corollary \ref{cor:2ndstagemb} lies in the condition $\|\bx^{(0)}-\bx\|_2 \le c_0 \delta$ with small enough $c_0$, which remains non-trivial under fine quantization  corresponding to  small $\delta$. In this subsection, we invoke the complementary approach to show that {\it   Algorithm \ref{alg:pgd} with $\bx^{(0)}=0$ enters a region surrounding the underlying signal with radius $c_0\delta$ after a few iterations.}

In the generic development in Section \ref{sec:raic_qpe}, $r_1$ is the (order of) reconstruction accuracy we want to attain, so we should let $$r_1=c_1\delta$$ 
for small constant $c_1$ to be chosen, in order to fulfill the initialization requirement in Corollary \ref{cor:2ndstagemb}. Moreover, we should choose $\zeta=1$ and use $\ba_i^\top\bu$ to approximate   $\calQ_{\delta,L}(\ba_i^\top \bu-\tau_i)$ \cite{chen2024uniform,thrampoulidis2020generalized,xu2020quantized,jung2021quantized}, which leads to the decomposition $\calQ_{\delta,L}(\ba_i^\top\bu-\tau_i)=\ba_i^\top\bu +\xi_i(\bu)$. Accordingly, we set the step size $\eta =1$.

Therefore,   all that remains is to control $\hat{T}$ in (\ref{eq:qpe2raic}), which can be written as
\begin{align*}
  \hat{T}:= \sup_{\bu\in \overline{\calX}}\sup_{\bw\in\calK_{(r_1)}}\left|\frac{1}{mr_1}\big\langle\calQ_{\delta,L}(\bA\bu-\btau) -\bA\bu,\bA\bw\big\rangle\right|. 
\end{align*}
We are able to establish the following bound. 
\begin{lem}[Bounding $\hat{T}$]
    \label{factqpe}In our DMbCS setting we let $r_1=c_1\delta$ for some small enough $c_1$. There exist some constants $C_1,C_2,C_3,c_4,c_5,c_6,c_7$ that may depend on $c_1$, if $C_1\log^{1/2}(\frac{m}{\omega^2(\calX)})\le L\delta \le C_2 \omega(\calX)$ and 
    \begin{align}\label{pro6sam}
        m \ge C_3\left[\scrH(\calX,c_4\delta)+\frac{\omega^2(\calK_{(\delta)})}{\delta^2}+\frac{\omega(\calX)\omega(\calK_{(\delta)})}{\delta^2}\right],
    \end{align}
   then with probability at least $1-\exp(-\frac{c_5\omega^2(\calX)}{(L\delta)^2})-\exp(-c_6\scrH(\calX,2c_7\delta))$, we have $$\hat{T}\le r_1.$$
\end{lem}
\begin{proof}
    The proof can be found in Appendix \ref{app:factgqpe}. 
\end{proof}

 Then we invoke Theorem \ref{thm:conver_qpe}     to obtain the following statement.
\begin{coro}[DMbCS via PGD: Phase I]\label{coro3}
    In our DMbCS setting, suppose we   obtain $\{\bx^{(t)}\}_{t=0}^\infty$ by runing Algorithm \ref{alg:pgd} with $\eta=1$ and $\bx^{(0)}=0$. In the same setting as Proposition \ref{factqpe} (i.e., constraint on $L\delta$, sample size and  probability), for all $\bx\in \calX$ we have
    $$ \|\bx^{(t)}-\bx\|_2 \le \left(\frac{C_8\omega(\calK_{(\delta)})}{\delta\sqrt{m}}\right)^t+ 24c_1\delta.$$ 
\end{coro}
\begin{proof}
     Under our assumptions,  Proposition \ref{factqpe} implies $\hat{T}\le r_1$, and we  further note that other conditions needed in Theorem \ref{thm:conver_qpe} are satisfied. Hence, our statement follows from Theorem \ref{thm:conver_qpe}. 
\end{proof}
\subsubsection{Final Result}
  By running Algorithm \ref{alg:pgd} with $\bx^{(0)}=0$, Corollary \ref{coro3} establishes the first stage that the sequence enters $\mathbb{B}_2^n(\bx;c_0\delta)$, and then the subsequent refinement to a sharp error rate is characterized by Corollary \ref{cor:2ndstagemb}. 
We thus obtain the following final theorem for DMbCS. 

\begin{theorem}[DMbCS via PGD]\label{thm:mb}
   Consider the recovery of $\bx\in\calX:=\calK\cap\mathbb{B}_2^n$ for some star-shaped set $\calK$ from $\by=\calQ_{\delta,L}(\bA\bx-\btau)$ where the $m$ rows of $\bA$ satisfy Assumption \ref{assum_sg1} and $\btau\sim\scrU([-\frac{\delta}{2},\frac{\delta}{2}]^m)$. To this end, we obtain $\{\bx^{(t)}\}_{t=0}^\infty$ by running Algorithm \ref{alg:pgd} with $\bx^{(0)}=0$ and $\eta = 1$.  There exist some absolute constants $c_1,C_2,C_3,C_4,c_5,c_6,c_7,C_8,C_9$ such that, given any $r\in (0,c_1\delta)$, if $C_2 \log^{1/2}(\frac{m}{\omega^2(\calX)})\le L\delta\le C_3\omega(\calX)$ and 
    \begin{align}\label{eq:mb_sam}
        m\ge C_4\left(\frac{\delta\cdot \omega^2(\calK_{(r)})}{r^3}+\frac{\delta\cdot \scrH(\calX,r/2)}{r}+ \frac{\omega(\calX)\cdot\omega(\calK_{(\delta)})}{\delta^2}\right),
    \end{align}
     then with probability at least $1-\exp(-c_5\scrH(\calX,c_6\delta))-\exp(-\frac{c_7\omega^2(\calX)}{(L\delta)^2})$, for any $\bx\in\calX$ we have 
    \begin{align*}
        \|\bx^{(t)}-\bx\|_2 \le C_8r\log^{1/2}\Big(\frac{\delta}{r}\Big),\qquad\forall t\ge C_9\log(r^{-1}). 
    \end{align*}
\end{theorem}
\begin{proof}
Because $\scrH(\calX,c\delta)\lesssim \frac{\delta}{r}\scrH(\calX,\frac{r}{2})$ and $\frac{\omega^2(\calK_{(\delta)})}{\delta^2}\lesssim \frac{\delta \omega^2(\calK_{(r)})}{r^3}$, (\ref{eq:mb_sam}) can fulfill the sample complexity requirements in    Corollaries \ref{cor:2ndstagemb}--\ref{coro3}.
     Suppose the small enough $c_0$ in Corollary \ref{cor:2ndstagemb} has been  fixed, and regarding this $c_0$ we invoke Corollary \ref{coro3} with small enough $c_1$ to ensure $$\|\bx^{(t)}-\bx\|_2\le 0.99^t + \frac{c_0\delta}{2},\qquad \forall t \ge 0.$$ Note that this implies $\|\bx^{(t)}-\bx\|_2\le c_0\delta$ for any $t\ge t_0 := \lceil C_2\log(\delta^{-1})\rceil$ for some absolute constant $C_2$. Then, Corollary \ref{cor:2ndstagemb}  leads to $$\|\bx^{(t_0+t)}-\bx\|_2\le 0.99^t (c_0\delta) + C_3 r\log^{1/2}(\frac{\delta}{r}),\qquad\forall t\ge 0,$$ which implies $\|\bx^{(t_0+t)}-\bx\|_2\le C_4r \log^{1/2}(\frac{\delta}{r})$ for any $t\ge t_1:= \lceil C_5\log(\frac{\delta}{r})\rceil$.  The result follows due to $t_0+t_1\le C_6 \log(r^{-1})$.  
\end{proof}
\begin{rem}\label{rem:88}
    Corollary \ref{cor:2ndstagemb}, when specialized to $\calK=\Sigma^n_k$ and $M_{\bar{r}}^{n_1,n_2}$, yields the error rates $\widetilde{O}(\frac{\delta k}{m})$ and $\widetilde{O}(\frac{\delta\bar{r}(n_1+n_2)}{m})$. Under $L$-level quantizer, since $\delta=\widetilde{O}(\frac{1}{L})$ suffices to fulfill $L\delta \gtrsim \log^{1/2}(\frac{m}{\omega^2(\calK\cap \mathbb{B}_2^n)})$, PGD achieves   error rates $\widetilde{O}(\frac{k}{mL})$ and $\widetilde{O}(\frac{\bar{r}(n_1+n_2)}{mL})$ that nearly match the lower bounds in Theorem \ref{thm:lower}. For the recovery of sparse vectors, to our best knowledge, this is the first result for a (memoryless) multi-bit quantized compressed sensing system to be decoded nearly optimally by an efficient algorithm. 
\end{rem}
We believe the condition $m\gtrsim \delta^{-2}\omega(\calX)\omega(\calK_{(c_5\delta)})$ is a proof artifact and can be removed. While we leave further investigation to future work, we mention two regimes where our current arguments are able to get rid of this constraint. 
\begin{rem}\label{rem:99}
   The constraint $m\gtrsim \delta^{-2}\omega(\calX)\omega(\calK_{(c_5\delta)})$, which arises in Lemma \ref{lemsatu} (bounding the effect of saturation), is benign for the most interesting cases but precludes fine quantization with very small $\delta$. We can get rid of the constraint in either of the following two regimes (where there is no saturation and hence Lemma \ref{lemsatu} is no longer needed): (i) We consider the uniform quantizer $\calQ_\delta$ without the saturation issue; (ii) We only pursue the non-uniform recovery of a fixed signal $\bx\in \calX$  as with results in \cite{thrampoulidis2020generalized}; the reason is that for fixed $\bx$ we have $\|\bA\bx-\btau\|_\infty\lesssim\sqrt{\log m}$ w.h.p.; thus, if $L\delta\gtrsim \sqrt{\log m}$, then w.h.p. we have $\by = \calQ_{\delta,L}(\bA\bx- \btau)=\calQ_{\delta}(\bA\bx-\btau)$. 
\end{rem}

\section{Extension to Noisy Setting}\label{sec:discuss}
\rev{While previous sections focus on quantized compressed sensing under noiseless measurements $\by= \calQ(\bA\bx-\btau-\be)$,
it is worth discussing how to extend our results to a noisy setting. We shall focus on two types of noise (or corruption) commonly treated in prior works: a small fraction of adversarial bit flips \cite{plan2012robust,dirksen2021non,awasthi2016learning,matsumoto2024robust,chinot2022adaboost} and pre-quantization random noise \cite{plan2012robust,dirksen2021non,jung2021quantized,matsumoto2025learning}.}

\subsection{Adversarial corruption} \label{sec:noise}
In a follow-up work of \cite{matsumoto2024binary}, the same authors showed \cite{matsumoto2024robust} that NBIHT achieves error rate $\widetilde{O}(\frac{k}{m}+\bar{\zeta})$  under a $\bar{\zeta}$-fraction of adversarial bit flips, that is, a setting where one observes $\by_{\rm cor}$ obeying $$d_H(\by_{\rm cor},\sign(\bA\bx))\le \bar{\zeta}m.$$ 
Under the same corruption pattern, Adaboost achieves $\widetilde{O}((\frac{k}{m}+\bar{\zeta})^{1/3})$ for signals in $\sqrt{k}\mathbb{B}_1^n\cap\mathbb{S}^{n-1}$, as shown in \cite[Corollary 2.3]{chinot2022adaboost}.

Regarding our PGD algorithm, we shall  note that the robustness to adversarial bit flips is in fact a straightforward extension of the current corruption-free results. 
Suppose that we observe $\by_{\rm cor}=\by+\be$ under the adversarial corruption $\be=(e_1,\cdots,e_m)^\top$, then the update rule of PGD reads
\begin{align}
    \bx^{(t)} =\calP_{\mathbbm{A}_\alpha^\beta}\left( \calP_{\calK}\Big(\bx^{(t-1)}-\eta\cdot \bh(\bx^{(t-1)},\bx)-\frac{\eta\cdot\bA^\top\be}{m}\Big)\right).
\end{align}
By using an adapted version of RAIC to imply convergence, the only difference is that an additional term $
     O\big(\frac{\eta}{mr}\|\bA^\top\be\|_{\calK^\circ_{(r)}}\big)$
  contributes to the estimation error. We shall characterize the corruption jointly by its sparsity and $\ell_2$-norm:  
$$\|\be\|_0 \le \bar{\zeta}m\quad\text{and}\quad \|\be\|_2\le \bar{\gamma}\sqrt{m}.$$
 Then, Lemma \ref{lem:max_ell_sum}, together with the sample complexity in Theorem \ref{thm:main}, establishes the high-probability event 
        \begin{align*}
        &\frac{\eta}{mr}\|\bA^\top\be\|_{\calK^\circ_{(r)}}=\frac{\eta}{mr}\sup_{\bw\in \calK_{(r)}}\sum_{i=1}^m e_i \ba_i^\top\bw \\& \le \frac{\eta}{m}\sup_{\bw\in r^{-1}\calK_{(r)}}\max_{\substack{I\subset [m]\\|I|\le \bar{\zeta}m}}\sum_{i\in I}e_i|\ba_i^\top\bw|\\& \le \frac{\eta\|\be\|_2}{m}\sup_{\bw\in r^{-1}\calK_{(r)}}\max_{\substack{I\subset [m]\\|I|\le \bar{\zeta}m}} \Big(\sum_{i\in I}|\ba_i^\top\bw|^2\Big)^{1/2} \\&\lesssim \frac{\eta\bar{\gamma}\cdot\omega(\calK_{(r)})}{\sqrt{m}r}+\eta\bar{\gamma}\sqrt{\bar{\zeta}\log(\bar{\zeta}^{-1})}\\& \lesssim \eta\bar{\gamma}\sqrt{\frac{r}{\Delta\vee\Lambda}}+\eta\bar{\gamma}\sqrt{\bar{\zeta}\log(\bar{\zeta}^{-1})}.
    \end{align*}
For 1bCS, we can always set $\bar{\gamma}=(2\bar{\zeta})^{1/2}$; combined with $(r\bar{\zeta})^{1/2}\le 2r+2\bar{\gamma}$,  we find that the corruption only increments the estimation error by $O(\bar{\zeta}\log^{1/2}(\bar{\zeta}^{-1}))$.  This   recovers the major result in \cite{matsumoto2024robust} and \rev{improves on the corruption term $O(\bar{\zeta}^{1/3})$ from \cite[Corollary 2.3]{chinot2022adaboost}. In fact, a subsequent work shows \cite[Section 6]{abdalla2026robust} that this approach is more general in the following sense: in the analysis of PGD via RAIC, robustness can be obtained by bounding an additional random process (such as the above $\frac{\eta}{mr}\|\bA^\top\be\|_{\calK_{(r)}^\circ}$). Moreover, it is straightforward to argue that the corruption term $O(\bar{\zeta}\log^{1/2}(\bar{\zeta}^{-1}))$ is tight up to a logarithmic factor; see \cite[Remark 6.2]{abdalla2026robust}.}

\subsection{Pre-quantization random noise} 
\rev{The robustness to pre-quantization noise, which arises prior to the quantization and  corrupts $\by$ to $\by_{\rm cor}:=\calQ(\bA\bx-\btau-\be)$, is more intricate. For instance, consider $\be\sim\calN(0,\sigma^2\bI_m)$, then \cite[Theorem 1.6]{dirksen2023robust} established a lower bound $\Omega(\sigma\sqrt{k/m})$ in the recovery of $k$-sparse vector. As such,  our PGD upper bound $\widetilde{O}(\frac{k}{mL})$ is no longer attainable if, for instance, $\sigma\gtrsim 1$. After the first appearance of our paper on arXiv, \cite{matsumoto2025learning} appeared and contributed significantly to the problem of recovery of $\bx\in\Sigma^{n,*}_k$ from $\by=\sign(\bA\bx - \be)$ where $\be\sim \calN(0,\sigma^2\bI_m)$; this problem can be viewed as 1bCS with pre-quantization Gaussian noise, but indeed, is also the sparse probit regression model in statistics. For any $\sigma>0$, it was shown \cite[Corollary 4.3]{matsumoto2025learning} that NBIHT (with the noisy observations) achieves 
\[\|\hat{\bx}_{nbiht}-\bx\|_2 =\widetilde{O}\big(\frac{k}{m}\big)+ \widetilde{O}\bigg(\sqrt{\frac{(\sigma+\sigma^2)k}{m}}\bigg),\]
implying that a tiny pre-quantization noise ($\sigma\lesssim \frac{k}{m}$) does not affect the optimal rate of $\widetilde{O}(k/m)$. Notably, this upper bound is likely tight (up to logarithmic factors) in light of the lower bounds for logistic regression \cite{hsu2024sample}. For D1bCS, some treatment to possibly heavy-tailed pre-quantization noise can be found in \cite{dirksen2020one}. }


\section{Numerical Simulations}\label{sec:experiment}
 We demonstrate the error rates of Algorithm \ref{alg:pgd} implied by  Theorems \ref{thm:1bcs}--\ref{thm:mb}  by simulating the following: 
\begin{itemize}
    \item \textbf{1bCS} where we   recover $\bx\in \calK=\Sigma^{n}_k,~  M^{n_1,n_2}_{\bar{r}},~\sqrt{k}\mathbb{B}_1^n$ with unit $\ell_2$-norm from $\by = \sign(\bA\bx)$ under $\bA\sim \calN^{m\times n}(0,1)$, by running  Algorithm \ref{alg:pgd} with  $\bx^{(0)}\in \calK\cap \mathbb{S}^{n-1}$ and  $\eta = \sqrt{\pi/2}$ for $100$ iterations; 
    \item \textbf{D1bCS} where we recover $\bx\in\calK=\Sigma^{n}_k,~  M^{n_1,n_2}_{\bar{r}},~\sqrt{k}\mathbb{B}_1^n$ with $\ell_2$-norm not greater than $1$ from $\by=\sign(\bA\bx-\btau)$, under $\bA$ having i.i.d. $\{-1,1\}$-valued Bernoulli entries (Bernoulli design) and $\btau\sim\scrU([-\Lambda,\Lambda]^m)$, by running Algorithm \ref{alg:pgd} with $\bx^{(0)}=0$  and $\eta= \Lambda$ for $100$ iterations; 
    \item \textbf{DMbCS} where we recover   $\bx\in\calK=\Sigma^{n}_k,~ M^{n_1,n_2}_{\bar{r}},~\sqrt{k}\mathbb{B}_1^n$ with $\ell_2$-norm not greater than $1$ from $\by = \calQ_{\delta,L}(\bA\bx-\btau)$, under Bernoulli design $\bA$ and $\btau\sim\scrU([-\frac{\delta}{2},\frac{\delta}{2}]^m)$, by   running Algorithm \ref{alg:pgd} with $\bx^{(0)}=0$ and $\eta= 1$ for $100$ iterations;
\end{itemize}

We draw a $k$-sparse signal from a uniform distribution over $\Sigma^n_k\cap \mathbb{S}^{n-1}$, and generate a rank-$\bar{r}$ matrix in $\mathbb{R}^{n_1\times n_2}$ by drawing $\bX_0\sim \calN^{n_1\times n_2}(0,1)$ and then retaining only the top $\bar{r}$ SVD components of $\bX_0$. Then, the signal is normalized to unit $\ell_2$-norm for 1bCS,   rescaled to $\|\bx\|_2\sim\scrU([0,1])$ for D1bCS and DMbCS. For effectively sparse signals living in $\mathbb{B}_1^n(\sqrt{k})$, we generate it as 
\begin{align*}
    \bx = \big(\underbrace{\pm a,\pm a,\cdots,\pm a}_{\text{first~$c$~entries}},\underbrace{\pm b,\pm b,\cdots,\pm b}_{\text{last~$n-c$~entries}}\big)^\top
\end{align*}
where $c\sim \scrU(\{1,\cdots,\lceil 0.6k\rceil\})$, $a=\frac{\sqrt{k}+\sqrt{k+n(n-k-c)/c}}{n}$,  $b=\frac{\sqrt{k}-ca}{n-c}$, and the sign of each entry is $1$ or $-1$ with equal probability. It is easy to verify that signals generated in this way have unit $\ell_2$-norm and $\ell_1$-norm  equaling $\sqrt{k}$, respectively. Each data point is averaged over $50$ independent trials. \rev{The MATLAB code used to generate the figures in this paper is available at \url{https://github.com/junrenchen58/optimal-qcs}.}

Our experimental results are reported in Figures \ref{fig:1bcs}--\ref{fig:dmbcs}. More details are provided in the caption. For the recovery of sparse signals and low-rank matrices (the first two sub-figures), we clearly find that the data points roughly shape a straight line with slope $-1$, confirming that the errors decay  with $m$ in the optimal rate $O(m^{-1})$. For the recovery of effectively sparse signals (the rightmost sub-figures), we observe decay rates that are noticeably slower than $O(m^{-1})$ but faster than $O(m^{-1/3})$. In fact, some curves clearly follow the $m^{-1/3}$ decay rate and suggest that this is tight in recovering effectively sparse signals via PGD. 
Moreover, simultaneously doubling $k$ (or $\bar{r}$) and the measurement number $m$ (e.g., $(k,m)=(3,400)$ v.s. $(k,m)=(6,800)$) almost maintains the same estimation error, which corroborates the scaling laws $\frac{k}{m}$ and $\frac{\bar{r}}{m}$.  It is   worth mentioning that $\Lambda$ for D1bCS should be chosen to fit the signal norm, and not doing so  leads to performance degradation, as seen by the curves of $\Lambda=0.8,~3.2$ in \rev{Figure~\ref{fig:d1bcs} (Left)}. This is also true for $L\delta$ in DMbCS, though we only report the results under the appropriate choice $L\delta=5$.

In the multi-bit model DMbCS we  additionally track the role of $L$ to corroborate the scaling law $\frac{1}{mL}$. To this end, we simultaneously double $L$ and halve $m$ to maintain the same $mL$ (e.g., $(L,m)=(4,200)$ v.s. $(L,m)=(8,100)$); we have observed  that   the   estimation errors in these cases almost coincide, which is consistent with our theoretical error rate. This is however not true for overly large $L$: the curves of $L=32,~64$ in \rev{Figure~\ref{fig:dmbcs} (Left) and $L=32$ in Figure~\ref{fig:dmbcs} (Middle)} fail to maintain the same estimation errors as others, especially under small measurement number. This phenomenon can also be nicely interpreted by our Theorem \ref{thm:mb}: consider the recovery of $k$-sparse signals, we will need $m\gtrsim k\log(\frac{en}{k})$ to fulfill (\ref{eq:mb_sam}), namely, the measurement number should exceed a threshold that is independent of $L$ and $\delta$; when $\frac{mL}{k}$ is small, the measurement numbers for the cases with large $L$  may not be numerous enough to fulfill (\ref{eq:mb_sam}).

\begin{figure}[ht!]
	\begin{centering}
		\includegraphics[width=0.3\columnwidth]{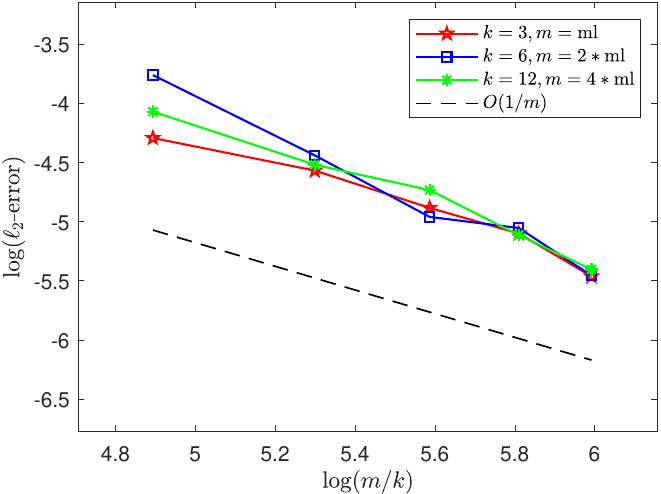} \quad \includegraphics[width=0.3\columnwidth]{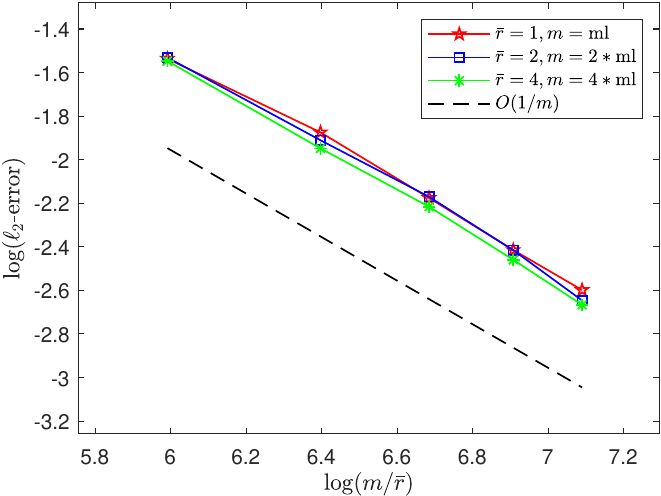}\quad
  \includegraphics[width=0.3\columnwidth]{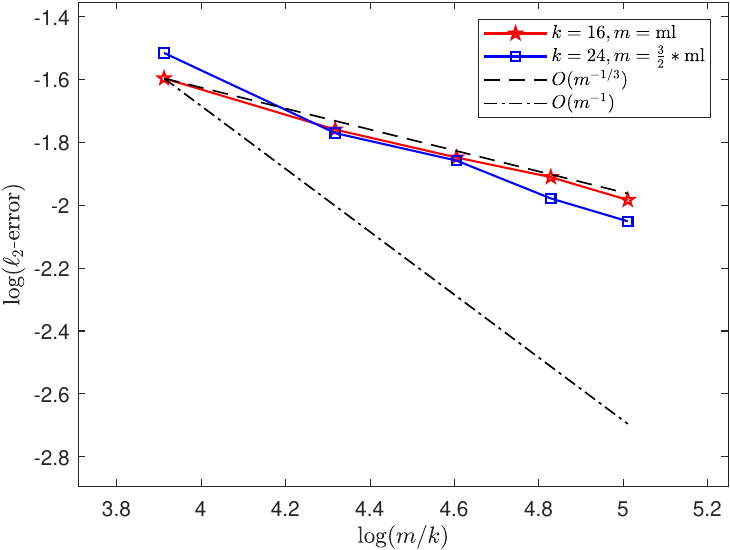}
		\par
	\end{centering}
	
	\caption{\label{fig:1bcs}\small  In 1bCS, PGD achieves error rates $\widetilde{O}(\frac{k}{m})$ for $\bx\in \Sigma^n_k$ (left), $\widetilde{O}(\frac{\bar{r}(n_1+n_2)}{m})$ for $\bx\in M^{n_1,n_2}_{\bar{r}}$ (middle), and $\widetilde{O}((\frac{k}{m})^{1/3})$ for $\bx\in \sqrt{k}\mathbb{B}_1^n\cap\mathbb{S}^{n-1}$ (right). In sparse recovery, we recover $k$-sparse $500$-dimensional signals under  $m=c\cdot\text{ml}$ with $\text{ml}=400:200:1200$. In low-rank recovery, we recover rank-$\bar{r}$ $25\times 25$ matrices under  $m=c\cdot\text{ml}$ with $\text{ml}=400:200:1200$. In recovering effectively sparse signals, we test $300$-dimensional signals under   $m=c\cdot\text{ml}$ with $\text{ml}=800:400:2400$.} 
\end{figure}

\begin{figure}[ht!]
	\begin{centering}
		~\includegraphics[width=0.3\columnwidth]{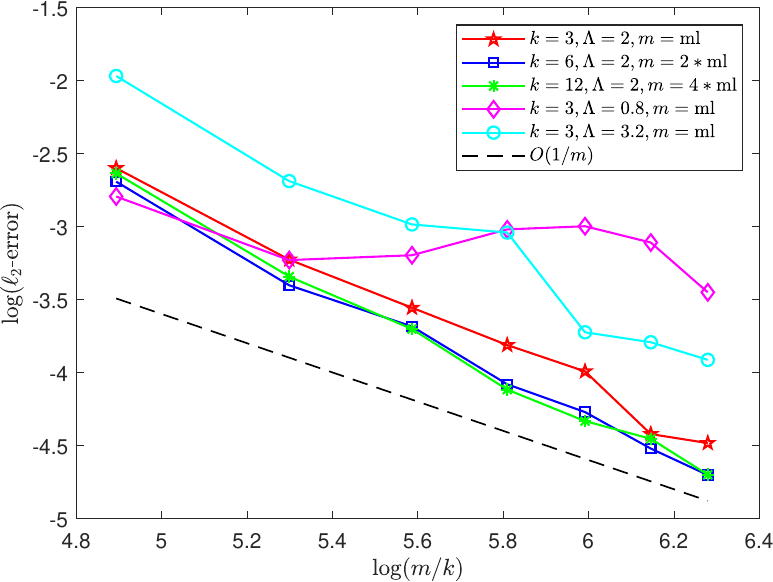} \quad \includegraphics[width=0.3\columnwidth]{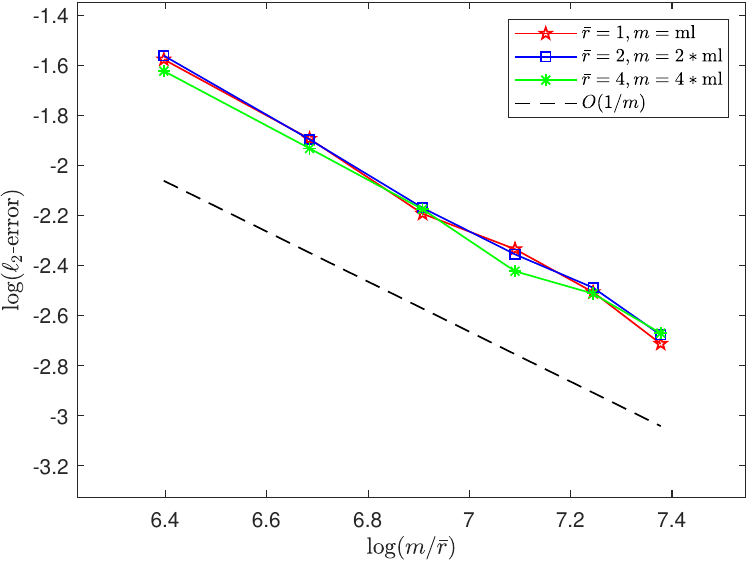}\quad \includegraphics[width=0.3\columnwidth]{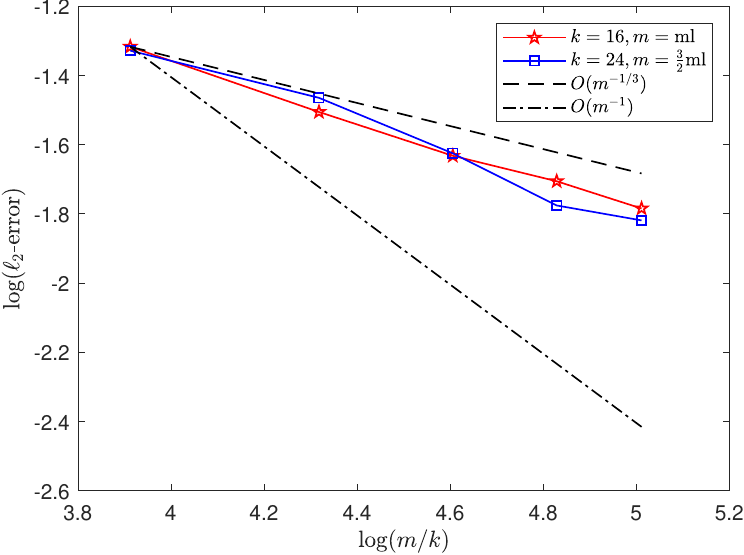}
		\par
	\end{centering}
	
	\caption{\label{fig:d1bcs}\small In D1bCS, PGD achieves error rates $\widetilde{O}(\frac{k}{m})$ for $\bx\in \Sigma^n_k$ (left), $\widetilde{O}(\frac{\bar{r}(n_1+n_2)}{m})$ for $\bx\in M^{n_1,n_2}_{\bar{r}}$ (middle), and $\widetilde{O}((\frac{k}{m})^{1/3})$ for $\bx\in \sqrt{k}\mathbb{B}_1^n\cap\mathbb{S}^{n-1}$ (right). In sparse recovery, we recover $k$-sparse $500$-dimensional signals under   $m=c\cdot\text{ml}$ with $\text{ml}=400:200:1600$. In low-rank recovery, we recover rank-$\bar{r}$ $25\times 25$ matrices under measurement number $m=c\cdot\text{ml}$ with $\text{ml}=600:200:1600$ with $\Lambda=1.5$. In recovering effectively sparse signals, we test $300$-dimensional signals under   $m=c\cdot\text{ml}$ with $\text{ml}=800:400:2400$.}
\end{figure}

\begin{figure}[ht!]
	\begin{centering}
		~~\includegraphics[width=0.29\columnwidth]{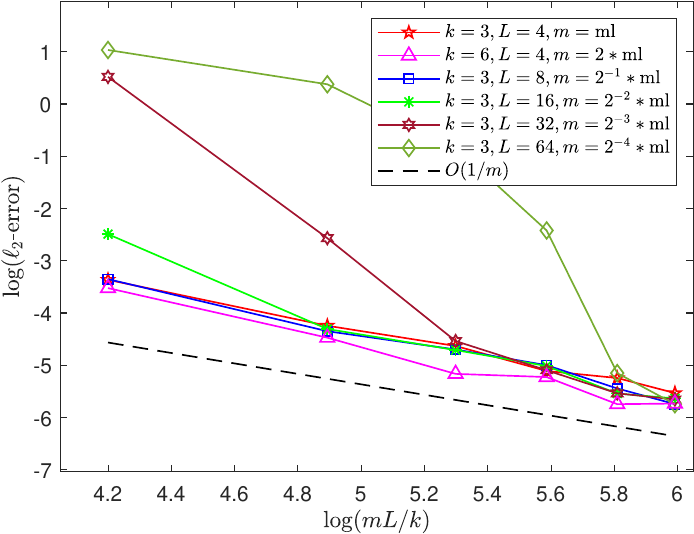} \quad \includegraphics[width=0.3\columnwidth]{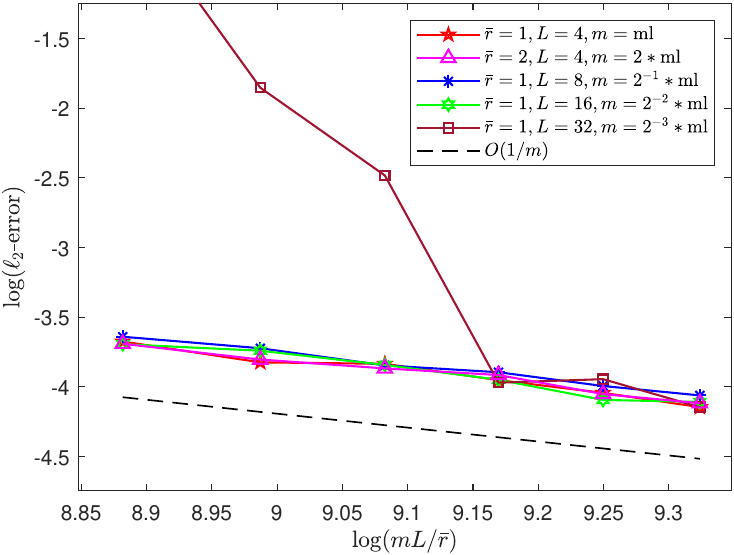}\quad \includegraphics[width=0.3\columnwidth]{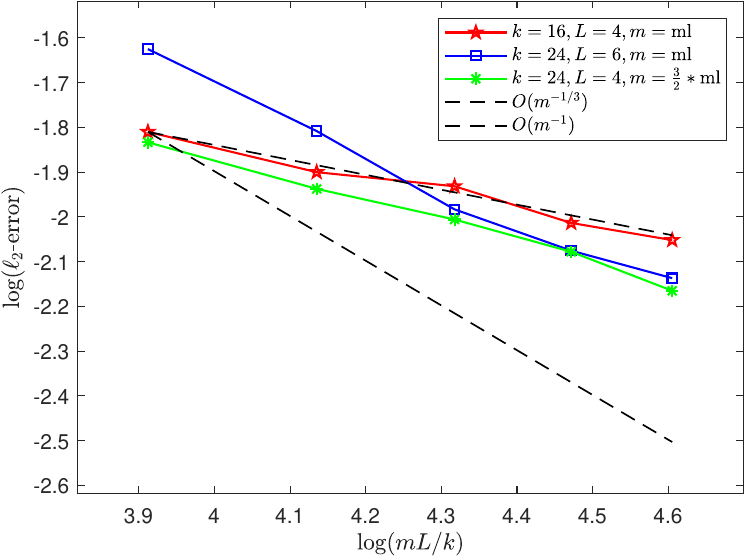}
		\par
	\end{centering}
	
	\caption{\label{fig:dmbcs}\small In DMbCS with $\delta=\frac{5}{L}$, PGD achieves error rates $\widetilde{O}(\frac{k}{mL})$ for $\bx\in \Sigma^n_k$ (left), $\widetilde{O}(\frac{\bar{r}(n_1+n_2)}{mL})$ for $\bx\in M^{n_1,n_2}_{\bar{r}}$ (middle), and $\widetilde{O}((\frac{k}{mL})^{1/3})$ for $\bx\in \sqrt{k}\mathbb{B}_1^n\cap\mathbb{S}^{n-1}$ (right). In sparse recovery, we recover $k$-sparse $500$-dimensional signals under measurement number $m=c\cdot\text{ml}$ with $\text{ml}=200:200:1200$. In low-rank recovery, we recover rank-$\bar{r}$ $25\times 25$ matrices under measurement number $m=c\cdot\text{ml}$ with $\text{ml}=1800:200:2800$. In recovering effectively sparse signals, we test $400$-dimensional signals under   $m=c\cdot\text{ml}$ with $\text{ml}=800:200:1600$.}
\end{figure}

\section{Conclusion}\label{sec:conclusion}
The goal of quantized compressed sensing is to recover   structured signals $\bx$ from only a few quantized compressive measurements. The information-theoretic lower bounds for recovering $k$-sparse signal are known as $\widetilde{O}(\frac{k}{mL})$ in $\ell_2$-norm. However, most efficient algorithms in the literature are sub-optimal and only achieve error rates inferior to $\widetilde{O}(\frac{1}{L}(\frac{k}{m})^{1/2})$.   The only exception is   NBIHT that was recently shown \cite{matsumoto2024binary} to be near-optimal for the very specific problem of recovering $k$-sparse $\bx$ from $\by=\sign(\bA\bx)$. Moreover, for recovering the effectively sparse signals in $\sqrt{k}\mathbb{B}_1^n$, the best known error rate (including that for the intractable HDM) is at the order $\widetilde{O}((\frac{k}{mL})^{1/3})$; see Table \ref{table1}.

In this work, we provide a unified treatment to the general quantized compressed sensing problem by analyzing a PGD algorithm. An intriguing finding is that, under sub-Gaussian design, PGD is a procedure that achieves the same error rate as HDM under the following conditions:   {\it separation probability}, {\it small-ball probability} and {\it some moment bounds}. By validating these conditions for the popular models 1bCS, D1bCS and DMbCS, we establish error rates for PGD that improve on or match the fastest error rates in all instances.

We mention a few questions that need further investigation.
An important direction  is to develop information-theoretic lower bound for the approximately sparse cases. For example, the rate $\widetilde{O}((\frac{k}{m})^{1/3})$ is the best known rate for recovering signals in $\sqrt{k}\mathbb{B}_1^n$, but it remains unclear whether this is tight in any sense, as we discussed in Remark \ref{lower_eff}. Moreover, the current analysis is built upon sub-Gaussian sensing matrix. It is of great interest to relax this assumption to fit better with the sensing matrix in actual signal processing applications or the features in machine learning problems.

\subsection*{Acknowledgment}
This work was done while JC was visiting the Department of Statistics at Columbia University in the summer of 2024, during his Ph.D. studies in the Department of Mathematics at The University of Hong Kong. The authors would like to thank the reviewers for some helpful comments that helped improve the quality and readibility of this paper. 
\bibliography{libr}

@book{13concen,
    author = {Boucheron, Stéphane and Lugosi, Gábor and Massart, Pascal},
    title = {Concentration Inequalities: A Nonasymptotic Theory of Independence},
    publisher = {Oxford University Press},
    year = {2013},
    month = {February}
}

@article{matsumoto2024binary,
  title={Binary iterative hard thresholding converges with optimal number of measurements for 1-bit compressed sensing},
  author={Matsumoto, Namiko and Mazumdar, Arya},
  journal={Journal of the ACM},
  volume={71},
  number={5},
  pages={1--64},
  year={2024},
  publisher={ACM New York, NY}
}

@article{oymak2017fast,
  title={Fast and reliable parameter estimation from nonlinear observations},
  author={Oymak, Samet and Soltanolkotabi, Mahdi},
  journal={SIAM Journal on Optimization},
  volume={27},
  number={4},
  pages={2276--2300},
  year={2017},
  publisher={SIAM}
}

@article{genzel2016high,
  title={High-dimensional estimation of structured signals from non-linear observations with general convex loss functions},
  author={Genzel, Martin},
  journal={IEEE Transactions on Information Theory},
  volume={63},
  number={3},
  pages={1601--1619},
  year={2016},
  publisher={IEEE}
}

@article{abdalla2026robust,
  title={Robust Uniform Recovery of Structured Signals from Nonlinear Observations},
  author={Abdalla, Pedro and Balan, Radu and Chen, Junren},
  journal={arXiv preprint arXiv:2604.20075},
  year={2026}
}

@article{mendelson2008uniform,
  title={Uniform uncertainty principle for Bernoulli and subgaussian ensembles},
  author={Mendelson, Shahar and Pajor, Alain and Tomczak-Jaegermann, Nicole},
  journal={Constructive Approximation},
  volume={28},
  number={3},
  pages={277--289},
  year={2008},
  publisher={Springer}
}

@article{chen2025unified,
  title={A Unified Approach to Statistical Estimation Under Nonlinear Observations: Tensor Estimation and Matrix Factorization},
  author={Chen, Junren and Ding, Lijun and Xia, Dong and Yuan, Ming},
  journal={arXiv preprint arXiv:2510.16965},
  year={2025}
}

@inproceedings{hsu2024sample,
  title={On the sample complexity of parameter estimation in logistic regression with normal design},
  author={Hsu, Daniel and Mazumdar, Arya},
  booktitle={The Thirty Seventh Annual Conference on Learning Theory},
  pages={2418--2437},
  year={2024},
  organization={PMLR}
}

@InProceedings{matsumoto2025learning,
  title = 	 {Learning sparse generalized linear models with binary outcomes via iterative hard thresholding},
  author =       {Matsumoto, Namiko and Mazumdar, Arya},
  booktitle = 	 {Proceedings of Thirty Eighth Conference on Learning Theory},
  pages = 	 {3933--4032},
  year = 	 {2025},
  volume = 	 {291},
  series = 	 {Proceedings of Machine Learning Research},
  month = 	 {30 Jun--04 Jul},
  publisher =    {PMLR}
}

@inproceedings{Foucart2013AMI,
  title={A Mathematical Introduction to Compressive Sensing},
  author={Simon Foucart and Holger Rauhut},
  booktitle={Applied and Numerical Harmonic Analysis},
  year={2013}
}

@article{chen2024one,
  title={One-Bit Phase Retrieval: Optimal Rates and Efficient Algorithms},
  author={Chen, Junren and Yuan, Ming},
  journal={IEEE Transactions on Information Theory},
  year={2026}
}

@article{chen2023parameter,
  title={A parameter-free two-bit covariance estimator with improved operator norm error rate},
  author={Chen, Junren and Ng, Michael K},
  journal={Applied and Computational Harmonic Analysis},
  year={2025}
}

@inproceedings{boufounos20081,
  title={1-bit compressive sensing},
  author={Boufounos, Petros T and Baraniuk, Richard G},
  booktitle={2008 42nd Annual Conference on Information Sciences and Systems},
  pages={16--21},
  year={2008},
  organization={IEEE}
}

@article{genzel2023unified,
  title={A unified approach to uniform signal recovery from nonlinear observations},
  author={Genzel, Martin and Stollenwerk, Alexander},
  journal={Foundations of Computational Mathematics},
  volume={23},
  number={3},
  pages={899--972},
  year={2023},
  publisher={Springer}
}

@article{friedlander2021nbiht,
  title={NBIHT: An efficient algorithm for 1-bit compressed sensing with optimal error decay rate},
  author={Friedlander, Michael P and Jeong, Halyun and Plan, Yaniv and Y{\i}lmaz, {\"O}zg{\"u}r},
  journal={IEEE Transactions on Information Theory},
  volume={68},
  number={2},
  pages={1157--1177},
  year={2021},
  publisher={IEEE}
}

@article{liu2019one,
  title={One-bit compressive sensing with projected subgradient method under sparsity constraints},
  author={Liu, Dekai and Li, Song and Shen, Yi},
  journal={IEEE Transactions on Information Theory},
  volume={65},
  number={10},
  pages={6650--6663},
  year={2019},
  publisher={IEEE}
}

@article{plan2016generalized,
  title={The generalized lasso with non-linear observations},
  author={Plan, Yaniv and Vershynin, Roman},
  journal={IEEE Transactions on Information Theory},
  volume={62},
  number={3},
  pages={1528--1537},
  year={2016},
  publisher={IEEE}
}

@article{plan2017high,
  title={High-dimensional estimation with geometric constraints},
  author={Plan, Yaniv and Vershynin, Roman and Yudovina, Elena},
  journal={Information and Inference: A Journal of the IMA},
  volume={6},
  number={1},
  pages={1--40},
  year={2017},
  publisher={Oxford University Press}
}

@ARTICLE{raskutti2011tit,
  author={Raskutti, Garvesh and Wainwright, Martin J. and Yu, Bin},
  journal={IEEE Transactions on Information Theory}, 
  title={Minimax Rates of Estimation for High-Dimensional Linear Regression Over $\ell_q$ -Balls}, 
  year={2011},
  volume={57},
  number={10},
  pages={6976-6994},
  doi={10.1109/TIT.2011.2165799}}

@article{plan2013one,
  title={One-bit compressed sensing by linear programming},
  author={Plan, Yaniv and Vershynin, Roman},
  journal={Communications on Pure and Applied Mathematics},
  volume={66},
  number={8},
  pages={1275--1297},
  year={2013},
  publisher={Wiley Online Library}
}

@article{plan2012robust,
  title={Robust 1-bit compressed sensing and sparse logistic regression: A convex programming approach},
  author={Plan, Yaniv and Vershynin, Roman},
  journal={IEEE Transactions on Information Theory},
  volume={59},
  number={1},
  pages={482--494},
  year={2012},
  publisher={IEEE}
}

@article{mendelson2016upper,
  title={Upper bounds on product and multiplier empirical processes},
  author={Mendelson, Shahar},
  journal={Stochastic Processes and their Applications},
  volume={126},
  number={12},
  pages={3652--3680},
  year={2016},
  publisher={Elsevier}
}

@inproceedings{awasthi2016learning,
  title={Learning and 1-bit compressed sensing under asymmetric noise},
  author={Awasthi, Pranjal and Balcan, Maria-Florina and Haghtalab, Nika and Zhang, Hongyang},
  booktitle={Conference on Learning Theory},
  pages={152--192},
  year={2016},
  organization={PMLR}
}

@article{ai2014one,
  title={One-bit compressed sensing with non-Gaussian measurements},
  author={Ai, Albert and Lapanowski, Alex and Plan, Yaniv and Vershynin, Roman},
  journal={Linear Algebra and its Applications},
  volume={441},
  pages={222--239},
  year={2014},
  publisher={Elsevier}
}

@article{soltanolkotabi2019structured,
  title={Structured signal recovery from quadratic measurements: Breaking sample complexity barriers via nonconvex optimization},
  author={Soltanolkotabi, Mahdi},
  journal={IEEE Transactions on Information Theory},
  volume={65},
  number={4},
  pages={2374--2400},
  year={2019},
  publisher={IEEE}
}

@article{candes2015phase,
  title={Phase retrieval via Wirtinger flow: Theory and algorithms},
  author={Candes, Emmanuel J and Li, Xiaodong and Soltanolkotabi, Mahdi},
  journal={IEEE Transactions on Information Theory},
  volume={61},
  number={4},
  pages={1985--2007},
  year={2015},
  publisher={IEEE}
}

@article{jacques2013robust,
  title={Robust 1-bit compressive sensing via binary stable embeddings of sparse vectors},
  author={Jacques, Laurent and Laska, Jason N and Boufounos, Petros T and Baraniuk, Richard G},
  journal={IEEE Transactions on Information Theory},
  volume={59},
  number={4},
  pages={2082--2102},
  year={2013},
  publisher={IEEE}
}

@article{candes2011tight,
  title={Tight oracle inequalities for low-rank matrix recovery from a minimal number of noisy random measurements},
  author={Candes, Emmanuel J and Plan, Yaniv},
  journal={IEEE Transactions on Information Theory},
  volume={57},
  number={4},
  pages={2342--2359},
  year={2011},
  publisher={IEEE}
}

@article{dirksen2020one,
  title={One-bit compressed sensing with partial Gaussian circulant matrices},
  author={Dirksen, Sjoerd and Jung, Hans Christian and Rauhut, Holger},
  journal={Information and Inference: A Journal of the IMA},
  volume={9},
  number={3},
  pages={601--626},
  year={2020},
  publisher={Oxford University Press}
}

@article{baraniuk2017exponential,
  title={Exponential decay of reconstruction error from binary measurements of sparse signals},
  author={Baraniuk, Richard G and Foucart, Simon and Needell, Deanna and Plan, Yaniv and Wootters, Mary},
  journal={IEEE Transactions on Information Theory},
  volume={63},
  number={6},
  pages={3368--3385},
  year={2017},
  publisher={IEEE}
}

@article{dirksen2023robust,
  title={Robust one-bit compressed sensing with partial circulant matrices},
  author={Dirksen, Sjoerd and Mendelson, Shahar},
  journal={The Annals of Applied Probability},
  volume={33},
  number={3},
  pages={1874--1903},
  year={2023},
  publisher={Institute of Mathematical Statistics}
}

@article{goldstein2019non,
  title={Non-Gaussian observations in nonlinear compressed sensing via Stein discrepancies},
  author={Goldstein, Larry and Wei, Xiaohan},
  journal={Information and Inference: A Journal of the IMA},
  volume={8},
  number={1},
  pages={125--159},
  year={2019},
  publisher={Oxford University Press}
}

@inproceedings{matsumoto2024robust,
  title={Robust 1-bit Compressed Sensing with Iterative Hard Thresholding},
  author={Matsumoto, Namiko and Mazumdar, Arya},
  booktitle={Proceedings of the 2024 Annual ACM-SIAM Symposium on Discrete Algorithms (SODA)},
  pages={2941--2979},
  year={2024},
  organization={SIAM}
}

@article{plan2014dimension,
  title={Dimension reduction by random hyperplane tessellations},
  author={Plan, Yaniv and Vershynin, Roman},
  journal={Discrete \& Computational Geometry},
  volume={51},
  number={2},
  pages={438--461},
  year={2014},
  publisher={Springer}
}

@article{dirksen2022sharp,
  title={Sharp estimates on random hyperplane tessellations},
  author={Dirksen, Sjoerd and Mendelson, Shahar and Stollenwerk, Alexander},
  journal={SIAM Journal on Mathematics of Data Science},
  volume={4},
  number={4},
  pages={1396--1419},
  year={2022},
  publisher={SIAM}
}

@article{oymak2015near,
  title={Near-optimal bounds for binary embeddings of arbitrary sets},
  author={Oymak, Samet and Recht, Ben},
  journal={arXiv preprint arXiv:1512.04433},
  year={2015}
}

@article{xu2020quantized,
  title={Quantized compressive sensing with rip matrices: The benefit of dithering},
  author={Xu, Chunlei and Jacques, Laurent},
  journal={Information and Inference: A Journal of the IMA},
  volume={9},
  number={3},
  pages={543--586},
  year={2020},
  publisher={Oxford University Press}
}

@book{vershynin2018high,
  title={High-dimensional probability: An introduction with applications in data science},
  author={Vershynin, Roman},
  volume={47},
  year={2018},
  publisher={Cambridge university press}
}

@article{thrampoulidis2020generalized,
  title={The generalized lasso for sub-gaussian measurements with dithered quantization},
  author={Thrampoulidis, Christos and Rawat, Ankit Singh},
  journal={IEEE Transactions on Information Theory},
  volume={66},
  number={4},
  pages={2487--2500},
  year={2020},
  publisher={IEEE}
}

@article{dirksen2021non,
  title={Non-Gaussian hyperplane tessellations and robust one-bit compressed sensing},
  author={Dirksen, Sjoerd and Mendelson, Shahar},
  journal={Journal of the European Mathematical Society},
  volume={23},
  number={9},
  pages={2913--2947},
  year={2021}
}

@article{jung2021quantized,
  title={Quantized compressed sensing by rectified linear units},
  author={Jung, Hans Christian and Maly, Johannes and Palzer, Lars and Stollenwerk, Alexander},
  journal={IEEE Transactions on Information Theory},
  volume={67},
  number={6},
  pages={4125--4149},
  year={2021},
  publisher={IEEE}
}

@article{yang2022towards,
  title={Towards an efficient approach for the nonconvex lp ball projection: algorithm and analysis},
  author={Yang, Xiangyu and Wang, Jiashan and Wang, Hao},
  journal={Journal of Machine Learning Research},
  volume={23},
  number={101},
  pages={1--31},
  year={2022}
}

@article{bahmani2013unifying,
  title={A unifying analysis of projected gradient descent for $\ell_p$-constrained least squares},
  author={Bahmani, Sohail and Raj, Bhiksha},
  journal={Applied and Computational Harmonic Analysis},
  volume={34},
  number={3},
  pages={366--378},
  year={2013},
  publisher={Elsevier}
}

@article{chinot2022adaboost,
  title={AdaBoost and robust one-bit compressed sensing},
  author={Chinot, Geoffrey and Kuchelmeister, Felix and L{\"o}ffler, Matthias and van de Geer, Sara},
  journal={Mathematical Statistics and Learning},
  volume={5},
  number={1},
  pages={117--158},
  year={2022}
}

@article{chen2024uniform,
author = {Chen, Junren and Liu, Zhaoqiang and Ding, Meng and Ng, Michael K.},
title = {Uniform Recovery Guarantees for Quantized Corrupted Sensing Using Structured or Generative Priors},
journal = {SIAM Journal on Imaging Sciences},
volume = {17},
number = {3},
pages = {1909-1977},
year = {2024}
}

@inproceedings{tu2016low,
  title={Low-rank solutions of linear matrix equations via procrustes flow},
  author={Tu, Stephen and Boczar, Ross and Simchowitz, Max and Soltanolkotabi, Mahdi and Recht, Ben},
  booktitle={International conference on machine learning},
  pages={964--973},
  year={2016},
  organization={PMLR}
}

@inproceedings{boufounos2015quantization,
  title={Quantization and compressive sensing},
  author={Boufounos, Petros T and Jacques, Laurent and Krahmer, Felix and Saab, Rayan},
  booktitle={Compressed Sensing and its Applications: MATHEON Workshop 2013},
  pages={193--237},
  year={2015},
  organization={Springer}
}
\bibliographystyle{plain}

\newpage

 {\centering \huge \bf Appendix \par}
 \begin{appendix}

\section{Roadmap and Notation}
  Below is the roadmap for the Appendix:  
\begin{itemize}
    \item The main aim of this appendix is to provide a table of notation for the convenience of the readers; 
    \item Appendix \ref{lower_proof} provides the proof of Theorem \ref{thm:lower}; 
    \item In Appendix \ref{app:tessella}, we prove Theorems \ref{thm:local_embed}--\ref{thm:itup};
    \item Appendix \ref{app:convergence} proves Theorem \ref{thm:convergence} regarding the convergence of PGD implied by RAIC; 
    \item Appendix \ref{app:provemain} proves our main result, Theorem \ref{thm:main}; 
    \item Appendix \ref{app:provecomplement} proves Theorem \ref{thm:conver_qpe} concerning the convergence via the complementary approach; 
    \item Appendix \ref{app:unified11} collects the missing proofs in the unified analysis for the sharp approach to RAIC in Section \ref{sec:sharp_ana} and the complementary approach in Section \ref{sec:raic_qpe};
    \item In Appendix \ref{app:verify}, we show the optimality of PGD for the specific models of 1bCS, D1bCS and DMbCS by validating all the assumptions made in our unified framework;
    \item The auxiliary lemmas that support our technical proofs are collected in Appendix \ref{app:lemma}.  
\end{itemize}
The recurring non-standard notation can be found in     Table \ref{tbl:notation}.

\begin{table}[ht]
\centering
\caption{List of recurring non-standard notation \label{tbl:notation}}
\vspace{1mm}

{\small
\begin{tabular}{|c|c|}
\hline 
$\bA,\bx,\by$ & sub-Gaussian   matrix, true signal, vector of quantized measurements\tabularnewline
\hline 
$\calQ$ & a generic $L$-level quantizer 
\tabularnewline
\hline 
$\Delta$ & resolution of the quantizer  (set as $2$ for 1-bit quantizer)
\tabularnewline
\hline 
$\btau,\Lambda$ & uniform dither $\btau\sim \scrU[-\Lambda,\Lambda]^m$, dithering scale
\tabularnewline
\hline 
$\calK$ &  the star-shaped set   that captures signal structure 
\tabularnewline
\hline
$\mathbbm{A}_\alpha^\beta$ &  the annulus   that captures signal norm 
\tabularnewline
\hline
$\calX,\overline{\calX}$ &  signal set  $\calX= \calK\cap \mathbbm{A}_\alpha^\beta$, enlarged signal set $\overline{\calX}=(2\calK)\cap \mathbbm{A}_\alpha^\beta$  
\tabularnewline
\hline 
$\calQ_\delta,\calQ_{\delta,L}$ &  uniform quantizer (\ref{eq:Qdelta}), $L$-level uniform quantizer (\ref{eq:QdeltaL})
\tabularnewline
\hline 
 $\sfP_{\bu,\bv}$ & the probability of $\bu,\bv$ being distinguished (or separated)
\tabularnewline
\hline 
$\calL(\bu),\calL_1(\bu)$ & hamming distance loss, one-sided $\ell_1$ loss  
\tabularnewline
\hline 
$\bh(\bu,\bv)$ & subgradient of $\calL_1(\bu)$ under true signal $\bv$ (\ref{eq:huvdef1}) 
\tabularnewline
\hline 
$\hat{\bh}(\bu,\bv)$ &  the clipped subgradient (\ref{eq:hathpq})
\tabularnewline
\hline 
 $r $ & desired accuracy for the sharp approach  
\tabularnewline
\hline 
 $r _1$ & desired accuracy for the complementary approach  
\tabularnewline
\hline 
$\calD^{(2)}_r $ & $\{(\bp,\bq)\in\overline{\calX}:\|\bp-\bq\|_2\le r\}$
\tabularnewline
\hline 
$\calN^{(2)}_{r,\mu_4}$ & $\{(\bp,\bq)\in\calN_r\times\calN_r: 0<\|\bp-\bq\|_2\le 2\mu_4\}$
\tabularnewline
\hline 
$\bR_{\bp,\bq}$ & the index set $\{i\in[m]:\calQ(\ba_i^\top\bp-\tau_i)\ne\calQ(\ba_i^\top\bq-\tau_i)\}$  
\tabularnewline
\hline
$E^{(i)}_{\bp,\bq}$ & the event $\{i\in \bR_{\bp,\bq}\}$
\tabularnewline
\hline
$c^{(i)},~ i\in[9]$ & constants in Assumptions \ref{assum2}--\ref{assum_Puvup} 
\tabularnewline
\hline
$\varepsilon^{(i)},~i\in[4]$ & sufficiently small constants in Assumptions  \ref{assump4}--\ref{assum_sg3} 
\tabularnewline
\hline
$\kappa_\alpha$ & constant that equals $1$ when $\alpha=0$ and equals $2$ when $\alpha>0$ 
\tabularnewline
\hline
$\zeta,\xi_i(\bu)$ & viewing nonlinear observations as noisy linear ones (\ref{nonlinearview})
\tabularnewline
\hline
\end{tabular}
}
\end{table}

\section{The Proof of Theorem \ref{thm:lower} (Lower Bounds)}\label{lower_proof}
\begin{proof}
We allow the constants in this proof to depend on the given $\alpha$ and $\beta$. Because $V_K\subset \calK$, 
we only need to establish the lower bounds for the  easier problems of recovering $\bx\in V_K \cap \mathbbm{A}_{\alpha}^\beta$ with $\beta>\alpha$ from $\by = \calQ(\bA\bx-\btau)$, and recovering $\bx\in V_K\cap \mathbb{S}^{n-1}$ from $\by=\sign(\bA\bx)$; these problems are {\it easier} because of the smaller signal space. For a small $\epsilon>0$, by standard packing results (e.g., \cite{vershynin2018high}), there exist  $P_1\subset V_K\cap \mathbbm{A}_{\alpha}^\beta$ with $|P_1|\ge (\frac{c_1}{\epsilon})^K$ and $P_2\subset V_K\cap\mathbb{S}^{n-1} $ with $|P_2|\ge (\frac{c_2}{\epsilon})^{K-1}$, such that the $\ell_2$-distance of any two different points in $P_i~(i=1,2)$ is greater than $\epsilon$. On the other hand, by standard counting results, we have $$|\{\calQ(\bA\bx-\btau):\bx\in V_K\}|\le \big|\big\{(\sign(\bA\bx -\btau- b_1)^\top,\cdots,\sign(\bA\bx -\btau- b_{L-1})^\top)^\top:\bx\in V_K\big\}\big| ,$$
which can be bounded by {\it the number of orthants in $\mathbb{R}^{(L-1)m}$ intersected by an affine linear subspace of dimension no greater than $K$.} (This affine linear subspace can be formally written as 
$ \big\{\big((\bA\bx-\btau-b_1)^\top,\cdots, (\bA\bx-\btau-b_{L-1})^\top\big)^\top:\bx\in V_K\big\}.$) 
Therefore, we have (see, e.g., \cite[Lemma 3.2]{chen2024one})
$$|\{\calQ(\bA\bx-\btau):\bx\in V_K\}| \le \Big(\frac{C_3(L-1)m}{K}\Big)^K;$$
 by similar reasoning, we also have (e.g., \cite[Eq. (16)]{jacques2013robust}) $$|\{\sign(\bA\bx): \bx\in V_K\}|\le \Big(\frac{C_4m}{K}\Big)^{K-1}.$$ Notice that both $(\frac{C_3(L-1)m}{K})^K \ge (\frac{c_1}{\epsilon})^K$ and $(\frac{C_4m}{K})^{K-1}\ge (\frac{c_2}{\epsilon})^{K-1}$ lead to $\epsilon\ge \frac{c_5K}{mL}$. 
\end{proof}

 \section{Quantized Embedding Property}\label{app:tessella}
  The general $L$-level quantizer (\ref{eq:Llevel}) associated with some $(\ba,\tau)$ forms  the quantized sampler $\calQ(\langle\ba,\cdot\rangle-\tau)$. 
\subsection{Preparations}\label{appendix_sepa}
We will need the notions of (well) separation. A hyperplane $\calH_{\ba,\tau}:=\{\bp\in \mathbb{R}^n:\ba^\top\bp-\tau = 0\}$ separates the whole space $\mathbb{R}^n$ into two sides:
\begin{align*}
    \calH_{\ba,\tau}^+ = \{\bp\in \mathbb{R}^n:\ba^\top\bp\ge\tau\}~~\text{and}~~\calH^-_{\ba,\tau} = \{\bp\in\mathbb{R}^n:\ba^\top\bp<\tau\}.
\end{align*}
We say two points $\bu,\bv$ are separated by $\calH_{\ba,\tau}$ if they live in different sides of the hyperplane, which is just a geometric statement equivalent to $\sign(\ba^\top\bu-\tau)\ne\sign(\ba^\top\bv-\tau)$.

To build this viewpoint for the $L$-level quantizer (\ref{eq:Llevel}), we note that $\calQ(\ba^\top\bu-\tau)\ne\calQ(\ba^\top\bv-\tau)$ happens if and only if $\sign(\ba^\top\bu-\tau-b_j)\ne\sign(\ba^\top\bv-\tau-b_j)$ holds for some $j\in [L-1]$. Hence,  the $L$-level quantizer corresponds to $L-1$ hyperplanes $\{\calH_{\ba,\tau+b_j}:j\in[L-1]\}$; geometrically, $\calQ(\ba^\top\bu-\tau)\ne\calQ(\ba^\top\bv-\tau)$ holds if and only if $\bu,\bv$ are separated by $\calH_{\ba,\tau+b_j}$ for some $j\in[L-1]$.

Such hyperplane separation, however, can be unstable: suppose that $\bu\in \calH_{\ba,\tau}^+$ and $\bv\in \calH_{\ba,\tau}^-$ live very close to the hyperplane itself $\calH_{\ba,\tau}$, then $\calH_{\ba,\tau}$ may no longer separate $\bu'$ and $\bv$ even though $\bu'$ and $\bu$ are very close. As with \cite{plan2014dimension,dirksen2021non,jung2021quantized,chen2024one}, we shall work with $\theta$-well-separation to resolve this issue.
\begin{definition}
    In the context of quantized compressed sensing under the quantizer (\ref{eq:Llevel}) with resolution $\Delta$, we say $\calH_{\ba,\tau}$ $\theta$-well-separates $\bu$ and $\bv$ if   $\calH_{\ba,\tau}$ separates $\bu$ and $\bv$, and it holds additionally that 
    \begin{align*}
       \min\big\{|\ba^\top\bu-\tau|,|\ba^\top\bv-\tau|\big\}\ge \theta\min\{\|\bu-\bv\|_2,\Delta\vee \Lambda\}.
    \end{align*}
\end{definition}
Next, we extend the notion of well (separation) to an $L$-level quantizer. 
\begin{definition}\label{def:theta_well}
For $\calQ$ in (\ref{eq:Llevel}) and some $(\ba,\tau)$, we say the quantized sampler $\calQ(\langle \ba,\cdot\rangle-\tau)$ distinguishes   $\bu$ and $\bv$, if there exists some $j\in[L-1]$ such that $\calH_{\ba,\tau+b_j}$ separates $\bu$ and $\bv$ (this is just equivalent to $\calQ(\ba^\top\bu-\tau)\ne\calQ(\ba^\top\bv-\tau)$). We say the quantized sampler $\calQ(\langle \ba,\cdot\rangle-\tau)$ $\theta$-well-distinguishes   $\bu$ and $\bv$, if there exists some $j\in[L-1]$ such that $\calH_{\ba,\tau+b_j}$ $\theta$-well-separates $\bu$ and $\bv$.
\end{definition}
 In the following, we present two useful facts. 
 \begin{lem}\label{lem:non_sepa11}
    Given $\bu$ and $\bv$, if 
    \begin{align}\label{eq:sepacon}
         \min_{j\in[L-1]}|\ba^\top\bu-\tau - b_j|>2|\ba^\top(\bu-\bv)|,
    \end{align}
    then $\bu$ and $\bv$ are not distinguished by $\calQ(\langle\ba,\cdot\rangle-\tau)$, i.e., $\calQ(\ba^\top\bu-\tau)=\calQ(\ba^\top\bv-\tau)$.
 \end{lem}
 \begin{proof}
     We shall assume $\bu,\bv$ are distinguished by $\calQ(\langle\ba,\cdot\rangle-\tau)$ and seek a contradiction. Particularly, suppose $\calQ(\ba^\top\bu-\tau)\ne \calQ(\ba^\top\bv-\tau)$, then there exists some $j\in[L-1]$ such that $\calH_{\ba,\tau+b_j}$ separates $\bu$ and $\bv$, namely, $\sign(\ba^\top\bu-\tau-b_j)\ne\sign(\ba^\top\bv-\tau-b_j)$. This leads to 
     \begin{align*}
         |\ba^\top(\bu-\bv)|=\big|(\ba^\top\bu-\tau-b_j)-(\ba^\top\bv-\tau-b_j)\big| \ge |\ba^\top\bu-\tau-b_j|,
     \end{align*}
     which contradicts (\ref{eq:sepacon}). 
 \end{proof}
 \begin{lem}\label{lem:wellimplyse}
     Given $\bu,\bv,\bp,\bq$ and some $\theta>0$, if $\calQ(\langle\ba,\cdot\rangle -\tau)$ $\theta$-well-distinguishes $\bu$ and $\bv$, and additionally 
     \begin{align}\label{eq:sepa_con2}
         \max\big\{|\ba^\top(\bp-\bu)|,|\ba^\top(\bq-\bv)|\big\} < \frac{\theta}{2}\min\big\{\|\bu-\bv\|_2,\Delta\vee\Lambda\big\}, 
     \end{align}
     then $\bp$ and $\bq$ are distinguished by $\calQ(\langle \ba,\cdot\rangle-\tau)$, i.e., $\calQ(\ba^\top\bp -\tau)\ne \calQ(\ba^\top\bq-\tau)$. 
 \end{lem}
 \begin{proof}
     By Definition \ref{def:theta_well},  for some $j\in [L-1]$ we have $\sign(\ba^\top\bu -\tau - b_j)\ne\sign(\ba^\top\bv -\tau-b_j)$ and 
     \begin{align*}
         \min\big\{|\ba^\top \bu-\tau-b_j|,|\ba^\top\bv-\tau-b_j|\big\}\ge\theta \min\{\|\bu-\bv\|_2,\Delta\vee\Lambda\}.
     \end{align*}
     Combining with (\ref{eq:sepa_con2}) we obtain 
         \begin{align*}
         &2|\ba^\top(\bp-\bu)|<|\ba^\top\bu-\tau-b_j|,\\
         &2|\ba^\top(\bq-\bv)|<|\ba^\top\bv-\tau-b_j|.
     \end{align*}
     Now we invoke Lemma \ref{lem:non_sepa11} with respect to the 1-bit quantized sampler $\sign(\langle \ba,\cdot\rangle-\tau - b_j)$, yielding that it does not distinguish $\bu$ and $\bp$, nor $\bv$ and $\bq$. We hence arrive at $\sign(\ba^\top\bp -\tau-b_j)=\sign(\ba^\top\bu-\tau-b_j)$ and $\sign(\ba^\top\bq-\tau-b_j)=\sign(\ba^\top\bv-\tau-b_j)$. Taken collectively with $\sign(\ba^\top\bu-\tau-b_j)\ne\sign(\ba^\top\bv-\tau-b_j)$, we obtain $\sign(\ba^\top\bp-\tau-b_j)\ne \sign(\ba^\top \bq-\tau-b_j)$ that implies the desired $\calQ(\ba^\top\bp-\tau)\ne\calQ(\ba^\top\bq-\tau)$. 
 \end{proof}
 \subsection{The Proof of Theorem \ref{thm:local_embed} (Quantized Embedding Property)}
\begin{proof}
 We now delve into the arguments for proving Theorem \ref{thm:local_embed}. For some  $\epsilon\in(0,c_0(\Delta\vee\Lambda))$ with small enough $c_0$, we let $\epsilon' = \frac{c_1\epsilon}{\log^{1/2}(\frac{\Delta\vee\Lambda}{\epsilon})}$ for some small enough constant  $c_1$. We construct $\calN_{\epsilon'}$ as a minimal $\epsilon'$-net of $\calW$  that satisfies $\log|\calN_{\epsilon'}|=\scrH(\calW,\epsilon')$. For any $\bu,\bv\in\calW$, we find their closest points in $\calN_{\epsilon'}$ as follows: 
\begin{align}\label{eq:foundpq}
    \bp = \text{arg}\min_{\bw\in \calN_{\epsilon'}}\|\bw-\bu\|_2~~\text{and}~~\bq = \text{arg}\min_{\bw\in \calN_{\epsilon'}}\|\bw-\bv\|_2.
\end{align}
\subsubsection*{(i) Establish the event $E_{\rm small}$}
Given any $\bu,\bv\in\calW$ obeying $\|\bu-\bv\|_2\le \epsilon'/2$, we have $\|\bu-\bp\|_2\le \epsilon'$ and $\|\bv-\bp\|_2\le\|\bv-\bu\|_2+\|\bu-\bp\|_2\le3\epsilon'/2$. By Lemma \ref{lem:non_sepa11} and a union bound, we can proceed with the following deterministic arguments:
    \begin{align*}
    &m-d_H\big(\calQ(\bA\bu-\btau),\calQ(\bA\bp-\btau)\big) \\
    &= \sum_{i=1}^m \mathbbm{1}\big(\calQ(\ba_i^\top\bu-\tau_i)=\calQ(\ba_i^\top\bp-\tau_i)\big)\\
    &\ge \sum_{i=1}^ m \mathbbm{1}\Big(\min_{j\in[L-1]}|\ba_i^\top \bp-\tau_i-b_j|>\epsilon,~|\ba_i^\top(\bu-\bp)|\le\frac{\epsilon}{2}\Big)\\
    &\ge \sum_{i=1}^m \mathbbm{1}\Big(\min_{j\in[L-1]}|\ba_i^\top \bp-\tau_i-b_j|>\epsilon\Big)-\sum_{i=1}^m\mathbbm{1}\Big(|\ba_i^\top(\bu-\bp)|>\frac{\epsilon}{2}\Big),
\end{align*}
which implies 
\begin{align*}
    d_H\big(\calQ(\bA\bu-\btau),\calQ(\bA\bp-\btau)\big) \le \sum_{i=1}^m \mathbbm{1}\Big(\min_{j\in[L-1]}|\ba_i^\top\bp-\tau_i-b_j|\le \epsilon\Big)+\sum_{i=1}^m\mathbbm{1}\Big(|\ba_i^\top(\bu-\bp)|>\frac{\epsilon}{2}\Big).
\end{align*}
This remains valid when replacing $\bu$ with $\bv$: 
\begin{align*}
    d_H\big(\calQ(\bA\bv -\btau),\calQ(\bA\bp-\btau)\big)\le \sum_{i=1}^m \mathbbm{1}\Big(\min_{j\in[L-1]}|\ba_i^\top\bp-\tau_i-b_j|\le \epsilon\Big)+\sum_{i=1}^m\mathbbm{1}\Big(|\ba_i^\top(\bv-\bp)|>\frac{\epsilon}{2}\Big).
\end{align*}
Noticing $\bp\in\calN_{\epsilon'}$ and $\bu-\bp,\bv-\bp\in \calW_{(3\epsilon'/2)}$, by triangle inequality and taking supremums, we arrive at 
     \begin{align}
    &d_H\big(\calQ(\bA\bu-\btau),\calQ(\bA\bv-\btau)\big)\nn\\&\le d_H \big(\calQ(\bA\bu-\btau),\calQ(\bA\bp-\btau)\big)+d_H\big(\calQ(\bA\bv-\btau),\calQ(\bA\bp-\btau)\big)\nn\\\nn
    &\le 2 \sum_{i=1}^m \mathbbm{1}\Big(\min_{j\in[L-1]}|\ba_i^\top\bp-\tau_i-b_j|\le \epsilon\Big) + \sum_{i=1}^m \mathbbm{1}\Big(|\ba_i^\top(\bu-\bp)|>\frac{\epsilon}{2}\Big)+\sum_{i=1}^m \mathbbm{1}\Big(|\ba_i^\top(\bv-\bp)|>\frac{\epsilon}{2}\Big)\\
    & \le 2 \underbrace{\sup_{\bp\in\calN_{\epsilon'}} \sum_{i=1}^m \mathbbm{1}\Big(\min_{j\in[L-1]}|\ba_i^\top\bp-\tau_i-b_j|\le \epsilon\Big)}_{:=\breve{T}_1} + 2 \underbrace{\sup_{\bw\in\calW_{(3\epsilon'/2)}}\sum_{i=1}^m \mathbbm{1}\Big(|\ba_i^\top\bw|>\frac{\epsilon}{2}\Big)}_{:=\breve{T}_2}. \label{eq:Eslocal}
\end{align}
 \paragraph{Bound $\breve{T}_1$:} To bound $\breve{T}_1$, we shall invoke Chernoff bound along with a union bound. By Assumption \ref{assum2}, for any $\bp\in\calN_{\epsilon'}$, $\min_{j\in[L-1]}|\ba_i^\top \bp-\tau_i-b_j|\le \epsilon$ holds with probability smaller than $\frac{c^{(1)}\epsilon}{\Delta\vee \Lambda}$, hence we have 
    \begin{align*}
    &\mathbbm{P}\left(\sum_{i=1}^m \mathbbm{1}\Big(\min_{j\in[L-1]}|\ba_i^\top\bp-\tau_i-b_j|\le \epsilon\Big)\ge \frac{3mc^{(1)}\epsilon}{2(\Delta\vee\Lambda)}\right)\\
    &\le\mathbbm{P}\left(\text{Bin}\Big(m,\frac{c^{(1)}\epsilon}{\Delta\vee\Lambda}\Big)\ge \frac{2mc^{(1)}\epsilon}{\Delta\vee\Lambda}\right)\\&\le \exp \Big(-\frac{c^{(1)}}{12}\frac{m\epsilon}{\Delta\vee\Lambda}\Big),
\end{align*}
where the last inequality is due to Chernoff bound. Hence, a union bound shows that $\breve{T}_1\le \frac{3c^{(1)}m\epsilon}{2(\Delta\vee\Lambda)}$ holds with probability at least $$1-\exp\Big(\scrH(\calW,\epsilon')-\frac{c^{(1)}}{12}\frac{m\epsilon}{\Delta\vee\Lambda}\Big)\ge 1-\exp\Big(-\frac{c^{(1)}}{24}\frac{m\epsilon}{\Delta\vee\Lambda}\Big),$$ with the proviso that $m\ge \frac{24(\Delta\vee\Lambda)}{c^{(1)}}\frac{\scrH(\calW,\epsilon')}{\epsilon}$, which has been assumed in (\ref{eq:local_scale}) in our statement.

\paragraph{Bound $\breve{T}_2$:} We seek to control $\breve{T}_2$ to a scaling similar to the bound for $\breve{T}_1$. In particular, we seek to show $\breve{T}_2\le \frac{c'm\epsilon}{\Delta\vee \Lambda}$ for some  constant $c'$ that is independent of other constants and can be set small enough if needed. To this end, we let $\ell_0:=\lceil \frac{c'm\epsilon}{\Delta\vee\Lambda}\rceil$ and observe that it is sufficient to ensure 
\begin{align*}
    \sup_{\bw\in\calW_{(3\epsilon'/2)}}\max_{\substack{I\subset [m]\\|I|\le \ell_0}}\left(\frac{1}{\ell_0}\sum_{i\in I}|\ba_i^\top\bw|^2\right)^{1/2} \le \frac{\epsilon}{2}.
\end{align*}
 By Lemma \ref{lem:max_ell_sum}, with probability at least $1-2\exp(-\frac{c'm\epsilon}{\Delta\vee\Lambda})$, we only need to ensure 
 \begin{align*}
    \frac{\omega(\calW_{(3\epsilon'/2)})}{\sqrt{\ell_0}} + \frac{3\epsilon'}{2}\log ^{1/2}\Big(\frac{em}{\ell_0}\Big) \le c_3\epsilon
 \end{align*}
for some small enough absolute constant $c_3$; note that this follows from (\ref{eq:local_scale}) and our choice $\epsilon' = \frac{c\epsilon}{\log^{1/2}(\frac{\Delta\vee\Lambda}{\epsilon})}$ with small enough $c$. Substituting the bounds on $\breve{T_1}$ and $\breve{T}_2$ into (\ref{eq:Eslocal}), we have shown that the event $E_{\rm small}$ holds with the promised probability. 

\subsubsection*{(ii) Establish the event $E_{\rm large}$}
With Assumptions \ref{assum2}--\ref{assum3}, we let $c^*:=\min\{\frac{c^{(2)}}{4c^{(1)}},\frac{1}{4}\}$ and first establish a lower bound on the probability of $\bu$ and $\bv$ being $c^*$-well-distinguished by $\calQ(\langle \ba_i,\cdot\rangle-\tau_i)$:
    \begin{align}
&\mathbbm{P}\Big(\calQ(\langle\ba_i,\cdot\rangle-\tau_i)~\text{$c^*$-well-distinguishes }\bu\text{ and }\bv\Big)\nn\\
    &\ge \mathbbm{P}\big(\calQ(\ba_i^\top\bu-\tau_i)\ne\calQ(\ba_i^\top\bv-\tau_i)\big) \nn\\&\quad -  \mathbbm{P}\left(\min_{j\in[L-1]}\min\{|\ba_i^\top\bu-\tau_i-b_j|,|\ba_i^\top\bv-\tau_i-b_j|\}<c^*\min\{\|\bu-\bv\|_2,\Delta\vee\Lambda\}\right)\nn\\
    &\ge  \mathbbm{P}\big(\calQ(\ba_i^\top\bu-\tau_i)\ne\calQ(\ba_i^\top\bv-\tau_i)\big) \nn\\&\quad   - \mathbbm{P}\left(\min_{j\in[L-1]}|\ba_i^\top \bu-\tau_i-b_j|<c^*\min\{\|\bu-\bv\|_2,\Delta\vee\Lambda\}\right)\nn\\
    &\quad -\mathbbm{P}\left(\min_{j\in[L-1]}|\ba_i^\top \bv-\tau_i-b_j|<c^*\min\{\|\bu-\bv\|_2,\Delta\vee\Lambda\}\right)\nn\\
    &\ge c^{(2)}\min\Big\{\frac{\|\bu-\bv\|_2}{\Delta\vee\Lambda},1\Big\}  - \frac{2c^{(1)}c^*\min\{\|\bu-\bv\|_2,\Delta\vee\Lambda\}}{\Delta\vee\Lambda}\nn \\
    &\ge \frac{c^{(2)}}{2}\min\Big\{\frac{\|\bu-\bv\|_2}{\Delta\vee\Lambda},1\Big\} .\label{eq:wellprob} 
\end{align} 
For any $(\bu,\bv)\in \calW\times\calW$ obeying $\|\bu-\bv\|_2\ge 2\epsilon$ and the corresponding $\bp,\bq$ found by (\ref{eq:foundpq}), by $\epsilon'<\frac{\epsilon}{2}$ we have 
\begin{align}\label{eq:chain_ine1}
    \|\bp-\bq\|_2 \ge \|\bu-\bv\|_2 - \|\bu-\bp\|_2-\|\bv-\bq\|_2 \ge \frac{\|\bu-\bv\|_2}{2} \ge \epsilon, 
\end{align}
where the second inequality holds because $\|\bu-\bp\|_2+\|\bv-\bq\|_2\le 2\epsilon'<\epsilon\le \frac{1}{2}\|\bu-\bv\|_2$. Then by Lemma \ref{lem:wellimplyse}, along with $\min\{\|\bp-\bq\|_2,\Delta\vee\Lambda\}\ge \epsilon$, we proceed with the following deterministic arguments:
\begin{subequations}
    \begin{align}\nn
    &d_H\big(\calQ(\bA\bu-\btau),\calQ(\bA\bv-\btau)\big)\\
    &\ge \sum_{i=1}^m \mathbbm{1}\Big(\calQ(\langle\ba_i,\cdot\rangle-\tau_i)~\text{$c^*$-well-distinguishes }\bp\text{ and }\bq,~\max\{|\ba_i^\top(\bu-\bp)|,|\ba_i^\top(\bv-\bq)|\}\le \frac{c^*\epsilon}{2}\Big)\nn\\
    &\ge \underbrace{\sum_{i=1}^m \mathbbm{1}\big(\calQ(\langle\ba_i,\cdot\rangle-\tau_i)~\text{$c^*$-well-distinguishes }\bp\text{ and }\bq\big)}_{:=\breve{T}_3^{\bp,\bq}}\label{eq:lower}\\&\qquad\qquad\qquad- \sum_{i=1}^m \mathbbm{1}\Big(|\ba_i^\top(\bu-\bp)|\ge\frac{c^*\epsilon}{2}\Big)-\sum_{i=1}^m \mathbbm{1}\Big(|\ba_i^\top(\bv-\bq)|\ge\frac{c^*\epsilon}{2}\Big).\label{eq:upup}
\end{align}\label{eq:lowerlocal_quan}
\end{subequations}
\paragraph{Bound (\ref{eq:lower}):} \rev{Recall that we denote the term in (\ref{eq:lower}) by 
$\breve{T}_3^{\bp,\bq}$,} and now we seek to lower bound $\breve{T}_3^{\bp,\bq}$
uniformly for all $(\bp,\bq)\in \calU_{\epsilon,\epsilon'}:=\{(\bp,\bq)\in\calN_{\epsilon'}\times\calN_{\epsilon'}:\|\bp-\bq\|_2\ge \epsilon\}$. We first achieve this for a fixed pair $(\bp,\bq)\in\calU_{\epsilon,\epsilon'}$. By (\ref{eq:wellprob})   and Chernoff bound we obtain that, for fixed $(\bp,\bq)\in\calU_{\epsilon,\epsilon'}$, 
    \begin{align*}
    &\mathbbm{P}\Big(\breve{T}_3^{\bp,\bq}\le \frac{c^{(2)}m}{4}\min\Big\{\frac{\|\bp-\bq\|_2}{\Delta\vee\Lambda},1\Big\}\Big)\\
    &\le \mathbbm{P}\left(\text{Bin}\Big(m,\frac{c^{(2)}}{2}\min\Big\{\frac{\|\bp-\bq\|_2}{\Delta\vee\Lambda},1\Big\}\Big)\le \frac{c^{(2)}m}{4}\min\Big\{\frac{\|\bp-\bq\|_2}{\Delta\vee\Lambda},1\Big\}\right)\\
    &\le \exp\left(-\frac{c^{(2)}m}{24}\min\Big\{\frac{\|\bp-\bq\|_2}{\Delta\vee\Lambda},1\Big\}\right)
    \\&\le \exp\Big(-\frac{c^{(2)}m\epsilon}{24(\Delta\vee\Lambda)}\Big).
\end{align*}
Taking a union bound,     
$
    \breve{T}_3^{\bp,\bq}> \frac{c^{(2)}m}{8}\min\big\{\frac{\|\bp-\bq\|_2}{\Delta\vee\Lambda},1\big\}
$ holds for any $(\bp,\bq)\in\calU_{\epsilon,\epsilon'} $
  with probability at least $$1-\exp\Big(2\scrH(\calW,\epsilon')-\frac{c^{(2)}m\epsilon}{24(\Delta\vee\Lambda)}\Big)\ge 1-\exp\Big(-\frac{c^{(2)}m\epsilon}{48(\Delta\vee\Lambda)}\Big),$$ with the proviso that $m\ge \frac{96(\Delta\vee\Lambda) }{c^{(2)}\epsilon}\scrH(\calW,\epsilon')$; note that this is assumed in (\ref{eq:local_scale}).

\paragraph{Bound (\ref{eq:upup}):} Next, we bound the terms in (\ref{eq:upup}). By $\bu-\bp,\bv-\bq\in \calW_{(3\epsilon'/2)}$, we have 
\begin{align*}
    \sum_{i=1}^m \mathbbm{1}\Big(|\ba_i^\top(\bu-\bp)|\ge\frac{c^*\epsilon}{2}\Big)+\sum_{i=1}^m \mathbbm{1}\Big(|\ba_i^\top(\bv-\bq)|\ge\frac{c^*\epsilon}{2}\Big)\le 2\sup_{\bw\in\calW_{(3\epsilon'/2)}}\sum_{i=1}^m\Big(|\ba_i^\top\bw|\ge \frac{c^*\epsilon}{2}\Big):=2\breve{T}_4.
\end{align*}
By arguments parallel to those for bounding $\breve{T}_2$, under (\ref{eq:local_scale}) and $\epsilon'= \frac{c_1\epsilon}{\log^{1/2}(\frac{\Delta\vee\Lambda}{\epsilon})}$ with some $c_1,C_2$ depending on $c^{(1)}$ and $c^{(2)}$, we can show $2\breve{T}_4\le \frac{c^{(2)}m\epsilon}{16(\Delta\vee\Lambda)}$ with probability at least $1-\exp(-\frac{c''m\epsilon}{\Delta\vee\Lambda})$ for some $c''$ depending on $c^{(1)}$ and $c^{(2)}$.

Hence, we can substitute the bounds on $\breve{T}_3^{\bp,\bq}$ and $\breve{T}_4$ into (\ref{eq:lowerlocal_quan}), yielding the following for any $(\bu,\bv)\in\calW\times\calW$ obeying $\|\bu-\bv\|_2\ge 2\epsilon$:
    \begin{align*}
    &d_H\big(\calQ(\bA\bu-\btau),\calQ(\bA\bv-\btau)\big)\\&\ge \frac{c^{(2)}m}{4}\min\Big\{\frac{\|\bp-\bq\|_2}{\Delta\vee\Lambda},1\Big\} - \frac{c^{(2)}m\epsilon}{16(\Delta\vee\Lambda)}\\
    & \ge \frac{c^{(2)}m}{8}\min\Big\{\frac{\|\bp-\bq\|_2}{\Delta\vee\Lambda},1\Big\}\\&\ge \frac{c^{(2)}m}{16}\min\Big\{\frac{\|\bu-\bv\|_2}{\Delta\vee\Lambda},1\Big\},
\end{align*}
where the last inequality follows from $\|\bp-\bq\|_2\ge \frac{1}{2}\|\bu-\bv\|_2$, see (\ref{eq:chain_ine1}). We have shown that the desired $E_{\rm large}$ holds with the promised probability. 
\end{proof}
\subsection{The Proof of Theorem \ref{thm:itup} (Information-theoretic upper bound)}\label{app:itupper}
\begin{proof}
    Setting $\calW=\calX=\calK\cap \mathbbm{A}_\alpha^\beta$ in Theorem \ref{thm:local_embed} and observing   $\calW_{(3r'/2)}=(\calW-\calW)\cap (\frac{3r'}{2}\mathbb{B}_2^n)\subset (\calK-\calK)\cap (\frac{3r'}{2}\mathbb{B}_2^n) \subset \frac{3}{2}\calK_{(r')}$, we find that under the assumed sample size, we have 
      $$d_H\big(\calQ(\bA\bu-\btau),\calQ(\bA\bv-\btau)\big)>0,\quad \forall (\bu,\bv)\in \calX~~\text{obeying}~~\|\bu-\bv\|_2\ge 2r$$
      with the promised probability. On the other hand, by the definition of HDM, for any $\bx\in \calX$, we have $\hat{\bx}_{\rm hdm}\in \calX$ and 
      $d_H\big(\calQ(\bA\hat{\bx}_{\rm hdm}-\btau),\by\big)\le d_H\big(\calQ(\bA\bx-\btau),\by\big)=0,$ giving $d_H\big(\calQ(\bA\hat{\bx}_{\rm hdm}-\btau),\calQ(\bA\bx -\btau)\big)=0.$ We therefore arrive at $\|\hat{\bx}_{\rm hdm}-\bx\|_2\le 2r$ ($\forall \bx\in\calX$), as claimed.
\end{proof}
\section{The Proof of Theorem \ref{thm:convergence} (RAIC Implies Convergence)}\label{app:convergence}
\begin{proof} 
We define a sequence $\{\varphi_t\}_{t=0}^\infty$ by the initial value $\varphi_0=\mu_4$ and the recurrence 
\begin{align}
    \varphi_{t} = 2\mu_1 \kappa_\alpha \varphi_{t-1} + 2 \kappa_\alpha\sqrt{\mu_2 \varphi_{t-1}} + 2\kappa_\alpha\mu_3,\quad t=1,2,\cdots .
\end{align}
Notice that $\bar{\kappa}$ defined in (\ref{barkappa}) is the unique fixed point of the above  recurrence. Together with $\bar{\kappa}<\mu_4$, we know that $\{\varphi_t\}_{t=0}^\infty$ is monotonically decreasing and converges to $\bar{\kappa}$ when $t\to\infty$. We omit the proof of this simple fact.

Now we prove that $\bx^{(t)}\in 2\calK $ and $\|\bx^{(t)}-\bx\|_2\le \varphi_t$ for any integer $t\ge 0$. We notice that this is trivially true for $t=0$ and use induction: suppose that $\bx^{(t-1)}\in 2\calK$ and $\|\bx^{(t-1)}-\bx\|_2 \le \varphi_{t-1}$ for some $t\ge 1$, we seek to show $\bx^{(t)}\in 2\calK$ and $\|\bx^{(t)}-\bx\|_2\le \varphi_t$. Recall that $\tilde{\bx}^{(t)}=\calP_{\calK}(\bx^{(t-1)}-\eta\cdot \bh(\bx^{(t-1)},\bx))$ and $\bx^{(t)}=\calP_{\mathbbm{A}_\alpha^\beta}(\tilde{\bx}^{(t)})$. Notice that $\bx^{(t-1)}\in\mathbbm{A}_\alpha^\beta$ and hence $\bx^{(t-1)}\in\overline{\calX}=(2\calK)\cap \mathbbm{A}_\alpha^\beta$, and that $\|\bx^{(t-1)}-\bx\|_2\le \varphi_{t-1}\le \mu_4$. 
Then by Lemma \ref{planlem} and the RAIC (set $\phi=\mu_3$), we have 
\begin{align}
    &\|\tilde{\bx}^{(t)}-\bx\|_2 = \big\|\calP_{\calK}\big(\bx^{(t-1)}-\eta \cdot\bh(\bx^{(t-1)},\bx)\big)-\bx\big\|_2\\
    &\le \max \left\{\mu_3, \frac{2}{\mu_3}\big\|\bx^{(t-1)}-\bx-\eta\cdot\bh(\bx^{(t-1)},\bx)\big\|_{\calK_{(\mu_3)}^\circ}\right\}\\
    & \le 2\mu_1 \|\bx^{(t-1)}-\bx\|_2 +2\sqrt{\mu_2\|\bx^{(t-1)}-\bx\|_2} + 2\mu_3\\
    &\le 2\mu_1 \varphi_{t-1}+ 2\sqrt{\mu_2\varphi_{t-1}}+2\mu_3.\label{useful11} 
\end{align}

Let us first show $\bx^{(t)}\in 2\calK$. Note that $\tilde{\bx}^{(t)}\in\calK$ and $\bx^{(t)}=\calP_{\mathbbm{A}_\alpha^\beta}(\tilde{\bx}^{(t)})=\gamma \tilde{\bx}^{(t)}$ for   $\gamma= \frac{\|\bx^{(t)}\|_2}{\|\tilde{\bx}^{(t)}\|_2}$. Thus, $\bx^{(t)}\in2\calK$ is evident when $\calK$ is a cone. If $\alpha=0$, then the projection onto $\mathbbm{A}_{\alpha}^\beta$ cannot increase norm, hence $\gamma\le 1$ and we obtain $\bx^{(t)}\in\calK$. If $\calK$ is not a cone and $\alpha>0$, we will need (\ref{eq:addit}) to ensure $\bx^{(t)}\in 2\calK$. Note that we can consider only the case when the projection onto $\mathbbm{A}_{\alpha}^\beta$ increases norm (i.e., $\|\bx^{(t)}\|_2\ge \|\tilde{\bx}^{(t)}\|_2$), and in this case we necessarily have $\|\bx^{(t)}\|_2=\alpha$. By (\ref{useful11}) and $\varphi_{t-1}\le\mu_4$ we have $\|\tilde{\bx}^{(t)}-\bx\|_2\le 2\mu_1\mu_4+2\sqrt{\mu_2\mu_4}+2\mu_3\le \frac{\alpha}{2}$. Therefore,  we have $\|\tilde{\bx}^{(t)}\|_2\ge \frac{\alpha}{2}$ and $\gamma= \frac{\|\bx^{(t)}\|_2}{\|\tilde{\bx}^{(t)}\|_2}\le 2$, which gives $\bx^{(t)}\in 2\calK$. 
Moreover, we shall show $\|\bx^{(t)}-\bx\|_2\le \varphi_t$. Note that $\kappa_\alpha=1$ when $\mathbbm{A}_{\alpha}^\beta$ is convex and $\kappa_\alpha=2$ when it is not, so by a well-known bound for projection (e.g., \cite[Lemma 15]{soltanolkotabi2019structured}), we have 
\begin{align}
    \|\bx^{(t)}-\bx\|_2\le \kappa_\alpha \|\tilde{\bx}^{(t)}-\bx\|_2\le 2\kappa_\alpha\mu_1\varphi_{t-1}+2\kappa_\alpha\sqrt{\mu_2\varphi_{t-1}}+2\kappa_\alpha\mu_3=\varphi_{t},
\end{align}
as desired. The induction is complete. Due to $\|\bx^{(t)}-\bx\|_2\le\varphi_t$, all that remains is to control $\{\varphi_t\}_{t=0}^\infty$.
\subsubsection*{Linear Convergence}
By using $$ 2\kappa_{\alpha} \sqrt{\mu_2 \varphi_{t-1}} \le \Big(\frac{1}{2}-\kappa_\alpha \mu_1\Big)\varphi_{t-1} + \frac{\kappa_{\alpha}^2 \mu_2}{\frac{1}{2}-\kappa_\alpha\mu_1},$$
we obtain 
$$\varphi_t \le \Big(\kappa_\alpha\mu_1+\frac{1}{2}\Big)\varphi_{t-1} + 2\kappa_\alpha \Big(\frac{\kappa_\alpha\mu_2}{1-2\kappa_\alpha\mu_1}+\mu_3\Big).$$
Due to $\kappa_\alpha\mu_1<\frac{1}{2}$, by initializing with $\varphi_0=\mu_4$ and iterating the above inequality, we obtain 
\begin{align}
    \varphi_t \le \Big(\kappa_\alpha\mu_1+\frac{1}{2}\Big)^t \mu_4 + 4\kappa_\alpha \left[\frac{\kappa_\alpha\mu_2}{(1-2\kappa_\alpha\mu_1)^2}+\frac{\mu_3}{1-2\kappa_\alpha\mu_1}\right]. 
\end{align}
The claim immediately follows. 

\subsubsection*{\rev{Faster} Convergence}
Combining $\mu_1\le \sqrt{\frac{\mu_2}{\mu_4}}+\frac{\mu_3}{\mu_4}$ and $\varphi_t\le \mu_4$, we have 
\begin{align*}
    \varphi_t\le 2\Big(\sqrt{\frac{\mu_2}{\mu_4}}+\frac{\mu_3}{\mu_4}\Big)\kappa_\alpha\varphi_{t-1} + 2\kappa_\alpha\sqrt{\mu_2\varphi_{t-1}} + 2\kappa_\alpha\mu_3 \le 4\kappa_\alpha \big(\sqrt{\mu_2\varphi_{t-1}}+\mu_3\big). 
\end{align*}
We observe that $\hat{\kappa}$ is the unique solution to the equation $x= 4\kappa_\alpha(\sqrt{\mu_2x}+\mu_3)$. We also construct $$\omega_t:=\mu_4^{2^{-t}}\hat{\kappa}^{1-2^{-t}},\quad\forall t\ge 0 $$ 
and notice that $\{\omega_t\}_{t=0}^\infty$ obeys $\omega_0=\mu_4$, $\omega_t = \sqrt{\hat{\kappa}\omega_{t-1}}$, $\omega_t\ge \hat{\kappa}$. Now let us show $\varphi_t\le \omega_t$ for any integer $t\ge 0$ by induction. This evidently holds when $t=0$. Suppose $\varphi_{t-1}\le \omega_{t-1}$ for some $t\ge 1$, then we proceed as 
\begin{align*}
    \varphi_{t} \le  4\kappa_{\alpha}(\sqrt{\mu_2\varphi_{t-1}}+ \mu_3)\le  4\kappa_{\alpha}(\sqrt{\mu_2\omega_{t-1}}+ \mu_3) \le \sqrt{\hat{\kappa}\omega_{t-1}}= \omega_{t},
\end{align*}
where the second inequality holds because $\omega_{t-1}\ge \hat{\kappa}$ and hence   $(\sqrt{\hat{\kappa}}-4\kappa_{\alpha}\sqrt{\mu_2})\sqrt{\omega_{t-1}}\ge\hat{\kappa} -4\kappa_{\alpha}\sqrt{\mu_2\hat{\kappa}}=4\kappa_{\alpha}\mu_3.$
 The induction is complete and we obtain $\|\bx^{(t)}-\bx\|_2\le \varphi_t\le \omega_t$, as claimed.      
\end{proof}
\section{The Proof of Theorem \ref{thm:main} (Main Result)}\label{app:provemain}
\begin{proof}
 The constants in this proof, including those hidden behind $\lesssim$, $O(\cdot)$, $\gtrsim$, $\Omega(\cdot)$ and $\asymp$, are allowed to depend on $(c^{(i)})_{i=1}^9$ and $(\varepsilon^{(i)})_{i=1}^4$.
 Recall that $\calN_r$ is a minimal $r$-net of $\overline{\calX}$, for any $\bu,\bv\in\overline{\calX}$ obeying $\|\bu-\bv\|_2\le \mu_4$, we find $\bu_1,\bv_1$ as the points in $\calN_r$ being closest to $\bu,\bv$  (see (\ref{eq:findu1v1})) and we have $(\bu_1,\bv_1)\in\calN_{r,\mu_4}^{(2)}$. 

Under $\eta=\frac{\Delta\vee\Lambda}{c^{(6)} \Delta }$, we combine (\ref{eq:large_decompose}), (\ref{eq:clipping1}), (\ref{eq:boundclip1}) and (\ref{eq:large_further_decom}) to obtain 
     \begin{align}
    &\| \bu-\bv-\eta\cdot\bh(\bu,\bv)\|_{\calK_{(r)}^\circ}\nn \\&\nn\le \|\bu_1-\bv_1-\eta\cdot \hat{\bh}(\bu_1,\bv_1)\|_{\calK_{(r)}^\circ}\\&\quad + \eta\cdot \|\hat{\bh}(\bu_1,\bv_1)-\bh(\bu_1,\bv_1)\|_{\calK_{(r)}^\circ}\nn +  2r^2+ 2\eta \cdot\sup_{(\bp,\bq)\in\calD_r^{(2)}}\|\bh(\bp,\bq)\|_{\calK_{(r)}^\circ} \nn
    \\\nn
    &\le r\big|\|\bu_1-\bv_1\|_2-T_1^{\bu_1,\bv_1}\big|+r|T_2^{\bu_1,\bv_1}|+T_3^{\bu_1,\bv_1}+\frac{r\varepsilon^{(1)}}{c^{(6)}}\|\bu_1-\bv_1\|_2\\&\quad+\Big(2+\frac{c^{(3)}}{c^{(6)}}\Big)r^2+  2\eta \cdot\sup_{(\bp,\bq)\in\calD_r^{(2)}}\|\bh(\bp,\bq)\|_{\calK_{(r)}^\circ}. \label{eq:seekraic}
\end{align}

\paragraph{Bound $r|\|\bu_1-\bv_1\|_2-T_1^{\bu_1,\bv_1}|+r|T_2^{\bu_1,\bv_1}|+T_3^{\bu_1,\bv_1}$:} Note that our assumed sample complexity (\ref{eq:sample_main}) implies $$\frac{\omega^2(\calK_{(r)})/r^2 + \scrH(\overline{\calX},r)}{m}\lesssim \frac{r}{\Delta\vee\Lambda}.$$
Therefore, with promised probability Propositions \ref{pro2}--\ref{pro3} yield
\begin{align*}
    \big|\|\bu_1-\bv_1\|_2-T_1^{\bu_1,\bv_1}\big|+ |T_2^{\bu_1,\bv_1}| \le \frac{C_1}{c^{(6)}}   \sqrt{ r(\Delta\vee\Lambda)\sfP_{\bu_1,\bv_1}}    +  \frac{ (\varepsilon^{(2)}+\varepsilon^{(3)})\|\bu_1-\bv_1\|_2+C_1r}{c^{(6)}} ,&
\end{align*}
and Proposition \ref{pro4} gives
\begin{align*}
    T_3^{\bu_1,\bv_1}\le \frac{C_2 r}{c^{(6)}}\sqrt{r\big((\Delta\vee\Lambda)\sfP_{\bu_1,\bv_1}+r\big) }+ \frac{5r  \varepsilon^{(4)}\|\bu_1-\bv_1\|_2}{c^{(6)}}\left(1+O\Big(\frac{r}{(\Delta\vee\Lambda)\sfP_{\bu_1,\bv_1}}\Big)\right).
\end{align*}
By $\|\bu_1-\bv_1\|_2\le C_0(\Delta\vee\Lambda)$ and Assumptions \ref{assum3}, \ref{assum_Puvup} we have $\sfP_{\bu_1,\bv_1}\asymp \frac{\|\bu_1-\bv_1\|_2}{\Delta\vee\Lambda}$. Taken collectively with  $\sum_{j=1}^4\varepsilon^{(j)}\lesssim 1$ implied by (\ref{eq:eta_condition}), we obtain 
     \begin{align}
    &r\big|\|\bu_1-\bv_1\|_2-T_1^{\bu_1,\bv_1}\big| +r|T_2^{\bu_1,\bv_1}|+T_3^{\bu_1,\bv_1}\nn \\
    &\le  r\left(   \frac{(\varepsilon^{(2)}+\varepsilon^{(3)}+5\varepsilon^{(4)})\|\bu_1-\bv_1\|_2}{c^{(6)}}    +C_3 \sqrt{r\|\bu_1-\bv_1\|_2}+C_4r\right)\nn\\
    &\le r\left(  \frac{(\varepsilon^{(2)}+\varepsilon^{(3)}+5\varepsilon^{(4)})\|\bu-\bv\|_2}{c^{(6)}}  +C_3 \sqrt{r\|\bu-\bv\|_2}+C_5r\right).\label{eq:uniform_u1v1}
\end{align} 

\paragraph{Bound $\sup_{(\bp,\bq)\in\calD_r^{(2)}}\|\bh(\bp,\bq)\|_{\calK_{(r)}^\circ}$:} We combine (\ref{eq:clipping2}), (\ref{eq:boundclip2}) and Proposition \ref{pro:boundsmall} to obtain 
\begin{align}
    \label{eq:raicbound2}\eta\cdot\|\bh(\bp,\bq)\|_{\calK^\circ_{(r)}} \le \eta \cdot \| \hat{\bh}(\bp,\bq)\|_{\calK^\circ_{(r)}} + c_4r^2 \le C_4r^2\log^{1/2}\Big(\frac{\Delta\vee\Lambda}{r}\Big)
\end{align}
that holds for all $(\bp,\bq)\in \calD_r^{(2)}$ with the promised probability.

\paragraph{Establish RAIC:} Substituting  (\ref{eq:uniform_u1v1}) and (\ref{eq:raicbound2}) into (\ref{eq:seekraic}) and combining with  $$\frac{r \varepsilon^{(1)}\|\bu_1-\bv_1\|_2}{c^{(6)}}\le \frac{r\varepsilon^{(1)}\|\bu-\bv\|_2}{c^{(6)}}+\frac{2 \varepsilon^{(1)}r^2}{c^{(6)}}=\frac{r \varepsilon^{(1)}\|\bu-\bv\|_2}{c^{(6)}}+O(r^2),$$  we arrive at the following that holds $\forall \bu,\bv\in\overline{\calX}\text{ obeying }\|\bu-\bv\|_2\le \mu_4$: 
     \begin{align*}
   & \frac{\| \bu-\bv-\eta\cdot\bh(\bu,\bv)\|_{\calK_{(r)}^\circ}}{r} \le  \frac{ (\sum_{j=1}^3\varepsilon^{(j)}+5\varepsilon^{(4)})\|\bu-\bv\|_2 }{c^{(6)}}  +C_3 \sqrt{r\|\bu-\bv\|_2} + C_6 r\log^{1/2}\Big(\frac{\Delta\vee\Lambda}{r}\Big).
\end{align*}
 This shows that, with the promised probability, $(\calQ,\bA,\btau,\calK,\eta)$ respects $(\overline{\calX},\bm{\mu})$-RAIC at the scale of $r$ with $\bm{\mu}$ given in (\ref{eq:bmmu}). 
 
 \paragraph{Convergence:} With RAIC and     (\ref{eq:eta_condition}) and (\ref{mu1mu2mu3}) that we assume to fulfill $\mu_1<\frac{1}{2\kappa_\alpha}$ and (\ref{eq:addit}), the desired guarantee immediately follows from Theorem \ref{thm:convergence}.  
\end{proof}

\section{The Proof of Theorem \ref{thm:conver_qpe} (Complementary approach)}\label{app:provecomplement}
\begin{proof}
    Under $\eta = \zeta^{-1}$ and $m\gtrsim \frac{\omega^2(\calK_{(r_1)})}{r_1^2}$, we substitute (\ref{boundfir}) and $\hat{T}\le  r_1$ into (\ref{eq:qpe2raic}) and obtain that $(\calQ,\bA,\btau,\calK,\eta)$ respects $(\overline{\calX},\bm{\mu})$-RAIC at the scale $r_1$ with $$\bm{\mu}=(\mu_1,\mu_2,\mu_3,\mu_4)^\top=\left( \frac{C_1\omega(\calK_{(r_1)}) }{r_1\sqrt{m}}, 0,3r_1, 2\beta\right)^\top,$$ for some absolute constant $C_1$. Note that we have $\mu_1<\frac{1}{8}$ and $r_1\lesssim \beta$, so by the proof of Theorem \ref{thm:convergence} we have $\|\bx^{(t)}-\bx\|_2\le\varphi_t$, where $\{\varphi_t\}_{t=0}^\infty$ is defined by $\varphi_0=2\beta$ and the recurrence $\varphi_t= 4\mu_1\varphi_{t-1}+12r_1.$
   It is easy to see $\varphi_t\le(4\mu_1)^t\varphi_{0} + \frac{12r_1}{1-4\mu_1} $. Substituting $\mu_1=\frac{C_1\omega_{(r_1)}}{r_1\sqrt{m}}$ and $\mu_1<\frac{1}{8}$, the claim follows. 
\end{proof}
\section{The Proofs in Sections \ref{sec:sharp_ana}--\ref{sec:raic_qpe} (Unified Analysis)}\label{app:unified11}
We collect the missing proofs for our unified analysis for Algorithm \ref{alg:pgd}.
\subsection{The Proof of Proposition \ref{pro:boundsmall} (Bounding $\|\hat{\bh}(\bp,\bq)\|_{\calK_{(r)}^\circ}$)}\label{app:prove_pro1}
\begin{proof}
The constants in this proof, including those hidden behind $\gtrsim$ and $\lesssim$, are allowed to depend on $c^{(1)}$ in Assumption \ref{assum2}. We invoke the first statement in Theorem \ref{thm:local_embed} to uniformly bound $|\bR_{\bp,\bq}|$ over $(\bp,\bq)\in\calD_r^{(2)}$. Specializing         Theorem \ref{thm:local_embed} to $(\epsilon',\epsilon)=(2r,r'')$ with $r'':= C_1r\log^{1/2}(\frac{\Delta\vee\Lambda}{r})$ for some sufficiently large $C_1$ and $\calW = \overline{\calX}=2\calK\cap \mathbbm{A}_\alpha^\beta$, we obtain that
   with probability at least $1-2\exp(-\frac{c_2mr}{\Delta\vee\Lambda})$, we have
\begin{align}\label{eq:uniRpq}
    |\bR_{\bp,\bq}|=d_H(\calQ(\bA\bp-\btau),\calQ(\bA\bq-\btau))\le \frac{c_3mr''}{\Delta\vee\Lambda},~~\forall (\bp,\bq)\in\calD_r^{(2)},
\end{align}
as long as 
\begin{align}\label{eq:needlocal}
    m\ge C_4(\Delta\vee\Lambda)\left(\frac{\omega^2(\overline{\calX}_{(3r)})}{(r'')^3}+\frac{\scrH(\overline{\calX},2r)}{r''}\right).
\end{align}
Due to $\overline{\calX}\subset 2\calK$, we have $\overline{\calX}_{(3r)} = (\overline{\calX}-\overline{\calX})\cap\mathbb{B}_2^n(3r)\subset (2\calK-2\calK)\cap \mathbb{B}_2^n(3r)\subset 3\calK_{(r)}$, together with $r''=C_1r\log^{1/2}(\frac{\Delta\vee\Lambda}{r})$,      (\ref{eq:needlocal}) can be ensured by (\ref{eq:sample_pro1}) we assumed.

On the event (\ref{eq:uniRpq}), by  (\ref{eq:hathpq}), Cauchy-Schwarz inequality and Lemma \ref{lem:max_ell_sum} we could proceed uniformly for all $(\bp,\bq)\in\calD_r^{(2)}$ as follows: 
     \begin{align*}\nn
    &\frac{1}{r}\big\|\hat{\bh}(\bp,\bq)\big\|_{\calK^\circ_{(r)}}=\frac{1}{r}\sup_{\bw\in\calK_{(r)}}\bw^\top\hat{\bh}(\bp,\bq) 
    \\& = \sup_{\bw\in\calK_{(r)}}\frac{\Delta}{mr}\sum_{i\in\bR_{\bp,\bq}}\sign(\ba_i^\top(\bp-\bq))\ba_i^\top \bw\\ \nn
    &\le \frac{c_3r''\Delta}{r(\Delta\vee\Lambda)}\sup_{\bw\in\calK_{(r)}}\max_{\substack{I\subset [m]\\|I|\le c_3mr''/(\Delta\vee\Lambda)}}\Big(\frac{c_3mr''}{\Delta\vee\Lambda}\Big)^{-1}\sum_{i\in I}|\ba_i^\top\bw| \\\nn
    &\le \frac{c_3r''\Delta}{r(\Delta\vee\Lambda)}\sup_{\bw\in\calK_{(r)}}\max_{\substack{I\subset [m]\\|I|\le c_3mr''/(\Delta\vee\Lambda)}}\left[\Big(\frac{c_3mr''}{\Delta\vee\Lambda}\Big)^{-1}\sum_{i\in I}|\ba_i^\top\bw| ^2\right]^{1/2}\\
    &\le \frac{C_5r''\Delta}{r(\Delta\vee\Lambda)}\left[\Big(\frac{(\Delta\vee\Lambda)\cdot\omega^2(\calK_{(r)})}{c_3mr''}\Big)^{1/2}+r\cdot\log^{1/2}\Big(\frac{\Delta\vee\Lambda}{r''}\Big)\right]\\&
    \le \frac{C_6 r''\Delta}{\Delta\vee\Lambda}\log^{1/2}\Big(\frac{\Delta\vee\Lambda}{r''}\Big) 
\end{align*}
 where  the first inequality in the last line holds with probability at least $1-2\exp(-\frac{c_7mr''}{\Delta\vee\Lambda})$ due to Lemma \ref{lem:max_ell_sum}, and the second inequality holds so long as $m\gtrsim \frac{(\Delta\vee\Lambda)\cdot\omega^2(\calK_{(r)})}{r''r^2\log(\frac{\Delta\vee\Lambda}{r''})}$, and this can be implied by (\ref{eq:sample_pro1}) in our statement due to $r''=C_1r\log^{1/2}(\frac{\Delta\vee\Lambda}{r})$. Combining everything and substituting $r''$ conclude the proof.      
\end{proof}

\subsection{The Proof of Proposition \ref{pro2} (Bounding $|\|\bp-\bq\|_2-T_1^{\bp,\bq}|$)}\label{app:pro2}
\begin{proof}
    
We shall first bound $|\|\bp-\bq\|_2-T_1^{\bp,\bq}|$ for fixed $(\bp,\bq)\in\calN^{(2)}_{r,\mu_4}$ and then take a union bound. Recall that $T_1^{\bp,\bq}=\frac{\eta\Delta}{m}\sum_{i=1}^m|\ba_i^\top\bbeta_1|\mathbbm{1}(E^{(i)}_{\bp,\bq})$, and so we have $\mathbbm{E}[T_1^{\bp,\bq}]=\eta\Delta \cdot\mathbbm{E}\big[|\ba_i^\top\bbeta_1|\mathbbm{1}(E^{(i)}_{\bp,\bq})\big]$. Therefore, triangle inequality and (\ref{eq:sg1_exp}) yield 
\begin{align}\label{center72}
    \big|\|\bp-\bq\|_2-T_1^{\bp,\bq}\big|\le \big|T_1^{\bp,\bq}-\mathbbm{E}[T_1^{\bp,\bq}]\big| + \left[\Big|1-\frac{\eta c^{(6)} \Delta }{\Delta\vee\Lambda}\Big|+\frac{\eta\Delta \varepsilon^{(2)}}{\Delta\vee\Lambda}\right]\cdot\|\bp-\bq\|_2.
\end{align}
Next, we use Bernstein's inequality (Lemma \ref{lem:chernoff}) to bound $$\big|T_1^{\bp,\bq}-\mathbbm{E}[T_1^{\bp,\bq}]\big|=\frac{\eta\Delta}{m}\left|\sum_{i=1}^m\Big(|\ba_i^\top\bbeta_1|\mathbbm{1}(E^{(i)}_{\bp,\bq})-\mathbbm{E}\big[|\ba_i^\top\bbeta_1|\mathbbm{1}(E^{(i)}_{\bp,\bq})\big]\Big)\right|.$$
For any integer $p\ge 2$, (\ref{eq:sg1_sg}) gives $$\sum_{i=1}^m\mathbbm{E}\Big||\ba_i^\top\bbeta_1|\mathbbm{1}(E^{(i)}_{\bp,\bq})\Big|^p\le \frac{p!}{2}c^{(5)}m\sfP_{\bp,\bq}.$$
Hence, for any $t>0$, Lemma \ref{lem:chernoff} yields 
\begin{align}
    \mathbbm{P}\left(\big|T_1^{\bp,\bq}-\mathbbm{E}[T_1^{\bp,\bq}]\big|\le  \eta\Delta \Big[\Big(\frac{2c^{(5)}\sfP_{\bp,\bq}t}{m}\Big)^{1/2}+\frac{t}{m}\Big]\right)\ge 1-2\exp(-t). 
\end{align}
We let $t=4 \scrH(\overline{\calX},r)$ and take a union bound over no more than $|\scrN(\overline{\calX},r)|^2$ pairs of $(\bp,\bq)\in \calN^{(2)}_{r,\mu_4}$, yielding that the uniform bound
\begin{align}
    \big|T_1^{\bp,\bq}-\mathbbm{E}[T_1^{\bp,\bq}]\big|\le 4\eta\Delta \left[\Big(\frac{c^{(5)}\sfP_{\bp,\bq}\scrH(\overline{\calX},r)}{m}\Big)^{1/2}+\frac{\scrH(\overline{\calX},r)}{m}\right],\quad\forall (\bp,\bq)\in\calN_{r,\mu_4}^{(2)}
\end{align}
holds with probability at least $1-2\exp(-2\scrH(\overline{\calX},r))$. Substituting this into (\ref{center72}) yields the claim. 
\end{proof}

\subsection{The Proof of Proposition \ref{pro3} (Bounding $|T_2^{\bp,\bq}|$)}\label{app:prove_pro2}
\begin{proof}
Recall that $T_2^{\bp,\bq}=\frac{\eta\Delta}{m}\sum_{i=1}^m\sign(\ba_i^\top\bbeta_1)(\ba_i^\top\bbeta_2)\mathbbm{1}(E^{(i)}_{\bp,\bq})$. Combining with (\ref{eq:sg2_exp}) we have  $$|\mathbbm{E}[T_2^{\bp,\bq}]|=\eta\Delta\big|\mathbbm{E}[\sign(\ba_i^\top\bbeta_1)(\ba_i^\top\bbeta_2)\mathbbm{1}(E^{(i)}_{\bp,\bq})]\big|\le \frac{\eta\Delta\varepsilon^{(3)}\|\bp-\bq\|_2}{\Delta\vee\Lambda}.$$
Therefore, by triangle inequality we reach
$$|T_2^{\bp,\bq}|\le |T_2^{\bp,\bq}-\mathbbm{E}[T_2^{\bp,\bq}]|+ \frac{\eta\Delta\varepsilon^{(3)}\|\bp-\bq\|_2}{\Delta\vee\Lambda}.$$
All that remains is to bound $|T_2^{\bp,\bq}-\mathbbm{E}[T_2^{\bp,\bq}]|$, which is parallel to bounding $|T_1^{\bp,\bq}-\mathbbm{E}[T_1^{\bp,\bq}]|$ in the proof of Proposition \ref{pro2}, following from the moment bound (\ref{eq:sg2_sg}) and Lemma \ref{lem:chernoff}. We omit the details.  \end{proof}

\subsection{The Proof of Proposition \ref{pro4} (Bounding $T_3^{\bp,\bq}$)}\label{app:pro4}
\begin{proof}
The idea is to first condition on $\bR_{\bp,\bq}$ to bound the random process over $\bw\in\calK_{(r)}$ and then get rid of the conditioning by analyzing $|\bR_{\bp,\bq}|$. Recall that we defined $J_{i,\bw}^{\bp,\bq}=\sign(\ba_i^\top\bbeta_1)[\ba_i-(\ba_i^\top\bbeta_1)\bbeta_1-(\ba_i^\top\bbeta_2)\bbeta_2]^\top\bw$ and this is linear on $\bw\in \mathbb{R}^n$, and we can write $$T_3^{\bp,\bq}=\sup_{\bw\in\calK_{(r)}}\frac{\eta\Delta}{m}\sum_{i\in\bR_{\bp,\bq}}J_{i,\bw}^{\bp,\bq}.$$ We already have noted that for $\bw\in\mathbb{S}^{n-1}$, the conditional distribution $J_{i,\bw}^{\bp,\bq}|E^{(i)}_{\bp,\bq}$ has $O(c^{(8)})$ sub-Gaussian norm; see (\ref{consg}).  

\subsubsection*{Conditioning on $|\bR_{\bp,\bq}|$}
For any integer $0\le r_{\bp,\bq}\le m$, $T_3^{\bp,\bq}|\{|\bR_{\bp,\bq}|=r_{\bp,\bq}\}$ is identically distributed as 
\begin{align*}
    \sup_{\bw\in\calK_{(r)}}\frac{\eta\Delta}{m}\sum_{i=1}^{r_{\bp,\bq}} {Z}_i^{\bp,\bq}(\bw),
\end{align*}
where $Z_1^{\bp,\bq}(\bw),Z_2^{\bp,\bq}(\bw),\cdots ,Z_{r_{\bp,\bq}}^{\bp,\bq}(\bw)$ are i.i.d. random variables following the conditional distribution $J_{i,\bw}^{\bp,\bq}|E^{(i)}_{\bp,\bq}$. Therefore, for any $\bw\in\mathbb{S}^{n-1}$, $\|Z_i^{\bp,\bq}(\bw)\|_{\psi_2}=O(c^{(8)})$ and hence $\|Z_i^{\bp,\bq}(\bw)-\mathbbm{E}[Z_i^{\bp,\bq}(\bw)]\|_{\psi_2}=O(c^{(8)})$ \cite[Lemma 2.6.8]{vershynin2018high}. Now we   center the random process to come to 
     \begin{align} 
   \sup_{\bw\in \calK_{(r)}}\frac{\eta\Delta}{m}\sum_{i=1}^{r_{\bp,\bq}}{Z}_i^{\bp,\bq}(\bw)\le \sup_{\bw\in \calK_{(r)}}\frac{\eta\Delta}{m}\sum_{i=1}^{r_{\bp,\bq}}\big[{Z}_i^{\bp,\bq}(\bw)-\mathbbm{E}\big({Z}_i^{\bp,\bq}(\bw)\big)\big]+\sup_{\bw\in\calK_{(r)}}\frac{\eta\Delta r_{\bp,\bq}\mathbbm{E}({Z}_i^{\bp,\bq}(\bw))}{m}.\label{147}
\end{align}

We seek to control the first term in the right-hand side of (\ref{147}). 
  For any $\bw_1,\bw_2\in \calK_{(r)}$ we let $\bw_3:=\frac{\bw_1-\bw_2}{\|\bw_1-\bw_2\|_2}$, then by the linearity of ${Z}_i^{\bp,\bq}(\bw)$ on $\bw$ we have 
\begin{align*}
   &\left\|\frac{\eta\Delta}{m}\sum_{i=1}^{r_{\bp,\bq}}\big[{Z}_i^{\bp,\bq}(\bw_1)-\mathbbm{E}\big({Z}_i^{\bp,\bq}(\bw_1)\big)\big]-\frac{\eta\Delta}{m}\sum_{i=1}^{r_{\bp,\bq}}\big[{Z}_i^{\bp,\bq}(\bw_2)-\mathbbm{E}\big({Z}_i^{\bp,\bq}(\bw_2)\big)\big]\right\|_{\psi_2}  \\
   &= \left\|\frac{\eta\Delta}{m}\sum_{i=1}^{r_{\bp,\bq}}\big[{Z}_i^{\bp,\bq}(\bw_3)-\mathbbm{E}\big({Z}_i^{\bp,\bq}(\bw_3)\big)\big]\right\|_{\psi_2}\|\bw_1-\bw_2\|_2\\&\le \frac{C_1c^{(8)}\eta\Delta \sqrt{r_{\bp,\bq}}\|\bw_1-\bw_2\|_2}{m},
\end{align*}
where the last inequality holds for some absolute constant $C_1$, due to \cite[Proposition 2.6.1]{vershynin2018high} and $\|Z_i^{\bp,\bq}(\bw)-\mathbbm{E}(Z_i^{\bp,\bq}(\bw))\|_{\psi_2}=O(c^{(8)})$ for $\bw\in \mathbb{S}^{n-1}$.

Therefore, by invoking Lemma \ref{lem:tala} with $K= \frac{C_1c^{(8)}\eta\Delta\sqrt{r_{\bp,\bq}}}{m}$ and $t = 2[\scrH(\overline{\calX},r)]^{1/2}$, we obtain that
\begin{align*}
    \sup_{\bw\in \calK_{(r)}}\left|\frac{\eta\Delta}{m}\sum_{i=1}^{r_{\bp,\bq}}\big[Z_i^{\bp,\bq}(\bw)-\mathbbm{E}\big(Z_i^{\bp,\bq}(\bw)\big)\big]\right|\le \frac{C_2c^{(8)}\eta\Delta\sqrt{r_{\bp,\bq}}}{m}\left[\omega(\calK_{(r)})+2r\sqrt{\scrH(\overline{\calX},r)}\right]   
\end{align*}
with probability at least $1-2\exp(-4\scrH(\overline{\calX},r))$. Due to $$T_3^{\bp,\bq}|\{|\bR_{\bp,\bq}|=r_{\bp,\bq}\}\sim \sup_{\bw\in\calK_{(r)}}\frac{\eta\Delta}{m}\sum_{i=1}^{r_{\bp,\bq}}Z_i^{\bp,\bq}(\bw)$$ and (\ref{147}), we reach the conditional concentration bound
\begin{align}\label{conditionalt3}
    \Big[T_3^{\bp,\bq}\big|\{|\bR_{\bp,\bq}|=r_{\bp,\bq}\} \Big]\le \frac{C_2c^{(8)}\eta\Delta\sqrt{r_{\bp,\bq}}}{m}\Big[\omega(\calK_{(r)})+2r\sqrt{\scrH(\overline{\calX},r)}\Big] +\sup_{\bw\in \calK_{(r)}}\frac{\eta\Delta r_{\bp,\bq}\mathbbm{E}(Z_i^{\bp,\bq}(\bw))}{m}   
\end{align}
that holds with probability at least $1-2\exp(-4\scrH(\overline{\calX},r))$.

\subsubsection*{Unconditional Bound}
We further analyze the behaviour of $|\bR_{\bp,\bq}|\sim\text{Bin}(m,\sfP_{\bp,\bq})$ to derive the final claim. Using Lemma \ref{lem:chernoff} (or Chernoff bound), for any $t>0$ we obtain 
\begin{align}
    \mathbbm{P}\Big(\big||\bR_{\bp,\bq}|-m\sfP_{\bp,\bq}\big|\le\sqrt{2m\sfP_{\bp,\bq}t}+t\Big)\ge 1-2\exp(-t). 
\end{align}
Therefore, with probability at least $1-2\exp(-4\scrH(\overline{\calX},r))$ we have \begin{align}\label{bernsteinrpq}
    |\bR_{\bp,\bq}|\le 5 \big(m\sfP_{\bp,\bq}+\scrH(\overline{\calX},r)\big).
\end{align} 
Therefore, we can combine (\ref{conditionalt3}) and (\ref{bernsteinrpq}), then take a union bound over $(\bp,\bq)\in\calN_{r,\mu_4}^{(2)}$, to obtain
\begin{align}\nn
    T_3^{\bp,\bq}&\le C_3c^{(8)}r\eta\Delta \sqrt{\Big[\sfP_{\bp,\bq}+\frac{\scrH(\overline{\calX},r)}{m}\Big]\Big[\frac{r^{-2}\omega^2(\calK_{(r)})+\scrH(\overline{\calX},r)}{m}\Big]}\\
    &+ 5\eta\Delta \left[1+\frac{\scrH(\overline{\calX},r)}{m\sfP_{\bp,\bq}}\right] \sup_{\bw\in\calK_{(r)}}\sfP_{\bp,\bq}\mathbbm{E}(Z_i^{\bp,\bq}(\bw)),\quad \forall (\bp,\bq)\in\calN_{r,\mu_4}^{(2)}\label{near_final}
\end{align}
with probability at least $1-4\exp(-2\scrH(\overline{\calX},r))$. Moreover, by the linearity of $Z_i^{\bp,\bq}(\bw)$ on $\bw$ and $Z_i^{\bp,\bq}(\bw)\sim \big(J^{\bp,\bq}_{i,\bw}|E^{(i)}_{\bp,\bq}\big)$ and (\ref{eq:sg3_exp}), we have
\begin{align*}
     &\sup_{\bw\in \calK_{(r)}}\sfP_{\bp,\bq}\mathbbm{E}(Z_i^{\bp,\bq}(\bw)) \\&\le r\sup_{\bw\in\mathbb{S}^{n-1}}\big|\sfP_{\bp,\bq}\mathbbm{E}(Z_i^{\bp,\bq}(\bw))\big|\\
   &= r\sup_{\bw\in\mathbb{S}^{n-1}}\Big|\mathbbm{P}(E^{(i)}_{\bp,\bq})\mathbbm{E}\big(J^{\bp,\bq}_{i,\bw}|E^{(i)}_{\bp,\bq}\big)\Big|  \\&=r \sup_{\bw\in\mathbb{S}^{n-1}}\Big|\mathbbm{E}\Big(J_{i,\bw}^{\bp,\bq}\mathbbm{1}(E^{(i)}_{\bp,\bq})\Big)\Big|
   \le \frac{r \varepsilon^{(4)}\|\bp-\bq\|_2}{\Delta\vee\Lambda}.
\end{align*}
Substituting this into (\ref{near_final}) yields the claim. 
\end{proof}

\subsection{The Proof of Proposition \ref{pro5} (Bounding the First Term in (\ref{eq:qpe2raic}))} \label{app:pro5}
\begin{proof}
By triangle inequality we  decompose the first term in (\ref{eq:qpe2raic}) into 
$$\underbrace{|\eta\zeta|\sup_{\bw\in\calK_{(r_1)}}\Big|\bw^\top\Big(\bI_n-\frac{1}{m}\sum_{i=1}^m\ba_i\ba_i^\top\Big)(\bu-\bv)\Big|}_{:=\hat{T}_1}+r_1|\eta\zeta -1|\|\bu-\bv\|_2.$$
To further control $\hat{T}_1$ above, we shall discuss ``large-distance regime'' and ``small-distance regime''. 
\paragraph{Large-distance regime ($\|\bu-\bv\|_2\ge r_1$):} Uniformly for all $\bu,\bv\in\overline{\calX}\subset 2\calK$ obeying $\|\bu-\bv\|_2\ge r_1$ we have $$\frac{\bu-\bv}{\|\bu-\bv\|_2}\in  \frac{2(\calK-\calK)}{r_1} \cap \mathbb{B}_2^n \subset \frac{2\calK_{(r_1)}}{r_1}.$$
Thus, we obtain  
\begin{align*}
   \hat{T}_1 \le \frac{2|\eta\zeta|\|\bu-\bv\|_2}{r_1}\sup_{\bw,\bs\in \calK_{(r_1)}}\Big|\frac{1}{m}\sum_{i=1}^m \bw^\top\ba_i\ba_i^\top\bs - \bw^\top\bs\Big|.
\end{align*}
\paragraph{Small-distance regime ($\|\bu-\bv\|_2<r_1$):} Uniformly for all $\bu,\bv$ obeying $\|\bu-\bv\|_2<r_1$, we can directly take the supremum over $\bu-\bv\in (2\calK)_{(r_1)}\subset 2\calK_{(r_1)}$ and obtain 
\begin{align*}
    \hat{T}_1 \le 2|\eta\zeta|\sup_{\bw,\bs\in \calK_{(r_1)}}\Big|\frac{1}{m}\sum_{i=1}^m \bw^\top\ba_i\ba_i^\top\bs - \bw^\top\bs\Big|.
\end{align*}
By Lemma \ref{lem:product_process}, if $m\ge\frac{\omega^2(\calK_{(r_1)})}{r_1^2}$, then for some absolute constants $C_1,c_2$, we have \begin{align}
    \label{productbb}
    \sup_{\bw,\bs\in\calK_{(r_1)}}\Big|\frac{1}{m}\sum_{i=1}^m\bw^\top\ba_i\ba_i^\top\bs-\bw^\top\bs\Big|\le \frac{C_1r_1 \omega(\calK_{(r_1)})}{\sqrt{m}}
\end{align}
with probability at least $1-\exp(-c_2\omega^2(\calK_{(r_1)})/r_1^2)$. Therefore, combining the two regimes and the bound (\ref{productbb}) yield the claim. 
\end{proof}

\section{The Proofs in Section \ref{sec:instance} (Validating Assumptions)}\label{app:verify}

\subsection{The Proof of Lemma \ref{lem:1bcssmallassum} (Assumptions \ref{assum1}--\ref{assump4}, \ref{assum_Puvup} for 1bCS)}\label{app:provelem1}
\begin{proof}
    Note that Assumption  \ref{assum1} is automatic for Gaussian $\bA$, and in the 1-bit sensing regime Assumption \ref{assump4} holds trivially with $\varepsilon^{(1)}=c^{(3)}=c^{(4)}=0$. Since $\ba_i^\top\bu\sim\calN(0,1)$ for any $\bu\in \mathbb{S}^{n-1}$, given $r\in(0,\frac{1}{2})$ we have 
 $$\mathbbm{P}\big(|\ba_i^\top\bu|\le r\big)\le \sqrt{2/\pi}\cdot r$$
 that verifies Assumption \ref{assum2} with $c^{(1)}=\sqrt{8/\pi}$. Moreover, it is well-documented  
 that $\sfP_{\bu,\bv}$ equals to the Geodesic distance:
 \begin{align}
     \sfP_{\bu,\bv}=\mathbbm{P}(\sign(\ba_i^\top\bu)\ne\sign(\ba_i^\top\bv)) = \frac{\arccos(\bu^\top\bv)}{\pi}.\label{eq:geodesic}
 \end{align}
  \rev{For the proof and history of (\ref{eq:geodesic}) one can refer to \cite{vershynin2018high} and the references therein.} 
  Then it follows from some algebra that
\begin{align}\label{eq:1bcs_Puv}
   \frac{1}{\pi}\|\bu-\bv\|_2\le \sfP_{\bu,\bv} \le \frac{1}{2}\|\bu-\bv\|_2,\quad \forall \bu,\bv\in\mathbb{S}^{n-1},
\end{align}
which establishes Assumption \ref{assum3} with $c^{(2)}=\frac{2}{\pi}$ and Assumption \ref{assum_Puvup} with $c^{(9)}=1$. 
\end{proof}

\subsection{The Proof of Lemma \ref{fac:1bcs_sg2} (Assumptions \ref{assum_sg1}--\ref{assum_sg3} for 1bCS)}\label{app:provefact2}
\begin{proof}
Continuing from the proof of Lemma \ref{fac:1bcs_sg2}, we have that $\ba_i^\top\bbeta_1:=a_1$ and $\ba_i^\top\bbeta_2:=a_2$ are independent $\calN(0,1)$ variables.

\subsubsection*{(i) Validating Assumption \ref{assum_sg1}}
First we show (\ref{eq:sg1_sg}) holds for some absolute constant $c^{(5)}$. As we observed, we can simply focus on small enough $u_1$, say, $|u_1|\le \frac{1}{2}$; this implies $|u_2|\ge\frac{1}{2}$ due to $u_1^2+u_2^2=\|\bp\|_2^2=1$. We note that $\ba_i^\top\bbeta_1:=a_1$ and $\ba_i^\top\bbeta_2:=a_2$ are independent $\calN(0,1)$ variables. (\ref{Eib1b2}) thus specializes to  
    $$E^{(i)}_{\bp,\bq}=\big\{\sign(u_1a_1+u_2a_2)\ne\sign(-u_1a_1+u_2a_2)\big\}=\{u_1|a_1|>u_2|a_2|\}.$$
    Therefore, for any $p\ge 2$ we have 
    \begin{align*}
        &\mathbbm{E}\Big(|\ba_i^\top\bbeta_1|^p\mathbbm{1}(E^{(i)}_{\bp,\bq})\Big) 
        \\&= \mathbbm{E}\Big(|a_1|^p\mathbbm{1}\big(|u_1a_1|>|u_2a_2|\big)\Big)\\
        &\le \mathbbm{E}\left(|a_1|^p\sqrt{\frac{2}{\pi}}\Big|\frac{u_1a_1}{u_2}\Big|\right)\\&\le |2u_1|\cdot \sqrt{\frac{2}{\pi}}\mathbbm{E}\big[|a_1|^{p+1}\big]\le c^{(5)}\sfP_{\bp,\bq}\cdot\frac{p!}{2}, 
    \end{align*}
where the first inequality is due to the randomness of $a_2$, and the last two inequalities hold because $|u_2|\ge\frac{1}{2}$ and $2u_1=\|\bp-\bq\|_2 \asymp \sfP_{\bp,\bq}$.

 Next we establish (\ref{eq:sg1_exp}) with $c^{(6)}=\sqrt{2/\pi}$ and $\varepsilon^{(2)}=0$. Note that the joint distribution of $(|a_1|,|a_2|)$ is $\frac{2}{\pi}\exp(-\frac{a_1^2+a_2^2}{2})$, we have
       \begin{align*}
          & \mathbbm{E}\Big(|\ba_i^\top\bbeta_1| \mathbbm{1}(E^{(i)}_{\bp,\bq})\Big)\\
           &=\mathbbm{E}\Big(|a_1| \mathbbm{1}\big(u_1|a_1|>u_2|a_2|\big)\Big)\\
           & = \frac{2}{\pi}\int_0^\infty \exp\Big(-\frac{w^2}{2}\Big)\left[\int_{u_2w/u_1}^\infty z\exp\Big(-\frac{z^2}{2}\Big)~\text{d}z\right]~\text{d}w\\
           & = \frac{2}{\pi}\int_0^\infty \exp\Big(-\frac{w^2}{2u_1^2}\Big) ~\text{d}w = \sqrt{\frac{2}{\pi}}u_1.
       \end{align*}
  The proof is complete. 

\subsubsection*{(ii) Validating Assumption \ref{assum_sg2}}
Recall that we can restrict our attention to $u_1\le \frac{1}{2}$ and $u_2\ge \frac{1}{2}$ in the validation of (\ref{eq:sg2_sg}). For any $p\ge 2$ we have
\begin{align*}
    &\mathbbm{E}\Big(|\ba_i^\top\bbeta_2|^p \mathbbm{1}(E^{(i)}_{\bp,\bq})\Big) \\&= \mathbbm{E}\Big(|a_2|^p\mathbbm{1}(u_1|a_1|>u_2|a_2|)\Big)\\
    &\le \mathbbm{E}\Big(|a_2|^{p-1}\cdot 2u_1|a_1|\Big)\\& = \|\bp-\bq\|_2 \sqrt{\mathbbm{E}[|a_2|^{2(p-1)}]}\\& \le \sfP_{\bp,\bq}(C_1\sqrt{p})^p\le C_2\sfP_{\bp,\bq}\cdot\frac{p!}{2}, 
\end{align*}
where in the first inequality we use the event $|a_2|<\frac{u_1|a_1|}{u_2}\le 2u_1|a_1|=\|\bp-\bq\|_2 |a_1|$. Therefore, (\ref{eq:sg2_sg}) holds with some absolute constant $c^{(7)}=C_2$.

In addition, (\ref{eq:sg2_exp}) holds with $\varepsilon^{(3)}=0$ because 
$$\mathbbm{E}\Big(\sign(\ba_i^\top\bbeta_1)\ba_i^\top\bbeta_2 \mathbbm{1}(E^{(i)}_{\bp,\bq})\Big)=\mathbbm{E}\Big(\sign(a_1)a_2 \mathbbm{1}(u_1|a_1|>u_2|a_2|)\Big)=0,$$
where we use the randomness of $\sign(a_1)$.

\subsubsection*{(iii) Validating Assumption \ref{assum_sg3}}
For any $\bw\in\mathbb{R}^n$ we have defined $J^{\bp,\bq}_{i,\bw}=\sign(\ba_i^\top\bbeta_1)[\ba_i-(\ba_i^\top\bbeta_1)^\top\bbeta_1-(\ba_i^\top\bbeta_2)^\top\bbeta_2]^\top\bw$, and we further define $$a_{\bw}:=[\ba_i-(\ba_i^\top\bbeta_1)^\top\bbeta_1-(\ba_i^\top\bbeta_2)^\top\bbeta_2]^\top\bw=\ba_i^\top(\bI_n-\bbeta_1\bbeta_1^\top-\bbeta_2\bbeta_2^\top)\bw.$$ Because $(\bI_n-\bbeta_1\bbeta_1^\top-\bbeta_2\bbeta_2^\top)\bw$ is orthogonal to $\bbeta_1,\bbeta_2$, the rotational invariance of $\ba_i\sim\calN(0,\bI_n)$ implies that $a_1,a_2,a_{\bw}$ are independent. In view of $E^{(i)}_{\bp,\bq}=\{u_1|a_1|>u_2|a_2|\}$, $a_\bw$ and $E^{(i)}_{\bp,\bq}$ are independent. Note that $(\bI_n-\bbeta_1\bbeta_1^\top-\bbeta_2\bbeta_2^\top)\bw\in\mathbb{B}_2^n$ when $\bw\in\mathbb{S}^{n-1}$. Therefore, for $p\ge 2$ we have
\begin{align*}
    \mathbbm{E}\Big(|J_{i,\bw}^{\bp,\bq}|^p\mathbbm{1}(E^{(i)}_{\bp,\bq})\Big) = \mathbbm{E}\big(|a_{\bw}|^p\mathbbm{1}(E^{(i)}_{\bp,\bq})\big) = \sfP_{\bp,\bq}\mathbbm{E}(|a_\bw|^p) \le \sfP_{\bp,\bq}(C_3\sqrt{p})^p,
\end{align*}
yielding (\ref{eq:sg3_sg}) with some absolute constant $c^{(8)}$.

Moreover, (\ref{eq:sg3_exp}) holds with $\varepsilon^{(4)}=0$ because we can use the randomness of $\sign(a_1)$ to obtain 
$$\mathbbm{E}\Big(J_{i,\bw}^{\bp,\bq}\mathbbm{1}(E^{(i)}_{\bp,\bq})\Big)=\mathbbm{E}\big(\sign(a_1)a_{\bw}\mathbbm{1}(u_1|a_1|>u_2|a_2|)\big)=0,$$
as desired. 
\end{proof}

\subsection{The Proof of Lemma \ref{lem:d1bcsearly} (Assumptions \ref{assum2}--\ref{assump4}, \ref{assum_Puvup} for D1bCS)}\label{app:proveearlyd1bcs}
\begin{proof}

For any $\bu\in\mathbb{B}_2^n$ and small enough $r>0$, we utilize the randomness of $\tau_i\sim \scrU([-\lambda,\lambda])$ to obtain 
$$\mathbbm{P}\big(|\ba_i^\top\bu-\tau_i|\le r\big) \le \frac{r}{\lambda},$$ which establishes Assumption \ref{assum2} with $c^{(1)}=1$. Also, Assumption \ref{assump4} holds trivially for 1-bit models with $\varepsilon^{(1)}=c^{(3)}=c^{(4)}=0$.

Now let us work on Assumptions \ref{assum3} and \ref{assum_Puvup}. Under $\lambda\ge \Delta=2$, we only need to find absolute constants $c_1$ and $C_2$ such that $$\frac{c_1\|\bu-\bv\|_2}{\lambda}\le \sfP_{\bu,\bv} \le \frac{C_2\|\bu-\bv\|_2}{\lambda},\quad\forall \bu,\bv\in\mathbb{B}_2^n.$$
To this end, in view of
$$\sfP_{\bu,\bv} = \mathbbm{E}\big[\mathbbm{1}(\sign(\ba_i^\top\bu-\tau_i)\ne\sign(\ba_i^\top\bv-\tau_i))\big],$$
we can first utilize the randomness of $\tau_i$ to obtain an upper bound
\begin{align}
    \sfP_{\bu,\bv}\le \mathbbm{E}\left[\frac{|\ba_i^\top(\bu-\bv)|}{2\lambda}\right]\label{tauupper}
\end{align}
and also a lower bound
     \begin{align}\nn
        \sfP_{\bu,\bv} &\ge \mathbbm{E}\big[\mathbbm{1}(\sign(\ba_i^\top\bu-\tau_i)\ne\sign(\ba_i^\top\bv-\tau_i))\mathbbm{1}(|\ba_i^\top\bu|\vee|\ba_i^\top\bv|\le\lambda)\big]\\ 
        & = \mathbbm{E}\left[\frac{|\ba_i^\top(\bu-\bv)|}{2\lambda}\cdot \mathbbm{1}\big(|\ba_i^\top\bu|\vee|\ba_i^\top\bv|\le\lambda\big)\right].\label{taulower}
    \end{align}

Now
 we introduce $\bn_0:=\frac{\bu-\bv}{\|\bu-\bv\|_2}\in\mathbb{S}^{n-1}$ and further deal with the randomness of $\ba_i$. By (\ref{tauupper}) we obtain Assumption \ref{assum_Puvup} with $c^{(9)}=\frac{1}{2}$: 
 $$\sfP_{\bu,\bv}\le \frac{\|\bu-\bv\|_2\mathbbm{E}[|\ba_i^\top\bn_0|]}{2\lambda}\le \frac{\|\bu-\bv\|_2}{2\lambda}.$$
 Conversely, if $\lambda$ is large enough, then we proceed from (\ref{taulower}) as 
\begin{align*}
    \sfP_{\bu,\bv} &\ge \frac{\|\bu-\bv\|_2}{2\lambda}\mathbbm{E}\Big[|\ba_i^\top\bn_0|-|\ba_i^\top\bn_0|\mathbbm{1}(|\ba_i^\top\bu|>\lambda)-|\ba_i^\top\bn_0|\mathbbm{1}(|\ba_i^\top\bv|>\lambda)\Big]\\
    &\ge \frac{\|\bu-\bv\|_2}{2\lambda}\Big[L_0-4\exp(-\bar{c}_0\lambda^2)\Big]\ge \frac{L_0\|\bu-\bv\|_2}{4\lambda},
\end{align*}
where the second inequality follows from  (\ref{eq:L1L2equ}) and (\ref{eq:d1bcs_sgtail}) that gives $$\mathbbm{E}\big[|\ba_i^\top\bn_0|\mathbbm{1}(|\ba_i^\top\bw|>\lambda)\big]\le \sqrt{\mathbbm{E}[|\ba_i^\top\bn_0|^2]\cdot\mathbbm{P}(|\ba_i^\top\bw|>\lambda)}\le 2\exp(-\bar{c}_0\lambda^2),\quad\forall \bw\in\mathbb{B}_2^n,$$ 
  and the last inequality holds when $\lambda$ is large enough. Thus, Assumption \ref{assum3} holds with absolute constant $c^{(2)}=\frac{L_0}{4}$.
\end{proof}
\subsection{The Proof of Lemma \ref{fact4} (Assumptions \ref{assum_sg1}--\ref{assum_sg3} for D1bCS)}\label{app:provefact3}
\begin{proof}
Because $(\ba_i^\top\bbeta_1,\ba_i^\top\bbeta_2,J_{i,\bw}^{\bp,\bq})$ are independent of $\tau_i$, we can always expect $\tau_i$ first and, as with (\ref{tauupper}) and (\ref{taulower}), we have the two-sided bound
\begin{align}
    \label{tau2sided11}
   \frac{|\ba_i^\top(\bp-\bq)|}{2\lambda}\cdot\mathbbm{1}\big(|\ba_i^\top\bp|\vee|\ba_i^\top\bq|\le\lambda\big)\le \mathbbm{E}_{\tau_i}\big[\mathbbm{1}(E^{(i)}_{\bp,\bq})\big]\le \frac{|\ba_i^\top(\bp-\bq)|}{2\lambda}.
\end{align}
This  enables us  to essentially treat $\mathbbm{1}(E^{(i)}_{\bp,\bq})$ as $\frac{|\ba_i^\top(\bp-\bq)|}{2\lambda}=\frac{\|\bp-\bq\|_2}{2\lambda}|\ba_i^\top\bbeta_1|$, as $\{|\ba_i^\top\bp|\vee|\ba_i^\top\bq|\le\lambda\}$ holds with overwhelming probability under sufficiently large $\lambda$.

    Recall that any two different points $\bp,\bq$ in $\mathbb{B}_2^n$   are parameterized by some orthonormal $(\bbeta_1,\bbeta_2)$, and we have derived a two-sided bound in (\ref{tau2sided}) that implies 
\begin{align}
    \label{tau2sided2}
    -\frac{|\ba_i^\top(\bp-\bq)|}{2\lambda}\cdot\big[\mathbbm{1}(|\ba_i^\top\bp|>\lambda)+\mathbbm{1}(|\ba_i^\top\bq|>\lambda)\big] \le \mathbbm{E}_{\tau_i}\big[\mathbbm{1}(E^{(i)}_{\bp,\bq})\big]- \frac{|\ba_i^\top(\bp-\bq)|}{2\lambda}\le 0.
\end{align}
We introduce the notation 
$$a_1:=\ba_i^\top\bbeta_1,~a_2:=\ba_i^\top\bbeta_2,~a_{\bw}:=[\ba_i-(\ba_i^\top\bbeta_1)\bbeta_1-(\ba_i^\top\bbeta_2)\bbeta_2]^\top\bw=\ba_i^\top(\bI_n-\bbeta_1\bbeta_1^\top-\bbeta_2\bbeta_2^\top)\bw$$
with $\bw\in\mathbb{S}^{n-1}$. Then by $\mathbbm{E}(\ba_i\ba_i^\top)=\bI_n$ it is easy to see $\mathbbm{E}(a_1^2)=1$ and $\mathbbm{E}(a_1a_2)=\mathbbm{E}(a_1a_{\bw})=0$. In this proof, we will first validate the moment bounds (\ref{eq:sg1_sg}), (\ref{eq:sg2_sg}) and (\ref{eq:sg3_sg}), and then show (\ref{eq:sg1_exp}), (\ref{eq:sg2_exp}) and (\ref{eq:sg3_exp}) via calculations.  

\subsubsection*{(i) Validating (\ref{eq:sg1_sg}), (\ref{eq:sg2_sg}) and (\ref{eq:sg3_sg})}
We shall give a unified treatment. Let $A=a_1$, $a_2$ or $a_\bw$, and note that  when $\bw\in \mathbb{S}^{n-1}$, for any choice of $A$ and any $p\ge 2$, we have $\mathbbm{E}|A|^p\le (\bar{c}_1\sqrt{p})^p$ as in (\ref{eq:d1bcs_moment}). By using the randomness of $\tau_i$ and the upper bound in (\ref{tau2sided2}), we obtain
\begin{align*}
    &\mathbbm{E}\Big(|A|^p\mathbbm{1}(E^{(i)}_{\bp,\bq})\Big)\\& = \mathbbm{E}\Big(|A|^p \mathbbm{E}_{\tau_i}\big[E^{(i)}_{\bp,\bq}\big]\Big)\\
    &\le \frac{\|\bp-\bq\|_2}{2\lambda}\mathbbm{E}\Big(|A|^p|\ba_i^\top\bbeta_1|\Big)\\&\le \frac{\|\bp-\bq\|_2}{2\lambda}\sqrt{\mathbbm{E}[|A|^{2p}]}\\
    &\le C_1\sfP_{\bp,\bq}(2\bar{c}_1\sqrt{p})^p \le C_2\sfP_{\bp,\bq}\cdot\frac{p!}{2}
\end{align*}
for some absolute constants $C_1$ and $C_2$, where we use Cauchy-Schwarz inequality and the verified $\sfP_{\bp,\bq}\asymp\frac{\|\bp-\bq\|_2}{\lambda}$. Therefore, (\ref{eq:sg1_sg}), (\ref{eq:sg2_sg}) and (\ref{eq:sg3_sg}) hold with absolute constants $c^{(5)}$, $c^{(7)}$ and $c^{(8)}$.

\subsubsection*{(ii) Validating (\ref{eq:sg1_exp}), (\ref{eq:sg2_exp}) and (\ref{eq:sg3_exp})}
    Again we give a unified treatment and let $A$ be one of the followings: $a_1$, $a_2$ or $a_{\bw}$ with $\bw\in\mathbb{S}^{n-1}$.  Then, we shall calculate $$\mathbbm{E}\big(\sign(a_1)A\cdot\mathbbm{1}(E^{(i)}_{\bp,\bq})\big)=\mathbbm{E}\big(\sign(a_1)A\cdot\mathbbm{E}_{\tau_i}\big[\mathbbm{1}(E^{(i)}_{\bp,\bq})\big]\big)$$ for all choices of $A$. Our two-sided estimate (\ref{tau2sided2}) yields 
  \begin{align*}
      &\left|\mathbbm{E}\big(\sign(a_1)A\cdot\mathbbm{E}_{\tau_i}\big[\mathbbm{1}(E^{(i)}_{\bp,\bq})\big]\big)-\mathbbm{E}\Big(\sign(a_1)A\cdot\frac{|\ba_i^\top(\bp-\bq)|}{2\lambda}\Big)\right|\\
      &\le \mathbbm{E}\left[|A|\cdot \frac{|\ba_i^\top(\bp-\bq)|}{2\lambda}\Big(\mathbbm{1}(|\ba_i^\top\bp|>\lambda)+\mathbbm{1}(|\ba_i^\top\bq|>\lambda)\Big)\right]\\
      &\le \frac{\|\bp-\bq\|_2}{\lambda}\sup_{\bu\in \mathbb{B}_2^n}\mathbbm{E}\left[|A|\cdot|\ba_i^\top\bbeta_1|\cdot \mathbbm{1}(|\ba_i^\top\bu|>\lambda)\right]\\
      &\le \frac{\|\bp-\bq\|_2}{\lambda} \Big[\mathbbm{E}|A|^4\Big]^{1/4}\Big[\mathbbm{E}|\ba_i^\top\bbeta_1|^4\Big]^{1/4}\sup_{\bu\in \mathbb{B}_2^n} \sqrt{\mathbbm{P}(|\ba_i^\top\bu|>\lambda)}\\
      &\le \frac{\|\bp-\bq\|_2}{\lambda} \cdot (2\bar{c}_1)^2\cdot \sqrt{2\exp(-2\bar{c}_0\lambda^2)}\\&
      \le 8(\bar{c}_1)^2\exp(-\bar{c}_0\lambda^2)\cdot \frac{\|\bp-\bq\|_2}{\lambda}, 
  \end{align*}
  where in the last line we use (\ref{eq:d1bcs_sg_abs}).  Therefore, to establish (\ref{eq:sg1_exp}), (\ref{eq:sg2_exp}) and (\ref{eq:sg3_exp}) with $\varepsilon^{(2)}=\varepsilon^{(3)}=\varepsilon^{(4)}=8(\bar{c}_1)^2\exp(-\bar{c}_0\lambda^2)$, all that remains is to evaluate $\mathbbm{E}(\sign(a_1)A\frac{|\ba_i^\top(\bp-\bq)|}{2\lambda})$:
\begin{align*}
    &\mathbbm{E}\Big(\sign(a_1)A\frac{|\ba_i^\top(\bp-\bq)|}{2\lambda}\Big) 
    \\&= \frac{\|\bp-\bq\|_2}{2\lambda}\mathbbm{E}\big(\sign(a_1)A|\ba_i^\top\bbeta_1|\big)
    \\&= \frac{\|\bp-\bq\|_2}{2\lambda}\mathbbm{E}(a_1A) = \begin{cases}
        \frac{\|\bp-\bq\|_2}{2\lambda},\quad &\text{if }~A=a_1\\
        ~~0~~,\quad &\text{if }~A=a_2,~a_{\bw}
    \end{cases}.
\end{align*}
We come to the desired claim. 
 \end{proof}
\subsection{The Proof of Lemma \ref{fact7} (Assumptions \ref{assum2}, \ref{assum3}, \ref{assum_Puvup} for DMbCS)}\label{app:provefact4}
\begin{proof}
For any $\bu\in\mathbb{B}_2^n$, we utilize the randomness of $\tau_i$ to obtain  
     \begin{align*}
        &\mathbbm{P}\left(\min_{|j|\le (L/2)-1}|\ba_i^\top \bu -\tau_i- j\delta |\le r\right) = \mathbbm{P}\left(-\tau_i\in \bigcup_{|j|\le (L/2)-1}\big[j\delta-r-\ba_i^\top\bu,j\delta+r-\ba_i^\top\bu\big]\right)\le \frac{2r}{\delta},
    \end{align*}
which establishes Assumption \ref{assum2} with $c^{(1)}=2$.

Note that $\sfP_{\bu,\bv}= \mathbbm{P}(E^{(i)}_{\bu,\bv})=\mathbbm{E}\big(\mathbbm{1}(E^{(i)}_{\bu,\bv})\big)= \mathbbm{E}_{\ba_i}\big[\mathbbm{E}_{\tau_i}\big(\mathbbm{1}(E^{(i)}_{\bu,\bv})\big)\big]$.

\paragraph{Expecting $\tau_i$:} By the definition of $\calQ_{\delta,L} $ in (\ref{eq:QdeltaL}), we have $\calQ_{\delta,L}(a)=\calQ_\delta(a)$ for $|a|<\frac{L\delta}{2}$, which together with $|\tau_i|\le \frac{\delta}{2}$ yields
\begin{subequations}
    \begin{gather}\label{86aa}
        \mathbbm{1}(E^{(i)}_{\bu,\bv}) \le \mathbbm{1}\big(\calQ_{\delta}(\ba_i^\top\bu-\tau_i)\ne\calQ_{\delta}(\ba_i^\top\bv-\tau_i)\big),\\\label{86bb}
        \mathbbm{1}(E^{(i)}_{\bu,\bv})\ge \mathbbm{1}\big(\calQ_{\delta}(\ba_i^\top\bu-\tau_i)\ne\calQ_{\delta}(\ba_i^\top\bv-\tau_i)\big)\mathbbm{1}\Big(|\ba_i^\top\bu|\vee|\ba_i^\top\bv|\le \frac{(L-1)\delta}{2}\Big)
    \end{gather}
\end{subequations}
Therefore, we can utilize the randomness of $\tau_i\sim\scrU([-\frac{\delta}{2},\frac{\delta}{2}])$ to obtain from (\ref{86aa})
\begin{align}\nn
    &\mathbbm{E}_{\tau_i}\big(E^{(i)}_{\bu,\bv}\big)\le \mathbbm{E}_{\tau_i}\Big(\mathbbm{1}\big(\calQ_{\delta}(\ba_i^\top\bu-\tau_i)\ne\calQ_{\delta}(\ba_i^\top\bv-\tau_i)\big)\Big)\\
    \nn&= \mathbbm{E}_{\tau_i}\left[\mathbbm{1}\big(|\ba_i^\top(\bu-\bv)|\ge\delta\big)+\mathbbm{1}\big(|\ba_i^\top(\bu-\bv)|<\delta\big)\mathbbm{1}\left(\tau_i\in \bigcup_{j\in \mathbb{Z}}\big(\ba_i^\top\bu\wedge\ba_i^\top\bv -j\delta,\ba_i^\top\bu\vee\ba_i^\top\bv -j\delta\big)\right)\right]\\
    \label{tau_uppermb}&=\mathbbm{1}\big(|\ba_i^\top(\bu-\bv)|\ge\delta\big)+\mathbbm{1}\big(|\ba_i^\top(\bu-\bv)|<\delta\big)\frac{|\ba_i^\top(\bu-\bv)|}{\delta}\le\frac{|\ba_i^\top(\bu-\bv)|}{\delta} 
\end{align}
and similarly obtain from (\ref{86bb})
\begin{align}
    &\mathbbm{E}_{\tau_i}\big(E^{(i)}_{\bu,\bv}\big)\ge \mathbbm{E}_{\tau_i}\Big(\mathbbm{1}\big(\calQ_{\delta}(\ba_i^\top\bu-\tau_i)\ne\calQ_{\delta}(\ba_i^\top\bv-\tau_i)\big)\Big)  \mathbbm{1}\Big(|\ba_i^\top\bu|\vee|\ba_i^\top\bv|\le \frac{(L-1)\delta}{2}\Big)\nn\\
    &=\left[ \mathbbm{1}\big(|\ba_i^\top(\bu-\bv)|\ge\delta\big)+\mathbbm{1}\big(|\ba_i^\top(\bu-\bv)|<\delta\big)\frac{|\ba_i^\top(\bu-\bv)|}{\delta} \right]\mathbbm{1}\Big(|\ba_i^\top\bu|\vee|\ba_i^\top\bv|\le \frac{(L-1)\delta}{2}\Big)\label{8888}\\
\label{taulowermb}    &\ge \frac{|\ba_i^\top(\bu-\bv)|}{\delta} \mathbbm{1}\Big(|\ba_i^\top\bu|\vee|\ba_i^\top\bv|\le \frac{(L-1)\delta}{2},~~|\ba_i^\top(\bu-\bv)|<\delta\Big)
\end{align}

\paragraph{Expecting $\ba_i$:} We further take expectation regarding $\tau_i$. By (\ref{tau_uppermb}) we have 
     \begin{align*}
    &\sfP_{\bu,\bv}\le \frac{\mathbbm{E}|\ba_i^\top(\bu-\bv)|}{\delta}\le \frac{\|\bu-\bv\|_2}{\delta}.
\end{align*}
  Hence, Assumption \ref{assum_Puvup} holds with $c^{(9)}=1$.

  By continuing from (\ref{8888}) and some simple algebra, with $\bn := \frac{\bu-\bv}{\|\bu-\bv\|_2}$
  we can lower bound $\sfP_{\bu,\bv}$ as 
\begin{subequations}\label{eq:lower_Puvmbmb}
    \begin{align}\nn
        &\sfP_{\bu,\bv}  
         \ge  \mathbbm{E}\left[\left[\mathbbm{1}\big(|\ba_i^\top(\bu-\bv)|\ge\delta\big)+\mathbbm{1}\big(|\ba_i^\top(\bu-\bv)|<\delta\big)\frac{|\ba_i^\top(\bu-\bv)|}{\delta}\right]\cdot\mathbbm{1}\Big(|\ba_i^\top\bu|\vee|\ba_i^\top\bv|\le \frac{(L-1)\delta}{2}\Big)\right] 
          \\\label{eq:term1_lower}
        & = \mathbbm{E}\left[\mathbbm{1}\Big(|\ba_i^\top \bn|\ge \frac{\delta}{\|\bu-\bv\|_2}\Big)+\mathbbm{1}\Big(|\ba_i^\top\bn|<\frac{\delta}{\|\bu-\bv\|_2}\Big)\frac{|\ba_i^\top\bn|\|\bu-\bv\|_2}{\delta}\right]\\\label{eq:term2_up}
        &- \mathbbm{E}\left[\left[\mathbbm{1}\Big(|\ba_i^\top\bn|\ge \frac{\delta}{\|\bu-\bv\|_2}\Big)+\mathbbm{1}\Big(|\ba_i^\top\bn|<\frac{\delta}{\|\bu-\bv\|_2}\Big)\frac{|\ba_i^\top\bn|\|\bu-\bv\|_2}{\delta}\right] \mathbbm{1}\Big(|\ba_i^\top\bu|\vee|\ba_i^\top\bv|>\frac{(L-1)\delta}{2}\Big)\right].
    \end{align}
\end{subequations}
It is not hard to see the following lower bound on (\ref{eq:term1_lower}): 
\begin{align*}
    \text{The term in (\ref{eq:term1_lower})}\ge \min\Big\{\frac{\|\bu-\bv\|_2}{\delta},1\Big\}\cdot \mathbbm{E}\Big[\mathbbm{1}(|\ba_i^\top\bn|\ge 1)+ |\ba_i^\top\bn|\mathbbm{1}(|\ba_i^\top\bn|<1)\Big].
\end{align*}
By (\ref{eq:prob_ge12}), for some absolute constant $L_0$ we have $$\mathbbm{E}\big[\mathbbm{1}(|\ba_i^\top\bn|\ge 1)+|\ba_i^\top\bn|\mathbbm{1}(|\ba_i^\top\bn|<1)\big]\ge \frac{1}{2}\mathbbm{P}\Big(|\ba_i^\top\bn|\ge \frac{1}{2}\Big)\ge L_0,$$ 
hence we obtain the lower bound
\begin{align*}
    \text{The term in (\ref{eq:term1_lower})}\ge L_0\min\Big\{\frac{\|\bu-\bv\|_2}{\delta},1\Big\}. 
\end{align*}

 Next, we invoke Cauchy–Schwarz  inequality to upper bound the  term in (\ref{eq:term2_up}): 
     \begin{align*}
        &\text{The term in (\ref{eq:term2_up})}\\
        &\le \sqrt{\mathbbm{E}\left[\mathbbm{1}\Big(|\ba_i^\top\bn|\ge\frac{\delta}{\|\bu-\bv\|_2}\Big) + \mathbbm{1}\Big(|\ba_i^\top\bn|<\frac{\delta}{\|\bu-\bv\|_2}\Big)\frac{|\ba_i^\top\bn|^2\|\bu-\bv\|^2_2}{\delta^2}\right]\cdot\mathbbm{P}\Big(|\ba_i^\top\bu|\vee|\ba_i^\top\bv|>\frac{(L-1)\delta}{2}\Big)}  \\
        & \le 2\min\Big\{\frac{\|\bu-\bv\|_2}{\delta},1\Big\} \exp\Big(-\frac{\bar{c}_0L^2\delta^2}{16}\Big),
    \end{align*}
 where in the last inequality we use (\ref{eq:d1bcs_sgtail}) and $\frac{L-1}{2}\ge \frac{L}{4}$; they imply $$\sqrt{\mathbbm{P}\Big(|\ba_i^\top\bu|\vee|\ba_i^\top\bv|>\frac{(L-1)\delta}{2}\Big)}\le  2\exp\Big(-\frac{\bar{c}_0L^2\delta^2}{16}\Big).$$ Therefore, when $L\delta$ is sufficiently large such that $2\exp(-\frac{\bar{c}_0L^2\delta^2}{16})<\frac{L_0}{2}$, we arrive at $$\sfP_{\bu,\bv}\ge \frac{L_0}{2}\min\Big\{\frac{\|\bu-\bv\|_2}{\delta},1\Big\},$$
thus establishing Assumption \ref{assum3} with $c^{(2)}=\frac{L_0}{2}$. The proof is complete.    
\end{proof}
\subsection{The Proof of Lemma \ref{fact8} (Assumption \ref{assump4} for DMbCS)}\label{app:provefact5}
 
Since we work with the clipped gradient $\hat{\bh}(\bp,\bq)$, now we shall validate Assumption \ref{assump4} to bound the deviation.   By (\ref{eq:hpq}) and (\ref{eq:hathpq}) we first formulate 
\begin{align*}
    \bh(\bp,\bq)-\hat{\bh}(\bp,\bq) = \frac{1}{m}\sum_{i\in\bR_{\bp,\bq}}\sign(\ba_i^\top(\bp-\bq))\big(|\calQ_{\delta,L}(\ba_i^\top\bp-\tau_i)-\calQ_{\delta,L}(\ba_i^\top\bq-\tau_i)|-\delta\big)\ba_i. 
\end{align*}
  Using  triangle inequality   we can formulate
    \begin{align}
        &\eta\cdot\| \bh(\bp,\bq)-\hat{\bh}(\bp,\bq)\|_{\calK_{(r)}^\circ}\nn\\
        & = \eta\cdot \sup_{\bw\in\calK_{(r)}}\left|\frac{1}{m}\sum_{i\in\bR_{\bp,\bq}}\sign(\ba_i^\top(\bp-\bq))\big(|\calQ_{\delta,L}(\ba_i^\top\bp-\tau_i)-\calQ_{\delta,L}(\ba_i^\top\bq-\tau_i)|-\delta\big)\ba_i^\top\bw\right|\nn\\
        &\le \eta\cdot \sup_{\bw\in\calK_{(r)}}\frac{1}{m}\sum_{i\in\bR_{\bp,\bq}}\big||\calQ_{\delta,L}(\ba_i^\top\bp-\tau_i)-\calQ_{\delta,L}(\ba_i^\top\bq-\tau_i)|-\delta\big||\ba_i^\top\bw|\nn\\
        &\le \eta\cdot \sup_{\bw\in\calK_{(r)}}\frac{1}{m}\sum_{i\in\bR_{\bp,\bq}}|\ba_i^\top\bw||\ba_i^\top(\bp-\bq)|\mathbbm{1}\big(|\ba_i^\top(\bp-\bq)|\ge\delta\big) ,\label{eq:formu_assum4}
    \end{align} 
where the last inequality is due to the following Lemma \ref{lem:mblem}.

\begin{lem}\label{lem:mblem}
    For any $a,b\in \mathbb{R}$ we have 
    \begin{align*}
        \big||\calQ_{\delta,L}(a)-\calQ_{\delta,L}(b)|-\delta\big|\cdot\mathbbm{1}(\calQ_{\delta,L}(a)\ne\calQ_{\delta,L}(b))\le |a-b|\cdot \mathbbm{1}(|a-b|\ge\delta). 
    \end{align*}
\end{lem} 

\begin{proof}[The Proof of Lemma \ref{lem:mblem}]
    By the definitions of $\calQ_{\delta} $ and $\calQ_{\delta,L} $ in (\ref{eq:Qdelta}) and (\ref{eq:QdeltaL}), we have $|\calQ_{\delta,L}(a)-\calQ_{\delta,L}(b)|\ge \delta $ when $\calQ_{\delta,L}(a)\ne\calQ_{\delta,L}(b)$, and we always have 
    $$|\calQ_{\delta,L}(a)-\calQ_{\delta,L}(b)|\le|\calQ_{\delta}(a)-\calQ_{\delta}(b)|\quad\text{and}\quad \mathbbm{1}(\calQ_{\delta,L}(a)\ne\calQ_{\delta,L}(b))\le \mathbbm{1}(\calQ_\delta(a)\ne\calQ_\delta(b)).$$
 Taken collectively, we arrive at 
        \begin{align*}
            \big||\calQ_{\delta,L}(a)-\calQ_{\delta,L}(b)|-\delta\big|\cdot \mathbbm{1}(\calQ_{\delta,L}(a)\ne\calQ_{\delta,L}(b))\le \big(|\calQ_{\delta}(a)-\calQ_{\delta}(b)|-\delta\big)\cdot\mathbbm{1}(\calQ_\delta(a)\ne\calQ_\delta(b)).
        \end{align*}
        Further,  by $|\calQ_\delta(a)-a|\le \frac{\delta}{2}$ ($\forall a\in\mathbb{R}$) we always have 
        \begin{align}\label{eq:Qab}
            |\calQ_\delta(a)-\calQ_\delta(b)|\le |a-b|+\delta,\quad \forall~a,b\in\mathbb{R}.
        \end{align}  
        This leads to $|\calQ_\delta(a)-\calQ_\delta(b)|-\delta\le |a-b|$. 
        Also note that when $\calQ_\delta(a)\ne\calQ_\delta(b)$, we have $|\calQ_\delta(a)-\calQ_\delta(b)|\in\{\delta,2\delta,3\delta,\cdots\}$ and hence $|\calQ_\delta(a)-\calQ_\delta(b)|-\delta\ge 0$; and note that if $|\calQ_\delta(a)-\calQ_\delta(b)|-\delta>0$, we must have $|\calQ_\delta(a)-\calQ_\delta(b)|\ge 2\delta$, which combined with (\ref{eq:Qab}) implies $|a-b|\ge \delta$. Combining all these observations leads to 
        \begin{align*}
           \big (|\calQ_\delta(a)-\calQ_\delta(b)|-\delta\big)\cdot\mathbbm{1}(\calQ_\delta(a)\ne\calQ_\delta(b)) \le |a-b|\cdot\mathbbm{1}(|a-b|\ge\delta),
        \end{align*}
       as desired.
 \end{proof}
 The major ideas for showing  (\ref{eq:boundclip1}) is to control the number of nonzero summands in the summation (\ref{eq:formu_assum4}), by making use of the  factor $\mathbbm{1}(|\ba_i^\top(\bp-\bq)|\ge\delta)$.  
   Specifically, by $\|\bp-\bq\|_2\le 2\mu_4\le 2c_1\delta$ with small enough $c_1$ we manage to show  that $|\{i\in [m]:|\ba_i^\top(\bp-\bq)|\ge\delta\}|$ is uniformly small over $(\bp,\bq)\in \calN_{r,\mu_4}^{(2)}$. In contrast, it turns out that  a cruder argument which proceeds with all $m$ summands is sufficient for (\ref{eq:boundclip2}).  We have the following statement.   It demonstrates that (\ref{eq:boundclip1}) holds with $\varepsilon^{(1)}$ exponentially decaying with $c_1^{-1}$ if $\mu_4\le c_1\delta$, thus ensures that the deviation has minimal impact.

\begin{proof}[The Proof of Lemma \ref{fact8}] 
Let us first show (\ref{eq:boundclip2}). 

\paragraph{Show (\ref{eq:boundclip2}):} Note that $\|\bp-\bq\|_2\le r$ implies $\bp-\bq\in \calK_{(r)}$, we  continue from  (\ref{eq:formu_assum4}) to obtain  
    \begin{align*}
       &\eta\cdot\|\bh(\bp,\bq)-\hat{\bh}(\bp,\bq)\|_{\calK^\circ_{(r)}} \le \eta \sup_{\bw\in\calK_{(r)}}\sup_{\bv\in\calK_{(r)}}\frac{1}{m}\sum_{i=1}^m|\ba_i^\top\bw||\ba_i^\top\bv|
        \\&\le \eta \sup_{\bw\in\calK_{(r)}}\frac{1}{m}\sum_{i=1}^m|\ba_i^\top\bw|^2 \le C_6\eta\Big(\frac{\omega(\calK_{(r)})}{\sqrt{m}}+r\Big)^2 \le 4C_6\eta r^2
    \end{align*}
     with probability at least $1-\exp(-c_7m)$, where $C_6,c_7$ are absolute constants; this is achieved by using   Lemma \ref{lem:max_ell_sum} with $\ell =m$ and the sample complexity $m\ge \frac{\delta\omega^2(\calK_{(r)})}{r^3}\ge \frac{\omega^2(\calK_{(r)})}{r}$. Therefore, (\ref{eq:boundclip2}) holds for some absolute constant $c^{(4)}$.

\paragraph{Show (\ref{eq:boundclip1}):}
Note that we can restrict our attention on $(\bp,\bq)\in\calN^{(2)}_{r,\mu_4}$ obeying $\|\bp-\bq\|_2\ge r$, since we can simply use the bound in (\ref{eq:boundclip2}) if $\|\bp-\bq\|_2\le r$. 
The core idea  is to control the number of effective contributors to the last line in (\ref{eq:formu_assum4}) and then invoke Lemma \ref{lem:max_ell_sum}.

We    use   $\mathbbm{1}(|\ba_i^\top(\bp-\bq)|\ge\delta)$ to control the number of non-zero contributors. For fixed $(\bp,\bq)\in\calN_{r,\mu_4}^{(2)}$ and a specific $i\in[m]$, by $\|\bp-\bq\|_2\le2\mu_4\le 2c_1\delta$ and (\ref{eq:d1bcs_sgtail}) where $\bar{c}_0$ is an absolute constant, we have 
\begin{align*}
    \mathbbm{P}\big(|\ba_i^\top(\bp-\bq)|\ge\delta\big) \le \mathbbm{P}\Big(\Big|\frac{\ba_i^\top(\bp-\bq)}{\|\bp-\bq\|_2}\Big|\ge \frac{1}{2c_1}\Big)\le 2\exp \Big(-\frac{\bar{c}_0}{2c_1^2}\Big):=p_0. 
\end{align*}
We assume $p_0\in (0,1/2)$, which can be ensured by setting $c_1$ sufficiently small.
By Chernoff bound and letting $c_2=\frac{\bar{c}_0}{2}$, we obtain 
     \begin{align*}
        &\mathbbm{P}\Big(\big|\{i\in[m]:|\ba_i^\top(\bp-\bq)|\ge\delta\}\big|\ge \frac{3mp_0}{2}\Big)\\&\le \mathbbm{P}\Big(\text{Bin}(m,p_0)\ge\frac{3mp_0}{2}\Big) \le \exp\Big(-\frac{mp_0}{12}\Big)=\exp \Big(-\frac{m\exp(-c_2/c_1^2)}{6}\Big).
    \end{align*}
    Because of the sample complexity $m\gtrsim\exp(\frac{c_2}{c_1^2}) \scrH(\calX,r)$ and $|\calN_{r,\mu_4}^{(2)}|\le |\scrN(\calX,r)|^2$, we can take a union bound over $(\bp,\bq)\in\calN_{r,\mu_4}^{(2)}$  and obtain that the event 
  \begin{align}\label{eq:uniEve}
      \big|\{i\in[m]:|\ba_i^\top(\bp-\bq)|\ge\delta\}\big|<\frac{3mp_0}{2},\quad\forall (\bp,\bq)\in\calN_{r,\mu_4}^{(2)}
  \end{align}
  holds with the probability at least $1-\exp(-c_8\scrH(\calX,r))$.

  Now we observe that any $(\bp,\bq)\in \calN_{r,\mu_4}^{(2)}$ obeying $\|\bp-\bq\|_2 \ge r$ satisfies $\frac{\bp-\bq}{\|\bp-\bq\|_2}\in \frac{\calK-\calK}{r}\cap \mathbb{B}_2^n=\frac{1}{r}\calK_{(r)}$.
  On the event (\ref{eq:uniEve}),     we continue from   (\ref{eq:formu_assum4})  and take supremum over $\frac{r(\bp-\bq)}{\|\bp-\bq\|_2}\in\calK_{(r)}$   to reach 
       \begin{align}\nn
          &\eta\cdot \|\bh(\bp,\bq) - \hat{\bh}(\bp,\bq)\|_{\calK_{(r)}^\circ}
          \\&\le \frac{\eta\|\bp-\bq\|_2}{r}\sup_{\bw \in\calK_{(r)}}\sup_{\bv\in \calK_{(r)}}\max_{\substack{I\subset [m]\\|I|\le 3mp_0/2}}\frac{1}{m}\sum_{i\in I}|\ba_i^\top\bw||\ba_i^\top\bv| \\\nn 
          &\le \frac{3\eta p_0\|\bp-\bq\|_2}{2r} \cdot  \sup_{\bw\in \calK_{( r)}}\max_{\substack{I\subset [m]\\|I|\le 3mp_0/2}} \frac{1}{3mp_0/2}\sum_{i\in I}|\ba_i^\top\bw|^2 \\
          &\le C_9\eta r\|\bp-\bq\|_2 \cdot \Big(\frac{\omega^2(\calK_{(r)})}{mr^2}+ p_0\log(p_0^{-1}) \Big),\label{eq:clip1uni}
      \end{align}
 where   the last step holds with probability at least $1-\exp(-c_{10}\scrH(\calX,r))$ for some absolute constants $C_9,c_{10}$ due to Lemma \ref{lem:max_ell_sum} and $mp_0= 2m\exp(-\frac{c_2}{c_1^2})\gtrsim \scrH(\calX,r)$. By $\|\bp-\bq\|_2\le \mu_4\le 2c_1\delta$ and the sample complexity $m\ge \frac{\delta\omega^2(\calK_{(r)})}{r^3}$, we have 
 $$\eta r\|\bp-\bq\|_2\cdot \frac{\omega^2(\calK_{(r)})}{mr^2}\lesssim \eta r\delta \cdot \frac{r}{\delta} = \eta r^2,$$ 
 which can be accommodated by $\frac{c^{(3)} \eta\Delta r}{\Delta\vee\Lambda}=c^{(3)}\eta r$ in (\ref{eq:boundclip1}) for some absolute constant $c^{(3)}$. When $c_1$ is small enough, then $p_0 = 2\exp(-\frac{\bar{c}_0}{2c_1^2})=2\exp(-\frac{c_2}{c_1^2})$ is also small enough, and we have $p_0\log(p_0^{-1}) \le \frac{C_5}{2}\sqrt{p_0}\le C_5\exp(-\frac{c_2}{2c_1^2})$ for some absolute constant $C_5$. Therefore, we have $$\eta r\|\bp-\bq\|_2 \cdot p_0\log(p_0^{-1})\le C_5\exp\Big(-\frac{c_2}{2c_1^2}\Big)\eta r\|\bp-\bq\|_2$$
 which can be accounted by $\frac{\varepsilon^{(1)}\phi\eta\Delta\|\bp-\bq\|_2}{\Delta\vee\Lambda} = \varepsilon^{(1)}\eta r\|\bp-\bq\|_2$ with $\varepsilon^{(1)} =C_5\exp(-\frac{c_2}{2c_1^2}) $, as claimed.  
 \end{proof} 
\subsection{The Proof of Lemma \ref{fact9} (Assumptions \ref{assum_sg1}--\ref{assum_sg3} for DMbCS)}\label{app:provefact6}
\begin{proof}
 The proof is similar to the proof of Lemma \ref{fact4} but involves more involved calculations. The main intuition is that the randomness of $\tau_i$ gives
\begin{align}
    \label{tau2sided}
  \frac{|\ba_i^\top(\bp-\bq)|}{\delta} \mathbbm{1}\Big(|\ba_i^\top(\bp-\bq)|<\delta,~|\ba_i^\top\bp|\vee |\ba_i^\top\bq|\le \frac{(L-1)\delta}{2}\Big)  \le \mathbbm{E}_{\tau_i}\big(\mathbbm{1}(E^{(i)}_{\bp,\bq})\big) \le \frac{|\ba_i^\top(\bp-\bq)|}{\delta},
\end{align}
which we obtain along the proof of Lemma \ref{fact7} in Appendix \ref{app:provefact4}. 
This allows us to essentially treat $\mathbbm{E}_{\tau_i}(\mathbbm{1}(E^{(i)}_{\bp,\bq}))$ as $\frac{|\ba_i^\top(\bp-\bq)|}{\delta}=\frac{\|\bp-\bq\|_2|\ba_i^\top\bbeta_1|}{\delta}$ for $(\bp,\bq)$ obeying $\|\bp-\bq\|_2\le 2\mu_4 \ll \delta$ and  $L\delta$ being large enough.

We consider any $\bp,\bq \in \mathbb{B}_2^n$ satisfying $\|\bp-\bq\|_2 \le 2\mu_4\le 2c_1\delta$ for some small enough $c_1$. By Assumptions \ref{assum3} and \ref{assum_Puvup} already validated in Lemma \ref{fact7}, we have $\sfP_{\bp,\bq}\asymp \frac{\|\bp-\bq\|_2}{\delta}$.

A useful technique is to first deal with the randomness of $\tau_i$. By (\ref{tau_uppermb}) and (\ref{taulowermb}) we have 
the two sided bound 
\begin{align}
    \label{2sidedmb}
    \frac{|\ba_i^\top(\bp-\bq)|}{\delta}\mathbbm{1}\Big(|\ba_i^\top\bp|\vee|\ba_i^\top\bq|\le \frac{(L-1)\delta}{2},~|\ba_i^\top(\bp-\bq)|<\delta\Big) \le \mathbbm{E}_{\tau_i}(E^{(i)}_{\bp,\bq}) \le \frac{|\ba_i^\top(\bp-\bq)|}{\delta}.
\end{align}
Recall that $\bbeta_1=\frac{\bp-\bq}{\|\bp-\bq\|_2}$ and $\bbeta_2$ are orthonormal. It is convenient to introduce the notation $\ba_i^\top\bbeta_1 := a_1$, $\ba_i^\top\bbeta_2:=a_2$ and  $[\ba_i-(\ba_i^\top\bbeta_1)\bbeta_1-(\ba_i^\top\bbeta_2)\bbeta_2]:=a_{\bw}$ for $\bw\in \mathbb{S}^{n-1}$. Moreover, we further use the generic notation $A$ which can be $a_1$, $a_2$ or $a_{\bw}$.

\subsubsection*{(i) Validating (\ref{eq:sg1_sg}), (\ref{eq:sg2_sg}) and (\ref{eq:sg3_sg})}
For any $p\ge 2$, by the upper bound $\mathbbm{E}_{\tau_i}(E^{(i)}_{\bp,\bq})\le \frac{|\ba_i^\top(\bp-\bq)|}{\delta}=\frac{\|\bp-\bq\|_2|\ba_i^\top\bbeta_1|}{\delta}$ from (\ref{2sidedmb}), we have 
\begin{align*}
    &\mathbbm{E}\big(|A|^p \mathbbm{1}(E^{(i)}_{\bp,\bq})\big) \le \mathbbm{E}\Big(|A|^p \frac{\|\bp-\bq\|_2|\ba_i^\top\bbeta_1|}{\delta}\Big) \\&\le \frac{\|\bp-\bq\|_2}{\delta}\sqrt{\mathbbm{E}[|A|^{2p}]\mathbbm{E}|\ba_i^\top\bbeta_1|^2}\\&\le C_1\sfP_{\bp,\bq} (C_2\sqrt{p})^{p} \le C_3\sfP_{\bp,\bq} \frac{p!}{2}
\end{align*}
for some absolute constants $C_1,C_2,C_3$. This establishes (\ref{eq:sg1_sg}), (\ref{eq:sg2_sg}) and (\ref{eq:sg3_sg}) with some absolute constants $c^{(5)}$, $c^{(7)}$ and $c^{(8)}$. 

\subsubsection*{(ii) Validating (\ref{eq:sg1_exp}), (\ref{eq:sg2_exp}) and (\ref{eq:sg3_exp})}

This amount to evaluating $\mathbbm{E}(\sign(a_1)A\cdot\mathbbm{1}(E^{(i)}_{\bp,\bq}))$ for $A=a_1,~a_2$ or $a_{\bw}$. By (\ref{2sidedmb}), Cauchy–Schwarz inequality, the moment bound (\ref{eq:d1bcs_moment}) and the tail bound (\ref{eq:d1bcs_sgtail}), for all three choices of $A$ we have 
\begin{align*}
    &\left|\mathbbm{E}\big(\sign(a_1)A\cdot\mathbbm{1}(E^{(i)}_{\bp,\bq})\big) - \mathbbm{E}\Big(\sign(a_1)A\frac{|\ba_i^\top(\bp-\bq)|}{\delta}\Big)\right|\\
    &\le \mathbbm{E}\left[\frac{|A||\ba_i^\top(\bp-\bq)|}{\delta}\mathbbm{1}\Big(|\ba_i^\top\bp|\vee|\ba_i^\top\bq|>\frac{(L-1)\delta}{2}\Big)\right]+\mathbbm{E}\left[\frac{|A||\ba_i^\top(\bp-\bq)|}{\delta}\mathbbm{1}\Big(|\ba_i^\top(\bp-\bq)|\ge\delta\Big)\right] \\
    & = \frac{\|\bp-\bq\|_2}{\delta}\left( \mathbbm{E}\left[|A||\ba_i^\top\bbeta_1|\mathbbm{1}\Big(|\ba_i^\top\bp|\vee|\ba_i^\top\bq|>\frac{(L-1)\delta}{2}\Big)\right]+\mathbbm{E}\left[|A||\ba_i^\top\bbeta_1|\mathbbm{1}\Big(|\ba_i^\top(\bp-\bq)|\ge\delta\Big)\right] \right)\\
    &\le \frac{\|\bp-\bq\|_2}{\delta}\sqrt{\mathbbm{E}\big[|A|^2|\ba_i^\top\bbeta_1|^2\big]}\left(\sqrt{\mathbbm{P}\Big(|\ba_i^\top\bp|\vee|\ba_i^\top\bq|>\frac{L\delta}{4}\Big)}+\sqrt{\mathbbm{P}\Big(|\ba_i^\top\bbeta_1|\ge\frac{1}{2c_1}\Big)}\right)\\
    &\le  \frac{\|\bp-\bq\|_2}{\delta}\cdot 8(\bar{c}_1)^2\Big[\exp \Big(-\frac{\bar{c}_0L^2\delta^2}{16}\Big)+ \exp\Big(-\frac{\bar{c}_0}{4c_1^2}\Big)\Big] = \frac{\varepsilon'\|\bp-\bq\|_2}{\delta}.
\end{align*}
All that remains is to compute $\mathbbm{E}(\sign(a_1)A|\ba_i^\top(\bp-\bq)|/\delta)$:
\begin{align*}
    \mathbbm{E}\Big(\sign(a_1)A\frac{|\ba_i^\top(\bp-\bq)|}{\delta}\Big) = \frac{\|\bp-\bq\|_2}{\delta} \mathbbm{E}\big(\sign(a_1)A |a_1|\big) = \begin{cases}
        \frac{\|\bp-\bq\|_2}{\delta},\quad &\text{if}~~A= a_1\\
        ~~~~0~,&\text{if}~~A= a_1\text{ or }a_{\bw}
    \end{cases}.
\end{align*}
The claim follows.     
\end{proof}

\subsection{The Proof of Lemma \ref{factqpe} (Bounding $\hat{T}$ for DMbCS)}\label{app:factgqpe}
We start with an outline of the proof. We shall first decompose $\hat{T}$ into 
\begin{align}
    \hat{T}&\le \underbrace{\sup_{\bu\in \overline{\calX}}\sup_{\bw\in\calK_{(r_1)}}\Big|\frac{1}{mr_1}\big\langle\calQ_{\delta}(\bA\bu-\btau) -\bA\bu,\bA\bw\big\rangle\Big|}_{:=\hat{T}^{(1)}}\nn\\&+  \underbrace{\sup_{\bu\in \overline{\calX}}\sup_{\bw\in\calK_{(r_1)}}\Big|\frac{1}{mr_1}\big\langle\calQ_{\delta}(\bA\bu-\btau)-\calQ_{\delta,L}(\bA\bu-\btau),\bA\bw\big\rangle\Big|}_{:=\hat{T}^{(2)}}\label{T12hat}
\end{align}
where $\hat{T}^{(1)}$ and $\hat{T}^{(2)}$  capture the distortion arising from the dithered uniform quantization and the impact of saturation, respectively.

To bound $\hat{T}^{(1)}$, we use the global quantized product embedding result for $\calQ_\delta$   from \cite{chen2024uniform},
which we rephrase in Lemma \ref{lem:qpe1}. Then, $\hat{T}^{(2)}$ is bounded in Lemma \ref{lemsatu}. Our idea to control $\hat{T}^{(2)}$ is similar to that for bounding the deviation from gradient clipping (Lemma \ref{fact8})---first we establish uniform bound on the number of measurements affected by  saturation, then we utilize Lemma \ref{lem:max_ell_sum}.

We first restate \cite[Corollary D.3]{chen2024uniform} that will be used to bound $\hat{T}^{(1)}$ in (\ref{T12hat}).  

\begin{lem}[Adapted from Corollary D.3 in \cite{chen2024uniform}] 
\label{lem:qpe1}
In our DMbCS setting, there exist absolute constants $c_0,C_1,c_2,C_3$ such that, given $\rho\in (0,c_0\delta)$ for small enough $c_0$ and two sets $\calA,\calC\subset \mathbb{B}_2^n$, if 
\begin{align*}
    m\ge C_1 \left[\frac{\delta^2(\scrH(\calA,\rho)+\omega^2(\calC))}{\rho^2\log^2(\frac{\delta}{\rho})}+\frac{\delta\omega^2(\calA_{(\rho)})}{\rho^3 \log^{3/2}(\frac{\delta}{\rho})}\right],
\end{align*}
then with probability at least $1-\exp(-c_2\scrH(\calA,\rho))$ we have 
$$\sup_{\ba\in\calA}\sup_{\bc\in\calC}\Big|\frac{1}{m}\big\langle \calQ_\delta(\bA\ba-\btau)-\bA\ba,\bA\bc\big\rangle\Big|\le C_3\rho \log\Big(\frac{\delta}{\rho}\Big).$$
\end{lem}

Then we  provide   Lemma \ref{lemsatu} that controls $\hat{T}^{(2)}$. We note that the fairly mild $\log^{1/2}(\frac{m}{\omega^2(\calX)})\lesssim L\delta\lesssim \omega(\calX)$ stems from certain technical treatment in our proof of Lemma \ref{lemsatu}. 
 
 \begin{lem} 
     \label{lemsatu}
      In our setting of DMbCS, there exist some absolute constants $C_0,C_1,C_2,c_3,C_4$ such that if $m\ge C_0\omega^2(\calX)$ and $C_1\log^{1/2}(\frac{m}{\omega^2(\calX)})\le L\delta \le C_2 \omega(\calX)$ hold, then with probability at least $1-\exp(-\frac{c_3\omega^2(\calX)}{(L\delta)^2})$ we have 
     \begin{align}\label{100}
       \sup_{\bu\in\overline{\calX}}\sup_{\bw\in\calK_{(r_1)}}\Big|\frac{1}{m}\langle \calQ_{\delta}(\bA\bu-\btau)-  \calQ_{\delta,L}(\bA\bu-\btau), \bA\bw\rangle\Big|\le \frac{C_4\omega(\calX)\omega(\calK_{(r_1)})}{m}. 
     \end{align}
 \end{lem}
 
\begin{proof} 
    By (\ref{eq:Qdelta}), (\ref{eq:QdeltaL}) and $|\calQ_\delta(a)|\le|a|+\frac{\delta}{2}$, for any $a\in\mathbb{R}$ we have 
    \begin{align*}
        |\calQ_\delta(a)-\calQ_{\delta,L}(a)| = \mathbbm{1}\Big(|a|\ge \frac{L\delta}{2}\Big)\Big(|\calQ_\delta(a)|-\frac{(L-1)\delta}{2}\Big)\le |a| \mathbbm{1}\Big(|a|\ge\frac{L\delta}{2}\Big).
    \end{align*}
 Combining with $|\tau_i|\le\frac{\delta}{2}$ and $\frac{(L-1)\delta}{2}\ge \frac{L\delta}{4}$, it follows that 
      \begin{align}\nn
      &\sup_{\bu\in\overline{\calX}}\sup_{\bw\in\calK_{(r_1)}}\left|\frac{1}{m}\sum_{i=1}^m \big[\calQ_\delta(\ba_i^\top\bu-\tau_i)-\calQ_{\delta,L}(\ba_i^\top\bu-\tau_i)\big]\cdot\ba_i^\top\bw\right|
         \\\nn
         &\le \sup_{\bu\in\overline{\calX}}\sup_{\bw\in\calK_{(r_1)}}\frac{1}{m}\sum_{i=1}^m \big|\calQ_\delta(\ba_i^\top\bu-\tau_i)-\calQ_{\delta,L}(\ba_i^\top\bu-\tau_i)\big||\ba_i^\top\bw|\\\nn
         &\le \sup_{\bu\in\overline{\calX}}\sup_{\bw\in\calK_{(r_1)}}\frac{1}{m}\sum_{i=1}^m |\ba_i^\top\bu-\tau_i||\ba_i^\top\bw|  \cdot\mathbbm{1}\Big(|\ba_i^\top\bu-\tau_i|\ge \frac{L\delta}{2}\Big) \\
         &\le\sup_{\bu\in\overline{\calX}}\sup_{\bw\in\calK_{(r_1)}}\frac{2}{m}\sum_{i=1}^m  |\ba_i^\top\bu||\ba_i^\top\bw|\cdot\mathbbm{1}\Big(|\ba_i^\top\bu|\ge\frac{L\delta}{4}\Big).\label{257d}
     \end{align}

  Next, we bound the number of nonzero contributors to (\ref{257d}) by Lemma \ref{lem:max_ell_sum}. Due to the factor $\mathbbm{1}\big(|\ba_i^\top\bu|\ge\frac{L\delta}{4}\big)$, we wish to ensure 
 \begin{align}\label{258}
     \Big|\Big\{i\in[m]:|\ba_i^\top\bu|\ge \frac{L\delta}{4}\Big\}\Big|\le r_0m,\quad\forall \bu\in\overline{\calX}
 \end{align}
 for some $r_0\in (0,1)$ such that $r_0m\in[m]$, which is to be chosen later. To this end, we observe that it is sufficient to ensure 
 \begin{align}
     \sup_{\bu\in\overline{\calX}}\max_{\substack{I\subset [m]\\|I|\le r_0m}}\left(\frac{1}{r_0m}\sum_{i\in I}|\ba_i^\top\bu|^2\right)^{1/2}< \frac{L\delta}{8}.
 \end{align}
 Further invoking Lemma \ref{lem:max_ell_sum} shows that, for some absolute constants $c_1,c_2$,  with probability at least $1-2\exp(-c_1r_0m)$, we only need to ensure 
 \begin{align}\label{260}
     \frac{\omega^2(\calX)}{r_0m}+ \log\Big(\frac{e}{r_0}\Big) \le c_2(L\delta)^2.
 \end{align}
 Hence, under the assumptions $m\gtrsim \omega^2(\calX)$ and $L\delta \gtrsim \log^{1/2}(\frac{m}{\omega^2(\calX)})$, we  set $r_0 := \frac{C_3\omega^2(\calX)}{m(L\delta)^2}$ with sufficiently large $C_3$ to justify  (\ref{260}), which further leads to (\ref{258}) with the promised probability.

 On the event (\ref{258}), we again apply  Lemma \ref{lem:max_ell_sum} to obtain  
      \begin{align*}
         &{\rm (\ref{257d})}\le  \sup_{\bu\in\overline{\calX}}\sup_{\bw\in\calK_{(r_1)}} \max_{\substack{I\subset [m]\\|I|\le r_0m}} \frac{2}{m}\sum_{i\in I}|\ba_i^\top\bu||\ba_i^\top\bw|\\
         &\le 2 r_0 \left({\sup_{\bu\in\overline{\calX}}\max_{\substack{I\subset[m]\\|I|\le r_0m}}(r_0m)^{-1}\sum_{i\in I}|\ba_i^\top\bu|^2}\right)^{1/2}\left(\sup_{\bw\in\calK_{(r_1)}}\max_{\substack{I\subset[m]\\|I|\le r_0m}}(r_0m)^{-1}\sum_{i\in I}|\ba_i^\top\bw|^2\right)^{1/2}\\
         & \le C_4r_1 \left(\frac{\omega(\calX)}{\sqrt{m}}+\sqrt{r_0\log(r_0^{-1})}\right)\left(\frac{\omega(\calK_{(r_1)})/r_1}{\sqrt{m}}+\sqrt{r_0\log(r_0^{-1})}\right)\\&\le \frac{C_5\omega(\calX)\omega(\calK_{(r_1)})}{\sqrt{m}}
     \end{align*}   
 where in the last line, the first inequality holds with the promised probability, and the second inequality holds due to $r_0=\frac{C_3\omega^2(\calX)}{m(L\delta)^2}$ and $L\delta \gtrsim \log^{1/2}(\frac{m}{\omega^2(\calX)})$. 
\end{proof}
\begin{proof}
    [The Proof of Proposition \ref{factqpe}] Under $r_1=c_1\delta$, we prove  Proposition \ref{factqpe} by showing $\hat{T}^{(1)}\le \frac{r_1}{2}=\frac{c_1\delta}{2}$ and $\hat{T}^{(2)}\le \frac{r_1}{2}=\frac{c_1\delta}{2}$, which are sufficient because of (\ref{T12hat}). The constants in this proof can depend on $c_1.$

\paragraph{Show $\hat{T}_1\le \frac{c_1\delta}{2}$:} We first write 
$$\hat{T}^{(1)} = \sup_{\bu\in\overline{\calX}}\sup_{\bw\in r_1^{-1}\calK_{(r_1)}}\Big|\frac{1}{m}\big\langle\calQ_{\delta}(\bA\bu-\btau)-\bA\bu,\bA\bw\big\rangle\Big|.$$
We set $\rho = c_2\delta$ with small enough $c_2$, $\calA= \overline{\calX}$ and $\calC= \frac{1}{r_1}\calK_{(r_1)}$ in Lemma \ref{lem:qpe1} and obtain that, if 
\begin{align}
    m\ge C_3\left[\scrH(\overline{\calX},c_2\delta)+\frac{\omega^2(\calK_{(r_1)})}{r_1^2}+\frac{\omega^2(\overline{\calX}_{(c_2\delta)})}{\delta^2}\right]\label{addd11}
\end{align}
then with the promised probability we have  $\hat{T}^{(1)} \le \frac{c_0\delta}{2}.$ Note that $\overline{\calX}\subset 2\calX\subset  2\calK$, hence (\ref{addd11}) can be implied by $m\gtrsim \scrH(\calX,\frac{c_2\delta}{2})+\frac{\omega^2(\calK_{(\delta)})}{\delta^2}$, which is assumed in the sample complexity (\ref{pro6sam}).   

\paragraph{Show $\hat{T}_2\le \frac{c_1\delta}{2}$:} Next, by Lemma \ref{lemsatu} and (\ref{pro6sam}) that gives $m\ge  \frac{C_5\omega(\calX)\omega(\calK_{(r_1)})}{r_1\delta}$ with some large enough $C_5$ depending on $c_0$, we obtain $\hat{T}^{(2)}\le \frac{c_0r_1\delta}{2}$ with the promised probability. The proof is complete. 
   \end{proof}
 
\section{Technical Lemmas}\label{app:lemma}
 \begin{lem}
     \label{lem:L1L2}
     For   random vectors $\ba_i$ satisfying Assumption \ref{assum1}, there exists some absolute constant $L_0>0$ such that 
     $$\mathbbm{P}\Big(|\ba_i^\top\bu|\ge \frac{1}{2}\Big)\ge 2L_0,~~\mathbbm{E}|\ba_i^\top\bu|\ge L_0,\quad \forall \bu\in \mathbb{S}^{n-1}.$$
 \end{lem}
 \begin{proof}
     We give a proof for this simple fact. Note that we only need to show $\mathbbm{P}(|\ba_i^\top\bu|\ge 1/2)\ge 2L_0$ since this implies $$\mathbbm{E}|\ba_i^\top\bu|\ge \mathbbm{E}[|\ba_i^\top\bu|\mathbbm{1}(|\ba_i^\top\bu|\ge1/2)]\ge \frac{\mathbbm{P}(|\ba_i^\top\bu|\ge 1/2)}{2}\ge L_0.$$ 
     For $\ba_i$ satisfying Assumption \ref{assum1}, we have (\ref{eq:d1bcs_moment}) holds for some absolute constant $\bar{c}_1$. Hence by Paley-Zygmund inequality (e.g., Lemma 7.16 in \cite{Foucart2013AMI}), 
     \begin{align*}
         \mathbbm{P}\big(|\ba_i^\top\bu|\ge 1/2\big)= \mathbbm{P}\big(|\ba_i^\top\bu|^2\ge 1/4\big)\ge \frac{(3/4)^2}{\mathbbm{E}|\ba_i^\top\bu|^4} \ge \frac{9}{256(\bar{c}_1)^4}. 
     \end{align*}
     Hence, we  set $L_0 := \frac{9}{512 (\bar{c}_1)^4}$ to yield the desired claim. 
 \end{proof}
\begin{lem}[Estimates on Gaussian width and Metric Entropy: (Effective) Sparsity, Low-rankness] \label{lem:gwcoveres} \rev{There exist absolute constants $C_1,C_2,C_3,C_4$ such that, for any $\epsilon\in(0,1)$, $k\in[n]$ and $\bar{r}\in [n_1\wedge n_2]$, we have}  
\begin{gather} \label{sparsegw}
    \omega^2(\Sigma^n_k\cap \mathbb{B}_2^n) \le \omega^2(\sqrt{k}\mathbb{B}_1^n\cap  \mathbb{B}_2^n)\le C_1 k \log\frac{en}{k},\\ \label{sparseentropy}
    \scrH(\Sigma^n_k\cap \mathbb{B}_2^n,\epsilon)\le k\log\Big(\frac{9n}{\epsilon k}\Big), \\  
    \scrH(\sqrt{k}\mathbb{B}_1^n\cap  \mathbb{B}_2^n,\epsilon)\le \frac{C_2k\log(en/k)}{\epsilon^2},\label{l1sparseentropy}\\
    \omega^2(M^{n_1,n_2}_{\bar{r}}\cap \{\bM\in :\|\bU\|_F\le 1\}) \le C_3\bar{r}(n_1+n_2),\label{lrgw}\\\label{lrcover}
    \scrH(M^{n_1,n_2}_{\bar{r}}\cap \{\bM\in :\|\bU\|_F\le 1\},\epsilon)\le C_4 \bar{r}(n_1+n_2)\log\Big(\frac{3}{\epsilon}\Big).
\end{gather} 
\end{lem} 
\begin{proof}
  For (effective) sparsity, see \cite[Equations (II.2), (III.3)]{plan2012robust} for (\ref{sparsegw}), and see \cite[Lemma 3.3]{plan2013one} for (\ref{sparseentropy}). Then (\ref{l1sparseentropy}) follows from $\omega^2(\sqrt{k}\mathbb{B}_1^n\cap  \mathbb{B}_2^n)\le C_1 k \log\frac{en}{k}$ and Sudakov's inequality \cite[Theorem 8.1.13]{vershynin2018high}. For low-rankness, see \cite[Lemma 3.1]{candes2011tight} for (\ref{lrcover}), and then (\ref{lrgw}) follows from Dudley's inequality \cite[Theorem 8.1.10]{vershynin2018high}. 
\end{proof}
\begin{lem}[Estimates on Gaussian width and metric entropy: $\ell_q$ sparsity]
\label{lem:lq-width-entropy}
\rev{Let $q\in(0,1)$ and $1\le k\le n$. Define
$\calK := k^{\frac1q-\frac12}\mathbb B_q^n,~
    \calX := \calK\cap\mathbb B_2^n.$  For $0<r\le 1$, define  $
    \calK_{(r)} := (\calK-\calK)\cap r\mathbb B_2^n.$ 
Then there exists an absolute constant $C$  and a constant $C_q>0$ depending only on $q$} such that
\begin{gather}
 \scrH(\calX,r)
    \le
    \min\left\{
        n\log\frac{3}{r},
        \,
        C_q\, k r^{-\frac{2q}{2-q}}
        \log\frac{C_q n}{r}
    \right\},\label{lqcoverbound}
\\
      \omega^2(\calK_{(r)})
    \le
    \min\left\{Cr^2n,C_q\, k\log(en)\,
    r^{\frac{4(1-q)}{2-q}} \right\}. \label{lqgwbound}
\end{gather}
\end{lem}

\begin{proof} 
Similar estimates can be found in \cite{mendelson2008uniform,raskutti2011tit} for instance. Here, we provide a self-contained proof.  Throughout the proof, $C_q>0$ denotes a constant depending only on $q$, whose value may change from line to line. We first prove (\ref{lqcoverbound}). The first branch $n\log\frac{3}{r}$ comes from 
$
    \scrH(\calX,r)
    \le
    \scrH(\mathbb B_2^n,r)
    \le
    n\log\frac{3}{r}$. It remains to prove the second branch. Let 
$R_q := k^{\frac1q-\frac12}$
    and $
    \gamma_q := \frac{2q}{2-q}.$ 
Take any $\bu\in\calX$. Let $(u_j^*)_{j=1}^n$ denote the nonincreasing
rearrangement of $(|u_j|)_{j=1}^n$. Since $\bu\in\calK$, we have 
$\sum_{j=1}^n |u_j|^q \le R_q^q = k^{1-\frac q2}.$
Therefore, for every $j\in[n]$, we have
$
    j (u_j^*)^q
    \le
    \sum_{\ell=1}^j (u_\ell^*)^q
    \le
    k^{1-\frac q2},$ 
and hence
$
    u_j^* \le R_q j^{-1/q}.$ Let $u^{(s)}$ denote the best $s$-term truncation of $\bu$, obtained by keeping
the largest $s$ coordinates of $\bu$ in magnitude and setting the rest to zero. 
Then, for $1\le s\le n$, it holds that
$
    \|\bu-\bu^{(s)}\|_2^2
    =
    \sum_{j>s} (u_j^*)^2
    \le
    R_q^2\sum_{j>s} j^{-2/q}.$ 
Since $q<2$, we have $2/q>1$, and thus
$
    \sum_{j>s} j^{-2/q}
    \le
    C_q s^{1-\frac2q}.
$ 
Consequently,
\begin{align}
     \|\bu-\bu^{(s)}\|_2
    \le
    C_q R_q s^{\frac12-\frac1q}
    =
    C_q \left(\frac{k}{s}\right)^{\frac1q-\frac12}.\label{use11}
\end{align}
We now choose
$
    s
    :=
    \min\left\{
        n,\,
        \left\lceil A_q k r^{-\gamma_q}\right\rceil
    \right\},$ 
where $A_q>0$ is a sufficiently large constant depending only on $q$.
With this choice, (\ref{use11}) implies
$
    \|\bu-\bu^{(s)}\|_2\le \frac r2.$
Indeed, if $s=n$, this is trivial; otherwise, the choice
$s\gtrsim_q k r^{-\gamma_q}$ gives
$
    \big(\frac{k}{s}\big)^{\frac1q-\frac12}
    \lesssim_q
    r$
in light of
$
    \gamma_q\big(\frac1q-\frac12\big)=1.$ By (\ref{sparseentropy}) in Lemma \ref{lem:gwcoveres}, we can construct an $(r/2)$-net of $\Sigma^n_s\cap \mathbb{B}_2^n$, denoted by
$\Lambda_s$, such that $\log|\Lambda_s|\le s\log\frac{18n}{rk}$. 
Since $\bu^{(s)}\in\widetilde U_s$, there exists $\bv\in\Lambda_s$ such that
$
    \|\bu^{(s)}-\bv\|_2\le \frac r2.$ 
Combining this with $\|\bu-\bu^{(s)}\|_2\le r/2$ gives $\|\bu-\bv\|_2\le r.$  Therefore $\Lambda_s$ is an $r$-net for $\calX$. Since
$
    s\le C_q k r^{-\gamma_q},
$ 
we conclude that
$
    \scrH(\calX,r)
    \le
    C_q k r^{-\gamma_q}\log\frac{C_q n}{r}
    =
    C_q k r^{-\frac{2q}{2-q}}\log\frac{C_q n}{r}.$ Together with the bound $n\log\frac{3}{r}$, this proves (\ref{lqcoverbound}). 

We now prove (\ref{lqgwbound}). The first branch follows from
$\omega^2(\calK_{(r)})\le \omega^2(r\mathbb{B}_2^n)=r^2 \omega^2(\mathbb{B}_2^n)\lesssim r^2n$. It remains to show the second branch. Take any $\bv\in\calK_{(r)}$. Then
$\bv=\bu-\bu'$ for some $\bu,\bu'\in\calK$, and $\|\bv\|_2\le r$. Since $q\in(0,1)$, we have 
$
    \|\bv\|_q^q
    =
    \sum_{j=1}^n |u_j-u_j'|^q
    \le
    \sum_{j=1}^n |u_j|^q
    +
    \sum_{j=1}^n |u_j'|^q
    \le
    2k^{1-\frac q2}.$ We next interpolate between $\ell_q$ and $\ell_2$. By H\"older's inequality, 
$\|\bv\|_1=\sum_{j=1}^n
    \big(|v_j|^q\big)^{\frac1{2-q}}
    \big(|v_j|^2\big)^{\frac{1-q}{2-q}}
    \le
    \left(\sum_{j=1}^n |v_j|^q\right)^{\frac1{2-q}}
    \left(\sum_{j=1}^n |v_j|^2\right)^{\frac{1-q}{2-q}},$ which then gives an $\ell_1$ estimate
\begin{align}
    \|\bv\|_1
    \le
    \left(2k^{1-\frac q2}\right)^{\frac1{2-q}}
    r^{\frac{2(1-q)}{2-q}}
    \le
    C_q \sqrt{k}\,
    r^{\frac{2(1-q)}{2-q}} .\label{l1lqestimate}
\end{align}
Let $\bg\sim \calN(0,\bI_n)$. By (\ref{l1lqestimate}), we obtain 
$\omega(\calK_{(r)})
    =
    \mathbb E\sup_{\bv\in\calK_{(r)}}\langle \bg,\bv\rangle
    \le
    \mathbb E\|\bg\|_\infty
    \sup_{\bv\in\calK_{(r)}}\|\bv\|_1\le C_q\sqrt{k}r^{\frac{2(1-q)}{2-q}} \mathbb E\|\bg\|_\infty.$ 
The standard Gaussian maximal inequality gives
$
    \mathbb E\|\bg\|_\infty
    \le
    C\sqrt{\log(en)}.$ 
Therefore $
    \omega(\calK_{(r)})
    \le
    C_q \sqrt{k\log(en)}\,
    r^{\frac{2(1-q)}{2-q}}$,
which implies the desired bound.
The proof is complete.
\end{proof}

 \begin{lem}
    [Bernstein's inequality, e.g., Theorem 2.10 in \cite{13concen}] \label{lem:chernoff} 
    Let $X_1,\cdots,X_n$ be independent real-valued random variables. Assume that there exist positive numbers $v,c$ such that $$\sum_{i=1}^n\mathbbm{E}|X_i|^q \le \frac{q!}{2}vc^{q-2}\quad \text{for all integers }q\ge 2.$$ 
    Then for any $t>0$ we have 
    \begin{align*}
        \mathbbm{P}\left(\Big|\sum_{i=1}^n\big(X_i-\mathbbm{E}(X_i)\big)\Big| \ge \sqrt{2vt}+ct\right) \le 2\exp(-t). 
    \end{align*}
\end{lem}
\begin{lem}
    [See, e.g., \cite{dirksen2021non}] \label{lem:max_ell_sum} Let $\ba_1,...,\ba_m$ be independent   random vectors in $\mathbb{R}^n$ satisfying $\mathbbm{E}(\ba_i\ba_i^\top)=\bI_n$ and  $\max_i\|\ba_i\|_{\psi_2}\leq L$. For some   given $\calW\subset \mathbbm{R}^n$ and  $1\leq \ell\leq m$, there exist  constants $C_1,c_2$ depending only on $L$ such that   the event 
    \begin{equation*}     \sup_{\bw\in\calW}\max_{\substack{I\subset [m]\\|I|\leq \ell}}\Big(\frac{1}{\ell}\sum_{i\in I}|\langle\ba_i,\bw\rangle|^2\Big)^{1/2}\leq C_1\left(\frac{\omega(\calW)}{\sqrt{\ell}}+\rad(\calW)\log^{1/2}\Big(\frac{em}{\ell}\Big)\right)
    \end{equation*}
    holds with probability at least $1-2\exp(-c_2\ell\log(\frac{em}{\ell}))$, where $\rad(\calW)=\sup_{\bw\in\calW}\|\bw\|_2$. 
\end{lem}
 
\begin{lem}
    [Corollary 8.3 in \cite{plan2017high}] \label{planlem} Let $\calK$ be a star-shaped set, then for any $\bx\in \calK$ and $\bu\in \mathbb{R}^n$ we have 
    \begin{align*}
        \|\calP_{\calK}(\bu)-\bx\|_2 \le \max\left\{\phi,\frac{2}{\phi}\|\bu-\bx\|_{\calK^\circ_{(\phi)}}\right\},\quad\forall \phi>0.
    \end{align*}
\end{lem}
\begin{lem}
    [A Useful Parameterization, Lemma H.4 in \cite{chen2024one}] \label{lem:parameterization}Given  two distinct points $\bu,\bv\in\mathbb{R}^n$, we let $\bbeta_1= \frac{\bu-\bv}{\|\bu-\bv\|_2}$ and can further find $\bbeta_2$ satisfying $\|\bbeta_2\|_2=1$ and $\langle \bbeta_1,\bbeta_2\rangle=0$, such that
        \begin{align*}
            \bu = u_1\bbeta_1 + u_2 \bbeta_2~~\text{and}~~\bv = v_1 \bbeta_1 + u_2 \bbeta_2
        \end{align*}
    hold for some $(u_1,u_2,v_1)$ satisfying $u_1>v_1~\text{ and }~u_2\ge 0.$
  We have $v_1=-u_1$ when $\|\bu\|_2=\|\bv\|_2$ holds. 
\end{lem}
\begin{lem}[See, e.g., Exercise 8.6.5 in \cite{vershynin2018high}]
    \label{lem:tala}
  Let $(R_{\bu})_{\bu\in\calW}$ be a random process (that is not necessarily zero-mean) on a subset $\calW\subset \mathbb{R}^n$. Assume that $R_{0}=0$, and  $\|R_{\bu}-R_{\bv}\|_{\psi_2}\leq K\|\bu-\bv\|_2$ holds for all $\bu,\bv\in\calW\cup\{0\}$. Then, for every $t\geq 0$, the event 
    \begin{equation*}
        \sup_{\bu\in \calW}\big|R_{\bu}\big|\leq CK \big(\omega(\calW)+t\cdot \rad(\calW)\big)
    \end{equation*}
    with probability at least $1-2\exp(-t^2)$,  where $\rad(\calW)=\sup_{\bw\in\calW}\|\bw\|_2$.  
\end{lem}
\begin{lem}
    [Concentration of Product Process, \cite{mendelson2016upper}; see also \cite{genzel2023unified}] \label{lem:product_process}Let $\{g_{\bu}\}_{\bu\in \calU}$ and $\{h_{\bv}\}_{\mathbf{b}\in\calV}$ be stochastic processes indexed by two sets    $\calU\subset \mathbb{R}^{p}$ and $\calV\subset \mathbb{R}^q$, both defined on a common probability space $(\Omega,A,\mathbbm{P})$. We assume that there exist $K_{\calU},K_{\calV},r_{\calU},r_{\calV}\geq 0$ such that 
    \begin{align*}
       & \|g_{\bu}-g_{\bu'}\|_{\psi_2}\leq K_{\calU}\|\bu-\bu'\|_2,~\|g_{\bu}\|_{\psi_2} \leq r_{\calU},~\forall\bu,\bu'\in \calU;\\
       & \|h_{\bv}-h_{\bv'}\|_{\psi_2} \leq {K}_{\calV}\|\bv-\bv'\|_2,~\|h_{\bv}\|_{\psi_2}\leq r_{\calV},~\forall\bv,\bv'\in \calV.
    \end{align*}
Suppose that $\ba_1,...,\ba_m$ are independent copies of a random variable $\ba\sim \mathbbm{P}$, then for every $t\geq 1$ the event\begin{equation}
   \begin{aligned}\nonumber
        &\sup_{\bu\in\calU}\sup_{\bv\in\calV} ~ \Big|\frac{1}{m}\sum_{i=1}^m g_{\bu}(\ba_i)h_{\bv}(\ba_i)-\mathbbm{E}[g_{\bu}(\ba_i)h_{\bv}(\ba_i)]\Big|\\
        &\leq C\left(\frac{(K_{\calU}\cdot\omega(\calU)+t\cdot r_{\calU})\cdot (K_{\calV}\cdot \omega(\calV)+t\cdot r_{\calV})}{m}+\frac{r_{\calU}\cdot K_{\calV}\cdot \omega(\calV)+r_{\calV}\cdot K_{\calU}\cdot \omega(\calU)+t\cdot r_{\calU}r_{\calV}}{\sqrt{m}}\right)
   \end{aligned}
\end{equation}
holds with probability at least $1-2\exp(-ct^2)$. 
\end{lem}
 \end{appendix}
\end{document}